\documentclass[12pt,notitlepage,a4paper]{article}
\usepackage{amsfonts}
\usepackage{float}
\usepackage{subfigure}
\usepackage[left=0.5in,right=0.5in,top=0.5in,bottom=0.75in]{geometry}
\usepackage{graphics}
\usepackage{enumerate}
\usepackage{hyperref}
\usepackage{tikz}
\usetikzlibrary{decorations.markings,quantikz}
\usepackage{ulem}
\usepackage{mathtools}
\mathtoolsset{showonlyrefs}
\usepackage{amsmath}
\usepackage{graphicx}
\usepackage{enumerate}
\usepackage{amssymb}
\usepackage{authblk}
\usepackage{color}
\usepackage{amsthm}
\usepackage[framemethod=tikz]{mdframed}
\usepackage[noend]{algpseudocode}

\newcommand{\appref}[1]{\hyperref[#1]{{Appendix~\ref*{#1}}}}
\newcommand{\be}{\begin{eqnarray} \begin{aligned}}
\newcommand{\ee}{\end{aligned} \end{eqnarray} }
\newcommand{\benn}{\begin{eqnarray*} \begin{aligned}}
\newcommand{\eenn}{\end{aligned} \end{eqnarray*}}
\newcommand*{\textfrac}[2]{{{#1}/{#2}}}

\newcommand*{\cA}{\mathcal{A}} 
\newcommand*{\cB}{\mathcal{B}}
\newcommand*{\cC}{\mathcal{C}}

\newcommand*{\cF}{\mathcal{F}}

\newcommand*{\cH}{\mathcal{H}}

\newcommand{\CC}{\mathbb{C}}
\newcommand*{\cL}{\mathcal{L}}

\newcommand*{\cS}{\mathcal{S}}

\newcommand{\ZZ}{\mathbb{Z}}

\newcommand{\calL}{{\cal L }}
\newcommand{\calH}{{\cal H }}
\newcommand{\ra}{\rangle}
\newcommand{\la}{\langle}
\newcommand{\calF}{{\cal F }}

\newcommand*{\tr}{\mathop{\mathrm{tr}}\nolimits}

\newcommand{\bc}{\begin{center}}
\newcommand{\ec}{\end{center}}

\newtheorem{theorem}{Theorem}[section]
\newtheorem{lemma}[theorem]{Lemma}

\newtheorem{corollary}[theorem]{Corollary}

\usepackage{amsfonts}

\def\01{\{0,1\}}

\newcommand{\ketbra}[2]{|#1\rangle\langle#2|}
\newcommand{\multgate}{\operatorname{CX}}

\newcommand*{\sysA}{\mathsf{A}}
\newcommand*{\sysB}{\mathsf{B}}
\newcommand*{\sysC}{\mathsf{C}}
\newcommand*{\der}{\mathsf{d}}% derivative operator for differentila form

\newcommand*{\idG}{\mathsf{e}} % identity element in G
\begin{document}

\renewcommand\Affilfont{\fontsize{10}{14}\itshape}

\author[1]{Sergey Bravyi}
\author[2]{Isaac Kim}
\author[3]{Alexander Kliesch}
\author[3]{Robert Koenig\footnote{Corresponding author: robert.koenig@tum.de}}
\affil[1]{IBM Quantum, IBM T.J. Watson Research Center, Yorktown Heights, USA}
\affil[2]{Department of Computer Science, University of California, Davis, USA}
\affil[3]{Munich Center for Quantum Science and Technology \& Zentrum Mathematik, Technical University of Munich, Germany}
\date{}
\title{Adaptive constant-depth circuits for\\
manipulating non-abelian anyons}

\maketitle

\begin{abstract}
We consider Kitaev's quantum double model based on a finite group $G$
and describe quantum circuits for (a) preparation of the ground state,
(b) creation of anyon pairs separated by an arbitrary distance, and
(c) non-destructive topological charge measurement. We show that for any solvable group $G$
all above tasks can be realized  by constant-depth adaptive circuits
with geometrically local unitary gates and mid-circuit measurements.
Each gate may be chosen adaptively depending on previous measurement outcomes.
Constant-depth circuits are well suited for implementation on a noisy 
hardware since it may be possible to execute the entire circuit
within the qubit coherence time.
Thus our results could facilitate an experimental study of exotic phases of matter with a non-abelian 
particle statistics.
We also show that adaptiveness is essential for our circuit construction. Namely,
task (b) cannot be realized by non-adaptive constant-depth local circuits
for any non-abelian group $G$. This is in a sharp contrast with abelian anyons which can be
created and moved over an arbitrary distance by a depth-$1$ circuit
composed of generalized Pauli gates. 
\end{abstract}

\tableofcontents

\section{Introduction}

Kitaev's quantum double model~\cite{kitaev} is a paradigmatic example of an exactly solvable model exhibiting topological order. For a finite group~$G$, the construction provides a spin-lattice model whose  low-energy physics is described by the non-chiral topological quantum field theory associated with Drinfeld's quantum double~$D(G)$. As such, it provides an ideal testbed for the study of topological quantum computation, and, in particular, the differences between non-abelian and abelian anyons.

Its most well-known manifestation -- the toric~\cite{kitaev} or surface code~\cite{bravyi1998quantum} based on the abelian group~$G=\mathbb{Z}_2$ --  is currently among primary contenders when it comes to experimental realizations, see e.g.,~\cite{satzingeretal21}. The preeminence of this abelian model can be attributed, in part, to the fact that it falls into a well-studied family of quantum error-correcting codes: It is a CSS-stabilizer code with local generators. As such, it permits the application of a plethora well-studied techniques for fault-tolerant quantum computing, including, in particular, efficient syndrome extraction and decoding as well as magic state distillation. Futhermore, its constituent physical degrees of freedom are  qubits since $|\mathbb{Z}_2|=2$, making this model especially amenable to experimental efforts.

Despite the promise of the surface codes for topological quantum computation, models with non-abelian anyons may still offer distinct advantages. In particular, unlike with abelian anyons, universal computation can be achieved by purely topological means when the associated braid group representation is sufficiently rich. The required  operations are
\begin{enumerate}[(i)]
\item\label{it:statecreation}
  the creation/initialization of a ground (``vacuum'') state having no excitations
\item
  the braiding or exchange of two particles
\item
  the fusion of two particles: this amounts to bringing two anyons to the same site and performing a measurement of the joint topological charge
\item\label{it:specificparticleantiparticle}
  the creation of specific particle-antiparticle pairs.
\end{enumerate}
Mochon~\cite{Mochon04} has shown that for anyons described by~$D(G)$ with a solvable but not nilpotent group~$G$, these operations suffice to realize universal computation\footnote{Mochon additionally assumes  that a supply of calibrated electric/magnetic charge ancillas is available. This is implied by operation~\eqref{it:specificparticleantiparticle}, where in contrast to~\cite{Mochon04}, we assume that the specific pair  created is not random.}. This applies, in particular, to the group~$S_3$. The power of $D(S_3)$-anyons to realize universal computation has prompted the construction of explicit protocols implementing these operations, including experimental proposals for their realization with atoms trapped in an optical lattice, see~\cite{AguadoetalPRL08,Brennen2009}. Unlike the corresponding operations for the surface code, however, the protocols proposed in~\cite{AguadoetalPRL08,Brennen2009} require quantum circuits whose depth scales with the system size, respectively the spatial separation  between anyons. It is natural to ask if this increased complexity is in fact necessary when dealing with non-abelian anyons.

To see how non-abelian and abelian anyons may differ in their (circuit) complexity of their associated operations, consider the problem~\eqref{it:statecreation} of initial state preparation. There is no significant difference here when considering unitary  circuits: A circuit consisting of nearest-neighbor two-qubit unitary gates (a unitary local circuit)  preparing the ground state of the surface code from a product state has been proposed in~\cite{dennisetal02}. To prepare a ground state on an~$L\times L$-lattice, the circuit has depth of order~$O(L^2)$. An analogous result applies to the preparation of a ground state of any quantum double model, see~\cite{Brennen2009,aguado11}. An $O(L)$-depth circuit for the surface code 
was given in~\cite{Higgott2021optimallocalunitary}, and linear depth circuits for the surface code and certain Levin-Wen models were also obtained in~\cite{pollmann2021}.  In~\cite{BravyiHastingsVerstraete06}, it is shown that a circuit depth of order~$\Omega(L)$ is necessary for any local unitary circuit. This lower bound applies to any topologically ordered system, irrespective of whether or not the corresponding anyons are abelian.
With non-local two-qudit gates, the circuit depth can be lowered to $O(\log L)$~\cite{AguadoVidal08} for any quantum double model; this matches the lower bound of~\cite{aharonov2018quantum} for $G=\mathbb{Z}_2$.

A more efficient way of preparing states with~$D(\mathbb{Z}_2)$-topological order involves the use of measurements. Indeed,  as with any stabilizer code, a code state (i.e., a ground state in the case of the surface code) is obtained from an arbitrary initial state by 
\begin{enumerate}[(a)] 
\item\label{it:syndromeextraction}
  measuring a complete set of  commuting stabilizer generators, obtaining a syndrome~$s$
\item\label{it:correctionoperation}
  applying an associated unitary correction operation~$C(s)$ to the post-measurement state. 
\end{enumerate}
We note that such an approach could in principle also be taken for any quantum double model. In particular, the required measurements in step~\eqref{it:syndromeextraction} can also be chosen to be local  as in any commuting local projector code. The distinguishing feature of the surface code is the simplicity of step~\eqref{it:correctionoperation}: The required correction~$C(s)$ is a Pauli operator which can  be determined from~$s$ by an efficient classical decoding algorithm. In particular, application of~$C(s)$ is easily accomplished by a depth-$1$ unitary circuit consisting of single-qubit Pauli gates only. In terms of the excitations (anyons) of the surface code, this process can  be summarized as follows. The syndrome~$s$ reveals the locations of the excitations (violations of the stabilizer constraints). The applied correction~$C(s)$ (determined, e.g., by minimal matching as in~\cite{dennisetal02}) is a product of ribbon-operators where each ribbon connects two excitations. In the surface code, such ribbon operators are tensor products of Pauli-$X$, Pauli-$Y$ or Pauli-$Z$-operators along the ribbon, see~\cite{kitaev} or~\cite{dennisetal02}. The simplicity of this preparation scheme is a consequence of the simple form of these ribbon operators for the case~$G=\mathbb{Z}_2$.

For a general finite group~$G$, each anyon is labeled by an irreducible representation~$\rho$ of the Drinfeld double~$D(G)$ of~$G$. For each (open) ribbon~$\xi$ on the lattice, there are associated ribbon operators~$\{F_\xi^{\rho;\alpha}\}_\alpha$ which create a corresponding particle/antiparticle pair at the two endpoints of~$\xi$. Here~$\alpha$ controls degrees of freedom localized at these endpoints. We refer to Section~\ref{sec:quantumdoublemodel} for details. We note that
a ribbon operator~$F_\xi^{\rho;\alpha}$ can also be understood as realizing a process where a pair is created locally at one endpoint, and one of the anyons is moved to the other endpoint of the ribbon~$\xi$. In particular, for non-abelian models, such operations can have a non-trivial logical action on encoded quantum information
in the presence of other (preexisting) anyons because of non-trivial braiding relations. The idea that -- since the quantum double model is an error-correcting code -- such logical operations require a circuit depth which scales with the length of the ribbon is folklore, see e.g.,~\cite[Section~8]{haahpreskill} for a related argument. Here we give a rigorous proof of this fact for non-abelian quantum double models, see Theorem~\ref{thm:nogotheorem} below. It gives the following statement, where we use the term {\em extensive} to refer to a scaling linear in the system size (or, more precisely, the code distance as given, e.g., by the separation of holes in the surface). 
\begin{corollary}\label{cor:nogoresult}
  For any non-abelian group~$G$, there are ribbon operators~$F_\xi^{\rho;\alpha}$ whose implementation by a local unitary circuit requires an extensive circuit depth.
\end{corollary}
Corollary~\ref{cor:nogoresult} is a no-go result ruling out the possibility of
applying ribbon operators using constant-depth circuits in non-abelian anyon models. This is a novel operational separation between  non-abelian and abelian topological order.

It is tempting to think that the conclusion of  Corollary~\ref{cor:nogoresult} extends to the problem~\eqref{it:statecreation}, preventing initialization by constant-depth quantum circuits because of the difficulty of implementing the ribbon operators involved in the correction step~\eqref{it:correctionoperation}. In fact, step~\eqref{it:correctionoperation} involves the additional challenges of classical  decoding the corresponding code (i.e., finding the right product of ribbon operators to apply), a problem which is non-trivial~\cite{wottonetal14}.  Furthermore,  in a non-abelian model, a single application of a ribbon operator is typically not sufficient  to remove a pair of excitations since anyon pairs do not need to fuse to the vacuum. Nevertheless, we find that for certain non-abelian quantum double models (including the case~$G=S_3$), efficient initialization is possible by using an approach  different from~\eqref{it:syndromeextraction}--\eqref{it:correctionoperation}. In fact, all operations~\eqref{it:statecreation}--\eqref{it:specificparticleantiparticle} can be realized by what we call {\em constant-depth local adaptive circuits}.  This notion captures the operations involved in the procedures used for the surface code: We allow a constant number of the alternating layers of 
\begin{enumerate}[(a)]
\item
constant-depth quantum circuits  that may use auxiliary qubits and consist of local unitaries and single-qubit measurements and
\item efficient, possibly non-local classical computations (based on measurement results).
\end{enumerate}
Adaptivity refers to the fact that quantum operations may be classically controlled by measurement outcomes and computational results obtained in previous layers. In other words, all involved quantum operations are local and realized by constant-depth circuits; this is supplemented by efficient non-local classical processing. 

In more detail, our result applies to any solvable group~$G$.
Recall that a group~$G$ is solvable if it can be mapped to the trivial group by iteratively factoring out normal  abelian subgroups. We show the following:
\begin{theorem}(informal) Suppose $G$ is a solvable group.  Then there is a constant-depth local adaptive circuit for each of the following tasks:\label{thm:informalmain}
  \begin{enumerate}[(i)]
  \item
    Creation of a ground state from a product initial state.
  \item
     For any anyon label~$\rho$ and local action specified by~$\alpha$, application of the ribbon operators $F_\xi^{\rho;\alpha}$ associated with an open ribbon~$\xi$ of possibly extensive length.
      Here the circuit has support (i.e., acts non-trivially) only on qudits along the ribbon~$\xi$.
     
  \item
    Execution of the topological charge measurement of any region encircled by a closed ribbon~$\sigma$ (possibly of extensive length). This is given by a POVM~$\{K_\sigma^\rho\}_{\rho}$ whose outcomes are anyon labels~$\rho$, and whose POVM elements~$K_\sigma^\rho$  have support on the ribbon~$\sigma$, see Section~\ref{sec:quantumdoublemodel}. The circuit realizing this  measurement
    has support only on qudits along~$\sigma$.
        \end{enumerate}
\end{theorem}
The operations considered in Theorem~\ref{thm:informalmain} can be used to  realize braiding/movement of anyons by (repeated) pair creation and topological charge measurement.

Let us conclude by mentioning prior work as well as some open problems. The power of measurements and adaptive operations for reducing the depth of quantum circuits  has been recognized in the seminal work by H\o{}yer and Spalek~\cite{hoyerspalek} -- some of our constructions are motivated by their fan-out gate. This work has lead to a discovery of constant-depth adaptive circuits realizing
any unitary in the Clifford group~\cite{gottesmanchuang99,jozsa2006introduction}, multiple control Toffoli gate and integer arithmetic circuits~\cite{hoyerspalek,takahashi2016collapse}, and Quantum Fourier Transform~\cite{hoyerspalek}.
These are important subroutines employed by numerous quantum algorithms. 
Moreover, it is known that the entire quantum part of Shor's factoring algorithm~\cite{shor1999polynomial} 
can be parallelized to a constant depth using adaptive circuits~\cite{browne2010computational}.
Measurements and  adaptive operations also play a central role in synthesizing quantum circuits over fault-tolerant 
gate sets using gate teleportation~\cite{gottesmanchuang99} and repeat-until-success techniques~\cite{paetznick2013repeat}.
In a recent breakthrough work
Liu and Gheorghiu~\cite{liu2021depth} established an efficiently verifiable quantum advantage
for certain classically hard interactive tasks that can be solved by constant depth adaptive quantum circuits.
The apparent power of low-depth adaptive circuits has led Jozsa to conjecture~\cite{jozsa2006introduction} that, in fact,
any polynomial depth quantum computation can be efficiently simulated by 
logarithmic depth adaptive quantum circuits~\cite{jozsa2006introduction}.
However, a more recent work~\cite{coudron2020computations} cast doubt on Jozsa's conjecture by demonstrating that 
the Welded Tree Problem~\cite{childs2003exponential} is provably hard for the logarithmic depth adaptive circuits
(as measured by query complexity) even though this problem admits an efficient quantum algorithm
based on quantum walks. Likewise, it is not known whether the depth of quantum algorithms for simulating 
unitary time evolution of many-body Hamiltonians 
can be significantly reduced  using measurements and adaptive operations.

In a recent work Piroli et al~\cite{cirac2021} examined constant-depth quantum circuits
assisted by LOCC (local operations and classical communication) and equivalence classes of many-body
quantum states convertible to each other by such circuits. It was demonstrated that certain
highly entangled  states exhibiting topological quantum order become trivial in the LOCC-assisted classification
framework. We note however that constant depth local adaptive circuits considered in the present work
are strictly weaker than LOCC-assisted circuits of~\cite{cirac2021} since we only allow $O(1)$ 
rounds of mid-circuits measurements whereas \cite{cirac2021} allows a linear number of rounds.
Closer to the topic of this paper, 
existing non-unitary protocols for preparing states with non-abelian topological order
include those  of~\cite{Brennen2009,aguado11} for creating quantum double states associated with~$S_3$, as well as the work~\cite{pollmann2021} which provides procedures for preparing ground states of the Levin-Wen model~\cite{levin2005string}.  In contrast to these procedures, which require extensive circuit depth,  
Verresen et al.~\cite{verresen2022efficiently} reported adaptive constant-depth circuits preparing the ground state of
the quantum double~$D(S_3)$ and $D(D_4)$, where $D_4$ is the dihedral group.  Subsequently, it was observed in~\cite{2112.01519} that this generalizes to any solvable group. Here we give explicit adaptive constant-depth circuits for ground state preparation for arbitrary solvable groups and additionally discuss the creation and manipulation of excited states.  We note that
constant-depth circuits also figure prominently in~\cite{zhulavasanibarkeshli}, where braiding with non-abelian anyons in the Levin-Wen model is achieved by constant-depth circuits by certain dynamic lattice deformations. Here we follow a different approach and ask for implementations of ``standard'' braiding operations etc.~without changing the underlying lattice.

Several immediate problems remain open. Perhaps the most intriguing one is  whether the operations realized here for a solvable group~$G$ may also be realized with local operations in constant adaptive depth in the case of a non-solvable (non-abelian) group. It is conceivable
that there are fundamental complexity-theoretic obstructions to this similar to those found in~\cite{grierschaeffer20}) related to Barrington's theorem~\cite{barrington}. More generally, one may ask for more general classes of topologically ordered systems where the considered operations (preparation, anyon creation, braiding, and topological charge measurements) can be realized in adaptive constant depth. For the Levin-Wen model, this  amounts to identifying relevant properties of the underlying tensor category.
A preliminary work shows that the ground state of the double semion model~\cite{levin2005string}
can be prepared in constant adaptive depth by measuring syndromes of suitable augmented
stabilizer generators composed from the elementary plaquette and vertex stabilizers.
Finally, it remains to be seen to what extent these procedures can be made robust to noise: Without suitable error correction mechanisms, the non-local operations  realized here will likely be highly susceptible to errors. Because of their
low complexity, they  may nevertheless be useful, e.g., for proof-of-principle demonstrations in the near term.

\subsubsection*{Outline}
In Section~\ref{sec:quantumdoublemodel}, we review the quantum double model.
In Section~\ref{sec:nogoresult}, we establish the no-go result for the unitary realization of ribbon operators in non-abelian anyon models.
We then discuss the design of constant-depth adaptive circuits for solvable groups~$G$. In Section~\ref{sec:statepreparation}, we construct such a circuit preparing a ground state of the quantum double model. In Section~\ref{sec:anyonicribbonimplementation}, we show how to implement anyonic ribbon operators. In Section~\ref{sec:implementingtopologicalchargemeasurements}  we discuss how to realize topological charge measurements.

\section{The quantum double model\label{sec:quantumdoublemodel}}
Here we give only a brief introduction to the quantum double model focusing on the concepts relevant to our work. We refer to~\cite{kitaev} as well as~\cite{BeigiShor,BombinDelgado} for more details.

\subsection{Definition of the quantum double model}
Let $G$~be a finite group with identity element denoted~$\idG\in G$. The quantum double model associated with~$G$ is defined by a planar graph or lattice~$\cL=(V,E)$ with a qudit of dimension~$|G|$ associated with every edge~$e\in E$. We use an orthonormal basis~$\{\ket{g}\}_{g\in G}$ for the Hilbert space~$\mathbb{C}^{|G|}$ of every qudit. We assume that every edge~$e\in E$ is oriented such that one endpoint of~$e$ is designated the head and the other endpoint is designated as the tail. We shall write~$e^+$ and $e^-$ for the head and the tail of~$e$. The definition of the quantum double Hamiltonian involves operators $\{A_s^g\}_{g\in G}$ and operators~$\{B^h_s\}_{h\in G}$ associated with every site~$s$ of the lattice. Here a site is defined as a pair $s=(v,p)$, where $p$~is a plaquette (face) of~$\cL$ and $v\in V$ is vertex adjacent to~$p$. The action of these operators is given
as in Fig.~\ref{fig:siteoperators} (with analogous definitions for non-square lattices). Different orientations of the edges can be accounted for by the following simple rule: Reversal of the direction of an arrow amounts to replacing the corresponding basis vector~$\ket{g}$ by $\ket{g^{-1}}$, i.e., applying inversion in the group.

 \begin{figure}
 \centering
\subfigure[The site-operator $B^h_s$]{\includegraphics[width=0.7\textwidth]{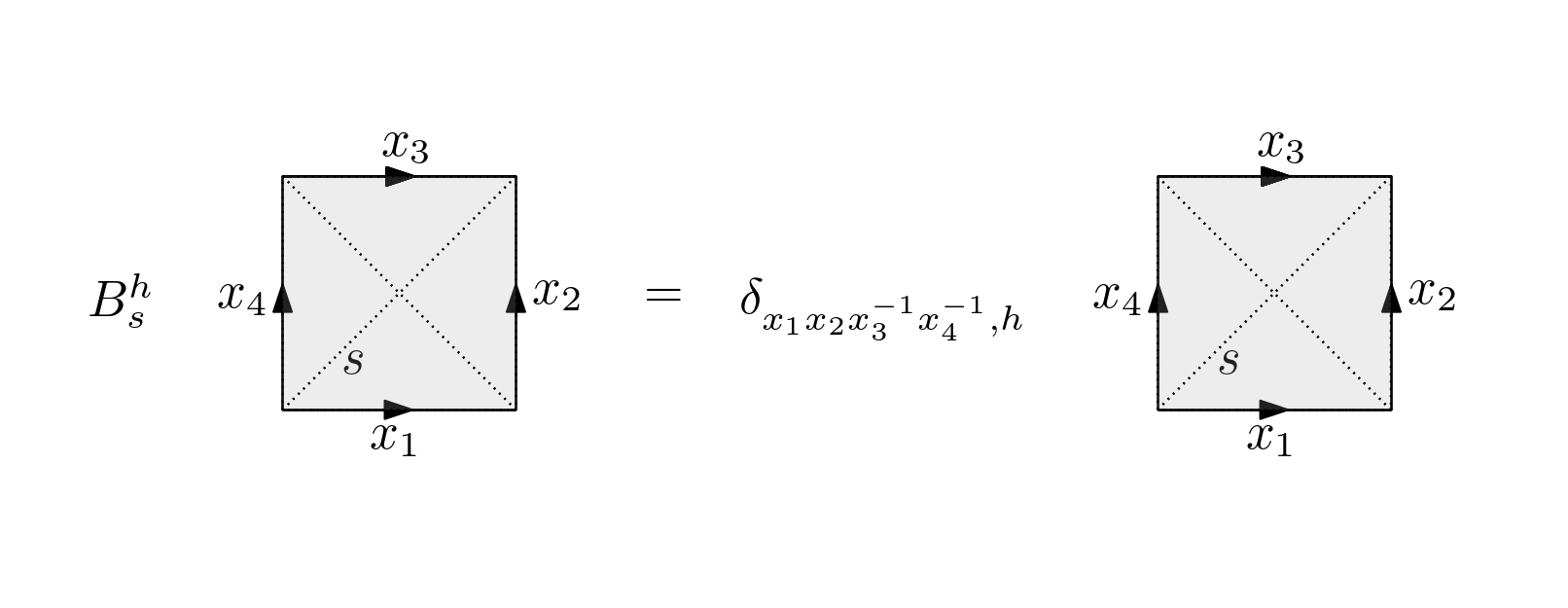}}\qquad\qquad
\subfigure[The site-operator $A^g_s$]{\includegraphics[width=0.7\textwidth]{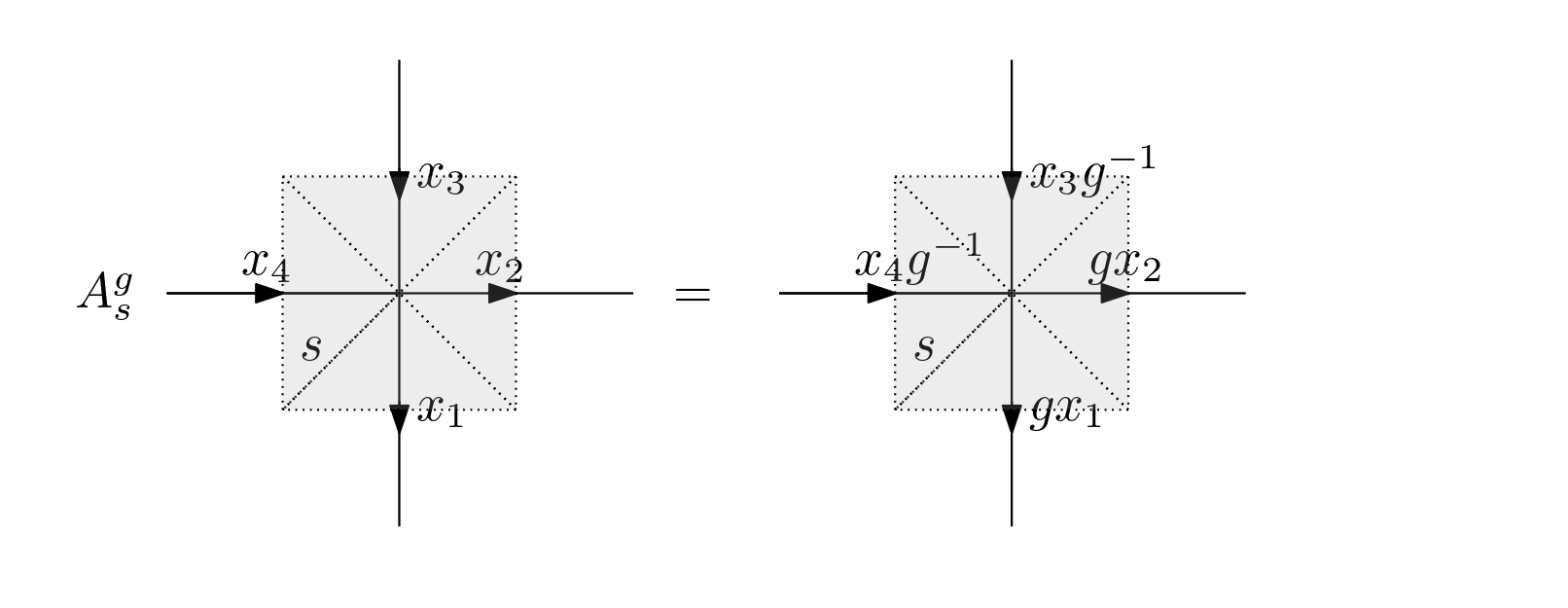}}
\caption{The action of the site-operators $B^h_s$ and $A^g_s$ on computational basis states (denoted~$x$). Here  $g,h\in G$ and $s$ is a site of the lattice.  
Both operators are ribbon operators associated with a closed ribbon (shaded) starting and ending at~$s$. 
\label{fig:siteoperators}}
\end{figure}

For a site $s=(v,p)$ define the site-operators
\begin{align}
  A_s&=\frac{1}{|G|}\sum_{g\in G}A^g_s\\
B_s &= B^\idG_s\ .
\end{align}
The operator $A_s$ only depends on the vertex~$v$ whereas $B_s$ only depends on the plaquette~$p$ of the site~$s$. As a consequence, we often write $A_v:=A_s$ and $B_p:=B_s$ for $s=(v,p)$, and call these the vertex- and plaquette-operators associated with~$v$ and $p$, respectively. The quantum double Hamiltonian is then given by
\begin{align}
H&=-\sum_{v\in V} A_v -\sum_{p\in P} B_p 
\end{align}
where $P$ is the set of plaquettes of~$\cL$. Since all vertex- and plaquette-operators commute, the ground space of~$H$ is the simultaneous eigenspace of these operators to eigenvalue~$1$.

\subsection{Basic ribbon operators}
The discussion of topological operations in the quantum double model necessitates the introduction of three types of ribbon operators. The first type of operator, which we refer to as a basic ribbon operator, is associated
with an open ribbon~$\xi$. Such a ribbon consists of a sequence of sites
 connecting a starting site $s_0=(v_0,p_0)$ to a (distinct) ending site~$s_1=(v_1,p_1)$ by a  contiguous sequence of primary and/or dual triangles. Following the terminology of~\cite{BombinDelgado}, a
 primary triangle is defined by a pair $(e,p)$ where $e\in E$ is an edge and $p\in P$ an adjacent plaquette. In contrast, a dual triangle is defined by a pair $(e,e_*)$ where $e\in E$ is an edge and $e_*\in \{e^+,e^-\}$ is its head or tail. See Fig.~\ref{fig:ribbonbasic} for an illustration of these notions.
 \begin{figure}
 \centering
  \includegraphics[width=0.2\textwidth]{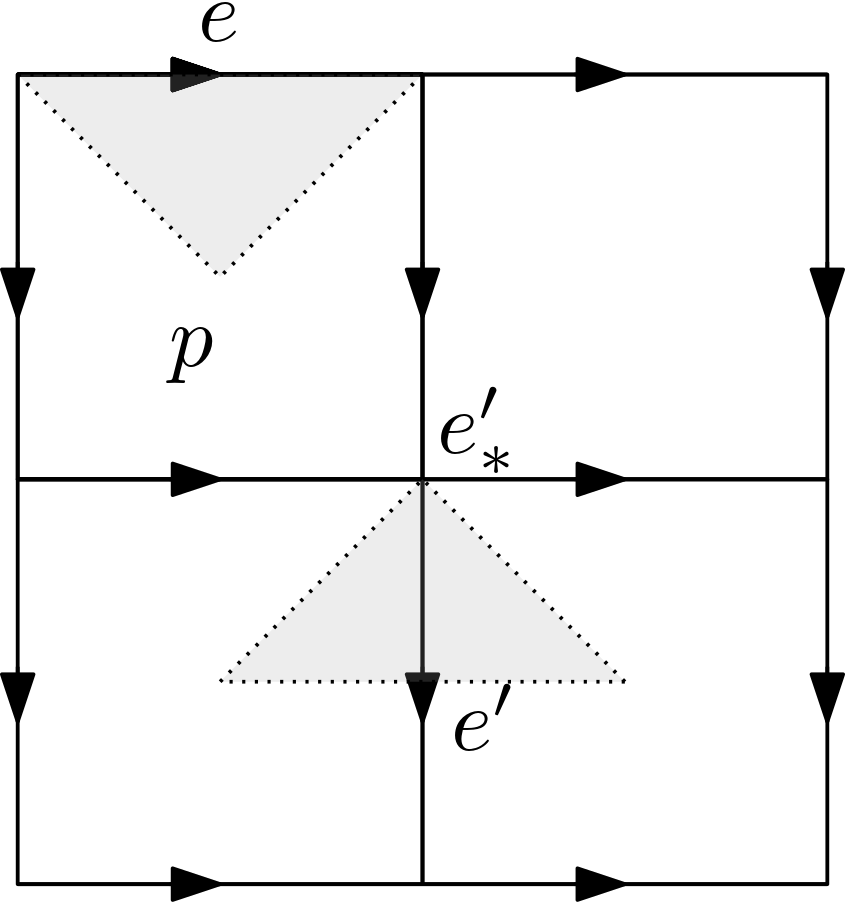}
 \caption{A primary triangle is specified by an edge~$e$ and an adjacent plaquette~$p$, and a dual triangle  by an edge~$e'$ and one of its vertices~$e_*$.\label{fig:ribbonbasic}}
 \end{figure}

 For each pair $(g,h)\in G^2$, there is a basic ribbon operator~$F^{h,g}_\xi$. Its action can be defined elegantly by defining the action for a ribbon consisting of a single (primary or dual) triangle, and subsequently  using the ``co-multiplication'' rule
 \begin{align}
   F_{\xi}^{h,g}&= \sum_{k\in G}F_{\xi_1}^{hk^{-1},g}F_{\xi_2}^{k,g}
 \end{align}
 (see \cite[Eq.~(26)]{kitaev}) to find the ribbon operator $F_\xi^{h,g}$ associated with a long ribbon $\xi=\xi_1\xi_2$. An example of the action of $F_\xi^{h,g}$ is given in Fig.~\ref{fig:ribbonoperatoraction}.
 
 \begin{figure}
  \includegraphics[width=0.9\textwidth]{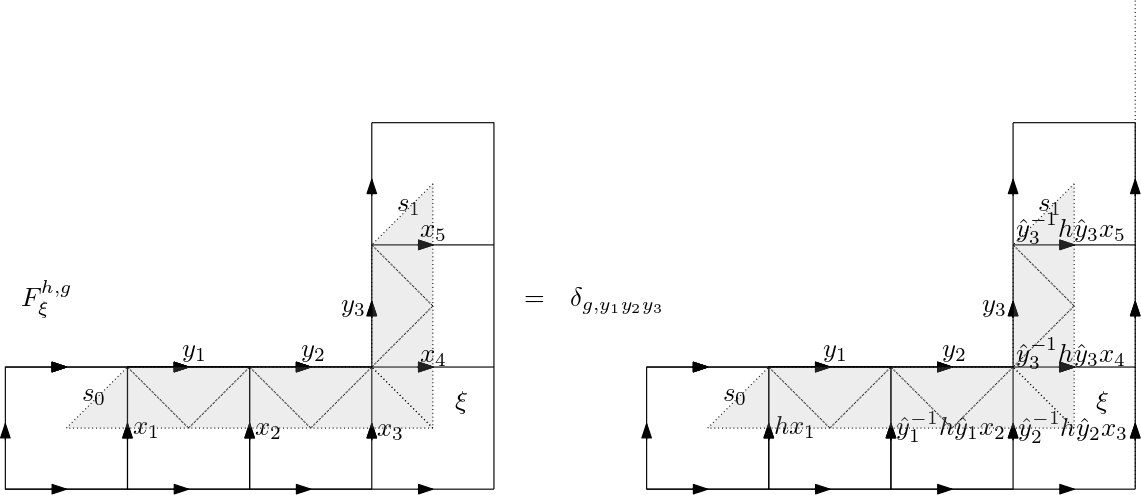}
 \caption{The ribbon operator $F_\xi^{h,g}$ for a ribbon~$\xi$ with starting site~$s_0$ and ending site~$s_1$. The figure illustrates  the action of the operator on a computational basis state
 specified by $(x,y)$.  Here $\hat{y}_j=y_1\cdots y_j$.\label{fig:ribbonoperatoraction}}
 \end{figure}

 The linear span of the operators $\{F^{h,g}_\xi\}_{(h,g)}$, when applied to a ground state~$\Psi$, gives the two-particle subspace~$\cH(s_0,s_1)$ of states that have excitations  at the end-points $s_0,s_1$ only. These are states that are simultaneous $+1$~eigenstates of all site-operators $A_s$ and $B_s$ with $s\not\in\{s_0,s_1\}$.

 For later use, we note the following properties of these basic ribbon operators. Let $\xi$ be a ribbon starting at site~$s_0$ and ending at site~$s_1$.
Then we have the relations
\begin{align}
F_\xi^{h,g}F_\xi^{h',g'}&=\delta_{g,g'}F_{\xi}^{hh',g}\label{eq:stringmultiplication}\\
  (F_\xi^{h,g})^\dagger &=F_\xi^{h^{-1},g}\ \label{eq:adjointstringops}
\end{align}
for all $g,h,g',h'\in G$. For any ground state~$\Psi$, we also have the identity (see e.g.,~\cite[Lemma~3.1]{BeigiShor})
\begin{align} 
\bra{\Psi} F_\xi^{h,g}\ket{\Psi}&=\frac{1}{|G|}\delta_{h,e}\ .\label{eq:groundstateexpect}
\end{align}
This is a consequence of certain commutation relations between the site-operators at the start- and end-points of~$\xi$ with the ribbon operators. At the startpoint~$s_0$ of $\xi$, we have the relations
\begin{align}
B_{s_{0}}^{k} F_{\xi}^{h, g}&=F_{\xi}^{h, g} B_{s_{0}}^{k h}\\
A_{s_{0}}^{k} F_{\xi}^{h, g}&=F_{\xi}^{k h k^{-1}, k g} A_{s_{0}}^{k} \label{eq:startpointcommutex} 
\end{align}
for all $g,h,k\in G$.  Similarly, at the endpoint~$s_1$ of~$\xi$, we have  the relations 
\begin{align}
  B_{s_1}^k F^{h,g}_\xi&= F_\xi^{h,g} B_{s_1}^{g^{-1}h^{-1}gk}\  \label{eq:commutendpoint}\\
A^k_{s_1} F^{h,g}_\xi&=F_\xi^{h,gk^{-1}}A^k_{s_1}\ \label{eq:seccommute}
\end{align}
for all $g,h,k\in G$.

%    \caption{A primary triangle $(e, p)$}
  %  \label{fig:primarytriangle}
%    \caption{A dual triangle $(e, e_*)$}
%    \label{fig:dualtriangle}
%   \label{fig:ribbon}
%    \label{fig:ribbonoperatoraction}
   % \caption{The definition of basic ribbon operators~$F^{h,g}_\xi$ associated with a ribbon~$\xi$.

\subsection{Elementary excitations (anyons) and anyonic ribbon operators}
 States of the form $F^{h,g}_\xi\Psi$ with $(h,g)\in G^2$
 form a basis of the two-particle space~$\cH(s_0,s_1)$. The action of the quantum double~$D(G)$, however, is better expressed with respect to a different basis whose elements can be interpreted as proper two-anyon states (and not superpositions thereof). This requires the introduction of a different set of ribbon operators~$\{F^{\rho;\alpha}_\xi\}_{(\rho;\alpha)}$ associated with an open ribbon~$\xi$. We call these {\em anyonic ribbon operators}. To introduce them, we need the following concepts. 

The (Drinfeld) quantum double~$D(G)$ of~$G$ is a Hopf algebra with basis~$\{D_{(h,g)}\}_{(g,h)\in G}$. It is also a $C^*$-algebra defined by the relations
\begin{align}
  D_{(h,g)}D_{(h',g')}&=\delta_{h,gh'g^{-1}}D_{h,gg'}\qquad\textrm{ for all }\qquad g,h,g',h'\in G\\
  D_{(h,g)}^*&= D_{g^{-1}hg,g^{-1}}\qquad\qquad \textrm{ for all }\qquad g,h\in G\ .
\end{align}
The site-operators $\{A^g_s\}_g$ and $\{B^h_s\}_h$ associated with a site~$s$
form a $*$-representation of this algebra via the map $D_{(h,g)}\mapsto A^g_sB^h_s$. For later reference, we note, in particular, the relations
\begin{align}
A^k_{s}A^{k'}_s &=A^{kk'}_s \label{eq:multiplicationas}\qquad \textrm{ and }\\
(A^k_s)^\dagger &=A^{k^{-1}}_s\label{eq:adjointAk}\qquad\textrm{ for every }\qquad k\in G\\
A^{\idG}_s&=I\ .\label{eq:Aesidentity}
  \end{align}

The elementary excitations (anyons) of the quantum double model are labeled by irreducible representations  (irreps) of~$D(G)$. Such an irrep can be specified by a pair $\rho=(C,\pi)$, where
\begin{enumerate}[(i)]
	\item
	$C$ is a conjugacy class of~$G$. We pick an arbitrary element $r_C\in G$ belonging of this conjugacy class such that $C=\{gr_Cg^{-1}\ |\ g\in G\}$, and enumerate the elements of~$C$ as $C=\{c_j\}_{j=1}^{|C|}$. The {\em centralizer} of $C$ is the group
	\begin{align}
		E(C)&=\{g\in G\ |\ gr_C=r_Cg\} .
	\end{align}
	We pick a set $P(C)=\{p_j\}_{j=1}^{|C|}$ of representatives of $G/E(C)$ such that $c_j=p_j r_C p_j^{-1}$.  We will choose these representatives such that $p_1=\idG$, i.e., $c_1=r_C$. 
	Every element $g\in G$ can be written in a unique way as $g=p_j n$ for some $j\in \{1,\ldots,|C|\}$ and $n\in E(C)$.
	\item
	$\pi$ is an irreducible (unitary) representation of~$E(C)$.
	For $g\in E(C)$, we will denote by $\Gamma_\pi(g)$ the corresponding unitary matrix.
\end{enumerate}
We refer to~\cite{bakalovkirillov01} for a specification of the action of~$D_{h,g}$ in an irrep~$\rho=(C,\pi)$.

We are now in position to give the definition of the anyonic ribbon operators~$\{F_\xi^{\rho;\alpha}\}_{(\rho;\alpha)}$ associated with an open ribbon~$\xi$. These were introduced in~\cite{BombinDelgado}; we follow the presentation/notation of~\cite{CongChenWang}. Every such ribbon operator is associated with a tuple~$(C,\pi,{\bf u},{\bf v})$; we write $F_\xi^{(C,\pi);({\bf u},{\bf v})}$ for the corresponding operator. Here $(C,\pi)=:\rho$ is an anyon label as discussed, i.e., it defines an irrep of~$D(G)$.  It remains to specify the labels $({\bf u},{\bf v})=:\alpha$. The definition is as follows:
\begin{enumerate}[(i)]
  \setcounter{enumi}{2}
	\item
	${\bf u}=(i,j)$, ${\bf v}=(i',j')$ are such that $i,i'\in \{1,\ldots,|C|\}$ index elements of the conjugacy class~$C$ and the pair~$(j,j')\in \{1,\ldots,d_\pi\}^{2}$ label matrix entries of the irrep~$\pi$.  Here and below we denote by~$d_\pi$ the dimension of the irrep~$\pi$.
\end{enumerate}
The ribbon operator associated with~$(C,\pi,{\bf u},{\bf v})$ is then given by
\begin{align}
	F_\xi^{(C,\pi);((i,j),(i',j'))}&:=\frac{d_\pi}{|E(C)|}\sum_{k\in E(C)}(\Gamma_{\pi}^{-1}(k))_{j,j'} F_\xi^{(c_i^{-1},p_i k p_{i'}^{-1})}\ .\label{eq:generalribbonoperatoranyonic}
\end{align}

The pair $(C,\pi)$ encodes the global degrees of freedom (in particular, these labels cannot be changed by local operators at the ends of a ribbon). In contrast, $({\bf u}, {\bf v})$ describes local degrees of freedom: these can be changed by applying the local operators $D^{(h,g)}$ (representing the quantum double) at the endpoints. Indeed, as argued in~\cite[Eq.~(23)]{BombinDelgado}, there is an operator~$d^{{\bf u'}}_{{\bf u''}}$ that is local at~$s_0$ (the start point of~$\xi$) and an operator~$d^{{\bf v'}}_{{\bf v''}}$ that is local at~$s_1$ (the end point of~$\xi$) such that
\begin{align}
  d^{{\bf u'}}_{{\bf u''}}d^{{\bf v'}}_{{\bf v''}}
  \ket{(C,\pi);({\bf u},{\bf v})}&=
  \delta_{{\bf u},{\bf u''}}\delta_{{\bf v},{\bf v''}}
  \ket{(C,\pi);({\bf u'},{\bf v'})}\ ,\label{eq:conversionlocaldegrees}
\end{align}
where we defined
\begin{align}
  \ket{(C,\pi);({\bf u},{\bf v})}&:=F_\xi^{(C,\pi);({\bf u},{\bf v})}\ket{\Psi}\label{eq:cpiuvstate}
\end{align}
for a fixed ground state~$\Psi$.

We note that the operators~\eqref{eq:generalribbonoperatoranyonic} are generally not unitary. We show the following in Appendix~\ref{app:unitarycombinations}: For an abelian group~$G$, we can fix any pair~$(C,\pi)$ and obtain a unitary linear combination
of the ribbon operators $\{F_\xi^{(C,\pi);({\bf u},{\bf v})}\}_{{\bf u},{\bf v}}$ corresponding to different local degrees of freedom. In contrast, for~$G=S_3$, such a unitary linear combination does not exist. This is important for how we formalize the task of implementing an operator~$F_\xi^{(C,\pi);({\bf u},{\bf v})}$, see e.g., Theorem~\ref{thm:nogotheorem} below.

\subsection{Topological charge projections\label{sec:topologicalchargeprojection}} 
Projections onto states of different topological charge have been constructed in~\cite[Appendix B.9]{BombinDelgado}.  They are defined for a closed ribbon~$\sigma$ (i.e., one for which start- and endpoint coincide) and are introduced in two steps. 

For a group $H$, let $H^{cj}$  denote the set of conjugacy classes of~$H$.
Let $C\in G^{cj}$ be a conjugacy class of~$G$. We pick an element $D\in E(C)^{cj}$, i.e., a conjugacy class of the centralizer~$E(C)$ of a fixed representative $r_C\in C$ as defined above. For each pair $(C,D)$ with $C\in G^{cj}$ and $D\in E(C)^{cj}$, one introduces the operator
\begin{align}
K_{\sigma}^{(D,C)}:=\sum_{j=1}^{|C|} \sum_{d\in D} F_\sigma^{p_jd p_j^{-1},p_jr_Cp_j^{-1}}\ \label{eq:ksigmadc}
\end{align}
where again $\{p_j\}_{j=1}^{|C|}$ is a fixed set of representatives of~$G/E(C)$ as before.

The operators~\eqref{eq:ksigmadc} are not projections. Instead, another set of operators is defined as follows. Let $(C,\pi)$ be an anyon label (i.e., $C\in G^{cj}$ and $\pi$ is an irrep of~$E(C)$). Then we define~\footnote{In contrast to~\cite{BombinDelgado}, it seems the correct definition here as stated, i.e.,  without a factor~$d_\pi$, see Appendix~\ref{app:checkingorthogonality}.}
  \begin{align}
K_\sigma^{(C,\pi)}&:=\frac{1}{|E(C)|}\sum_{D\in E(C)^{cj}} \overline{\chi_\pi(D)} K_\sigma^{(D,C)}\label{eq:ksigmarc}
\end{align}
Here $\chi_\pi(D)=\chi_\pi(d)$ is the character of~$\pi$ for any representative~$d\in D$.  The operators $\{K_\sigma^{(C,\pi)}\}_{(C,\pi)}$ are pairwise orthogonal projections that sum to the identity. For completeness, we include a proof of this claim in Appendix~\ref{app:checkingorthogonality}. They therefore constitute a POVM whose measurement outcome~$(C,\pi)$ reveals the topological charge (anyon label) of the region enclosed by the ribbon~$\sigma$.

\section{A circuit depth lower bound for anyonic ribbon operators \label{sec:nogoresult}}
Let us use the shorthand
\begin{align}
 F_\xi^{(C,\pi)}:=F_\xi^{(C,\pi),((1,1),(1,1))}\ \label{eq:fxicpi}
\end{align}
such that
  \begin{align}
    F_\xi^{(C,\pi)}&=\frac{d_\pi}{|E(C)|}\sum_{\ell\in E(C)}
(\Gamma^{-1}_\pi(\ell))_{1,1}    F_\xi^{r_C^{-1},\ell}\ .\label{eq:fxicpidef}
    \end{align} 
  We show that some of these operators are  not implementable by a non-adaptive constant-depth unitary when the group is non-abelian, see Corollary~\ref{cor:nogomain} below. It is a simple consequence of the  following result:
  \begin{theorem}\label{thm:nogotheorem}
    Suppose the pair $(C,\pi)$ satisfies
    \begin{align} 
\exists r\in E(C)\qquad \textrm{ such that }\qquad
|\Gamma_\pi(r)_{1,1}|<1\  .\label{eq:Kproperty}
    \end{align}
    Let $U$ be a local unitary circuit implementing $F_\xi^{(C,\pi)}$ in the sense that it maps a ground state~$\Psi$  to a state proportional to~$F_{\xi}^{(C,\pi)}\Psi$. Then~$U$ has an extensive circuit depth.    
    \end{theorem}
  We note that property~\eqref{eq:Kproperty} cannot be satisfied if $E(C)$ is abelian. Indeed, if $E(C)$ is abelian, then any irreducible representation~$\pi$ of $E(C)$ is $1$-dimensional. In particular, a unitary irreducible representation~$\pi$ satisfies $|\Gamma_\pi(r)|=1$ for all $r\in E(C)$. 
  
 A sufficient condition for property~\eqref{eq:Kproperty} to hold is the following.
  \begin{lemma}\label{lem:conditiondimension}
  Suppose $d_\pi>1$. Then~\eqref{eq:Kproperty} is satisfied. 
  \end{lemma}
  \begin{proof}
  Suppose for the sake of contradiction that 
  \begin{align}
  |\Gamma_\pi(r)_{1,1}|=1\qquad\textrm{ for all }\qquad r\in E(C)\ .
  \end{align}
  Then the first basis vector $e_1$ spans an invariant subspace (by unitarity of the representation~$\pi$). This contradicts the irreducibility of~$\pi$.
    \end{proof}

  Let us show more explicitly that for $G=S_3$, there are pairs~$(C,\pi)$ such that~\eqref{eq:Kproperty} is satisfied. Consider the conjugacy class $C=\{\idG\}$. Here $E(C)=S_3$.
  Consider the $2$-dimensional irreducible representation~$\pi$ of $S_3$ obtained by restricting the defining representation on $\mathbb{C}^3$ to the subspace
  \begin{align}
V:= \left\{
    \begin{pmatrix}
      a\\
      b\\
      c
      \end{pmatrix}\ \Big|\ a+b+c=0
    \right\}\ .
  \end{align}
  This representation is unitary when expressed in the orthonormal basis
  \begin{align}
    e_1=\begin{pmatrix}
    \sqrt{2/3}\\
    -1/\sqrt{6}\\
    -1/\sqrt{6}
    \end{pmatrix}\qquad\qquad\qquad
    e_2&=\begin{pmatrix}
    0\\
    1/\sqrt{2}\\
    -1/\sqrt{2}
    \end{pmatrix}\ .
    \end{align}
  Let $r=(1\ 2)\in S_3$ be the transposition of the first two elements. Then  (expressed in this basis)
  \begin{align}
          \Gamma_\pi(r)&=\frac{1}{2}\begin{pmatrix}
    -1 & \sqrt{3}\\
    \sqrt{3} & 1
    \end{pmatrix}\ ,
  \end{align}
  and we conclude that~\eqref{eq:Kproperty} is satisfied.

  More generally,  we can show that \eqref{eq:Kproperty} is always satisfied for some pair $(\{\idG\},\pi)$  when~$G$ is non-abelian. Indeed, assume for the sake of contradiction that this is not the case. Then this implies by 
  Lemma~\ref{lem:conditiondimension} that  $d_\pi=1$ for any irrep~$\pi$ of $E(\{\idG\})=G$. By the Peter-Weyl theorem, the group algebra~$\mathbb{C}[G]$ is isomorphic to a direct sum of matrix algebras determined by irreps, hence this implies that $\mathbb{C}[G]$ and hence~$G$ is abelian, a contradiction.   In summary, this yields the following corollary to Theorem~\ref{thm:nogotheorem}:
  \begin{corollary}\label{cor:nogomain}
    For any non-abelian group~$G$, there is an irrep~$\pi$ of~$G$ such that
    any constant-depth unitary circuit implementing~$F_\xi^{(\{\idG\},\pi)}$ has extensive circuit depth.
  \end{corollary}

  In the remainder of this section, we give the proof of Theorem~\ref{thm:nogotheorem}. We begin by computing certain matrix elements.  
\begin{lemma}\label{lem:KOlemma}
  Let $\Psi$ be a ground state. Let $(C,\pi)$ be a pair satisfying~\eqref{eq:Kproperty} with some $r\in E(C)$. Set $k:=r^{-1}$ and define 
\begin{align}
K&=A^k_{s_1}\label{eq:Kdefinition}\\
O_\xi&=\sqrt{\frac{|E(C)|\cdot |G|}{d_\pi}} F_\xi^{(C,\pi)}\label{eq:Odefinition}\ .
  \end{align}
Then 
\begin{align}
  \|O_\xi\Psi\|&=\|KO_\xi\Psi\|=1\\
  \langle O_\xi\Psi,KO_\xi\Psi\rangle &=\Gamma_\pi(k^{-1})_{1,1}\ .
\end{align}
In particular, $\langle O_\xi\Psi,KO_\xi\Psi\rangle|<1$. 
\end{lemma}

\begin{proof}
  We show that for any~$k\in E(C)$, we have 
 \begin{align}
   \langle \Psi,(F_\xi^{(C,\pi)})^\dagger A_{s_1}^kF_\xi^{(C,\pi)}\Psi\rangle&=
   \frac{d_\pi}{|E(C)|\cdot |G|} \Gamma_{\pi}(k^{-1})_{1,1}\label{eq:akhequation}\\
\|F_\xi^{(C,\pi)}\Psi\|^2&=
    \frac{d_\pi}{|E(C)|\cdot |G|}\label{eq:secondclaimakh}\\
\|A^k_{s_1}F_\xi^{(C,\pi)}\Psi\|^2&=
    \frac{d_\pi}{|E(C)|\cdot |G|}\ .\label{eq:thirdclaimakh}
 \end{align}
 The claim follows from these three identities.
 
 Note that for any $k\in G$, we have 
  \begin{align}     (F^{r_C^{-1},\ell'}_\xi)^\dagger A_{s_1}^k F^{r_C^{-1},\ell}_\xi    &=    (F^{r_C^{-1},\ell'}_\xi)^\dagger F^{r_C^{-1},\ell k^{-1}}_\xi   A_{s_1}^k \qquad\qquad\textrm{by~\eqref{eq:seccommute}}\\
    &= F^{r_C,\ell'}_\xi   F^{r_C^{-1},\ell k^{-1}}_\xi   A_{s_1}^k \qquad\qquad\textrm{by~\eqref{eq:adjointstringops}}\\
&=\delta_{\ell',\ell k^{-1}} F^{\idG,\ell k^{-1}}_\xi     A_{s_1}^{k}\qquad\qquad\textrm{by~\eqref{eq:stringmultiplication}}\ .
      \end{align}  
  With~\eqref{eq:fxicpidef}  we obtain
  \begin{align}
    (F_\xi^{(C,\pi)})^\dagger A_{s_1}^kF_\xi^{(C,\pi)}&=\frac{d_\pi^2}{|E(C)|^2}
  \sum_{\ell\in E(C)} \delta_{\ell k^{-1}\in E(C)}
\overline{(\Gamma^{-1}_\pi(\ell k^{-1}))_{1,1}} (\Gamma^{-1}_\pi(\ell))_{1,1}    F^{\idG,\ell k^{-1}}_\xi     A_{s_1}^{k}\ .
  \end{align}
  Let us now consider the case where $k\in E(C)$ such that $\ell k^{-1}\in E(C)$ for every $\ell\in E(C)$.   Since $A^k_{s_1}A_{s_1}=A_{s_1}$ and thus $A^k_{s_1}\Psi=\Psi$ and because
  of~\eqref{eq:groundstateexpect}, we get
  \begin{align}
    \langle \Psi,(F_\xi^{(C,\pi)})^\dagger A_{s_1}^k F_\xi^{(C,\pi)}\Psi\rangle&=\frac{d_\pi^2}{|E(C)|^2\cdot |G|}
  \sum_{\ell\in E(C)}
  \overline{(\Gamma^{-1}_\pi(\ell k^{-1}))_{1,1}} (\Gamma^{-1}_\pi(\ell))_{1,1}    .
  \end{align}
Using the unitary of~$\pi$ and Schur's orthogonality relations for matrix elements of irreducible representations we have 
  \begin{align}
  \sum_{\ell\in E(C)}
  \overline{(\Gamma^{-1}_\pi(\ell k^{-1}))_{1,1}} (\Gamma^{-1}_\pi(\ell))_{1,1}
  &=\sum_{\ell\in E(C)} \Gamma_\pi(\ell k^{-1})_{1,1}\overline{\Gamma_{\pi}(\ell)_{1,1}}\\
  &=\sum_{r=1}^{d_\pi }\Gamma_{\pi}(k^{-1})_{r,1}\sum_{\ell\in E(C)}\Gamma_{\pi}(\ell)_{1,r}\overline{\Gamma_{\pi}(\ell)_{1,1}}\\
  &=\frac{|E(C)|}{d_{\pi}}\Gamma_{\pi}(k^{-1})_{1,1}\ .
  \end{align}  
  Eq.~\eqref{eq:akhequation} follows.

  Eq.~\eqref{eq:secondclaimakh} follows by setting $k=\idG$ because of~\eqref{eq:Aesidentity}.

  For~\eqref{eq:thirdclaimakh} we have
  \begin{align}
    \|A^k_{s_1} F^{(C,\pi)}_\xi \Psi\|^2 &= \langle \Psi, (A^k_{s_1} F_\xi^{(C,\pi)})^\dagger A^k_{s_1} F^{(C,\pi)}_\xi \Psi\rangle\\
    &= \langle \Psi, (F_\xi^{(C,\pi)})^\dagger A^{\idG}_{s_1} F_\xi^{(C,\pi)}\Psi\rangle \qquad\textrm{by~\eqref{eq:multiplicationas}}\\
    &=\|F_\xi^{(C,\pi)}\Psi\|^2\qquad\textrm{by~\eqref{eq:Aesidentity}}\ .
  \end{align}
  The claim~\eqref{eq:thirdclaimakh} now follows from~\eqref{eq:akhequation}.
\end{proof} 

The vertex operator~$A^k_{s_1}$ in the definition~\eqref{eq:Kdefinition} of~$K$ is an example of a ``smallest'' closed ribbon operator. It is a ribbon operator associated with the  dual closed ribbon consisting only of dual triangles enclosing the single vertex associated with site~$s_1$. Let~$\tau$ denote this dual ribbon. Since~$K$ is
equal  to~$A^k_{s_1}$, we will write $K_\tau=K$  to emphasize that~$K$ is  defined by the ribbon~$\tau$. Importantly, the operator~$K_\tau$ does not depend on the ends (starting/ending site) of~$\tau$. As such, its action on the ground space is invariant under deformations of~$\tau$ in which the end is not fixed (i.e., ``rotations'' of the ribbon, see~\cite{BombinDelgado}). More generally, if $\tau'$ is another closed ribbon and  $\Psi$ a state such that the ribbon~$\tau$ can be deformed into~$\tau'$ without crossing any excitation, then
\begin{align}
K_\tau \Psi&=K_{\tau'}\Psi\label{eq:deformationcondition}
\end{align}
(see~\cite[Eq.~(11)]{BombinDelgado}), where $K_{\tau'}$ is the ribbon operator associated with~$\tau'$. It will be convenient in the following to reformulate condition~\eqref{eq:deformationcondition} as follows: Suppose $\cS(\tau,\tau')$
is a region of qudits such that $\tau$ can be deformed into~$\tau'$  without leaving~$\cS(\tau,\tau')$ (i.e., $\cS(\tau,\tau')$ contains the support of a deformation retract of $\tau$ to~$\tau'$). Let $\Pi_{\cS(\tau,\tau')}$ denote the corresponding local ground state projection, i.e., the product of plaquette- and vertex-operators whose support intersects~$\cS(\tau,\tau')$. Then we have
\begin{align}
K_\tau \Pi_{\cS(\tau,\tau')}&=K_{\tau'}\Pi_{\cS(\tau,\tau')}\ .\label{eq:KtauPiS}
  \end{align}

\begin{proof}[Proof of Theorem~\ref{thm:nogotheorem}]
  We work on the sphere where the ground state~$\Psi$ is unique but comment on how this assumption can be lifted below. According to Lemma~\ref{lem:KOlemma}, it suffices to show that
  implementing~$O_\xi$, i.e., mapping~$\Psi$ to $O_\xi\Psi$, requires at least linear circuit depth.

  Suppose for the sake of contraction that there is a local unitary circuit~$U$ implementing~$O_\xi$ with a non-extensive depth.  Consider the closed ribbon~$\tau$ associated with the endpoint of~$\xi$, and let $K_\tau$ be as in Lemma~\ref{lem:KOlemma}.
  Define \begin{align}
\Psi_1&:=O_\xi\Psi\\
\Psi_2&:=K_\tau\Psi_1=K_\tau O_\xi\Psi\ .
  \end{align}
  By Lemma~\ref{lem:KOlemma}, both states~$\Psi_1,\Psi_2$ are normalized, and
\begin{align}
|\langle \Psi_1,\Psi_2\rangle|&<1\label{eq:psionetwooverlap}
\end{align}
by assumption.   Let $\tau'$ be any ribbon obtained by locally deforming the closed ribbon~$\tau$. Then we have
  \begin{align}
K_{\tau'} O_\xi \Psi &=K_\tau O_\xi \Psi\ ,\label{eq:deformationidentity}
  \end{align}
   In particular,  identity~\eqref{eq:deformationidentity} implies that the state~$\Psi_2$ does not depend on the detailed form of~$\tau$.

  Eq.~\eqref{eq:deformationidentity}
  is a consequence of the following steps:
  \begin{align}
  K_\tau O_\xi \Psi&= 
  K_\tau O_\xi \Pi_{\cS(\tau,\tau')} \Psi \qquad\textrm{ since $\Psi$ is a ground state }\\
  &=K_\tau\Pi_{\cS(\tau,\tau')} O_\xi  \Psi \qquad\textrm{ since $O_\xi$ commutes with $\Pi_{\cS(\tau,\tau')}$ }\\
  &=K_{\tau'}\Pi_{\cS(\tau,\tau')} O_\xi  \Psi \qquad\textrm{ by~\eqref{eq:KtauPiS} }\\
  &=K_{\tau'} O_\xi  \Pi_{\cS(\tau,\tau')}\Psi \qquad\textrm{ since $[O_\xi, \Pi_{\cS(\tau,\tau')}]=0$}\\
  &=K_{\tau'} O_\xi  \Psi \ ,
    \end{align}
    see Fig.~\ref{fig:diagrammaticproof} for a diagrammatic representation. 
    
     \begin{figure}
 \centering
\subfigure[$K_\tau O_\xi \Psi$]{\includegraphics[width=8cm]{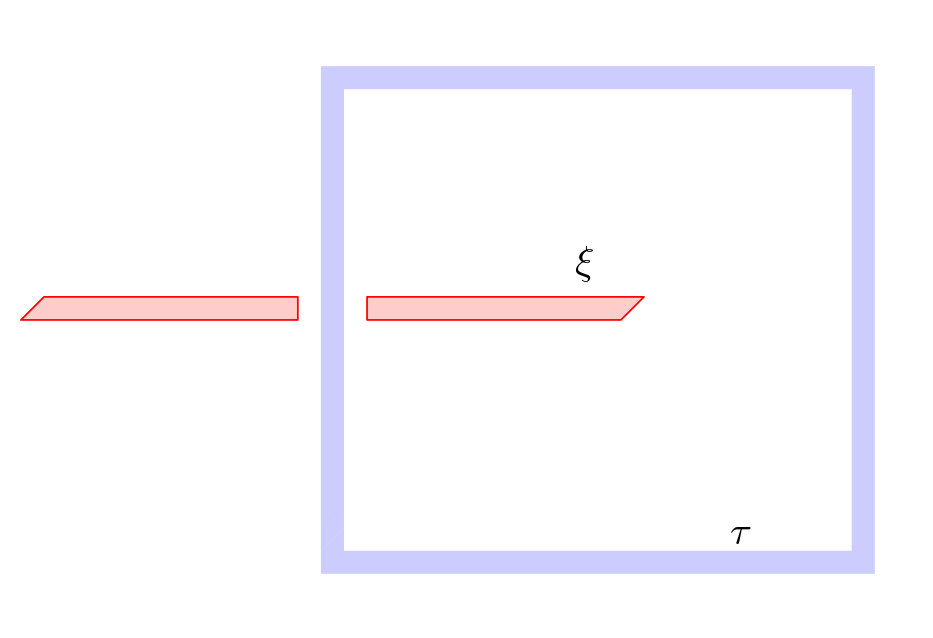}}\qquad\qquad
\subfigure[$K_\tau O_\xi \Pi_{\cS(\tau,\tau')} \Psi$]{\includegraphics[width=8cm]{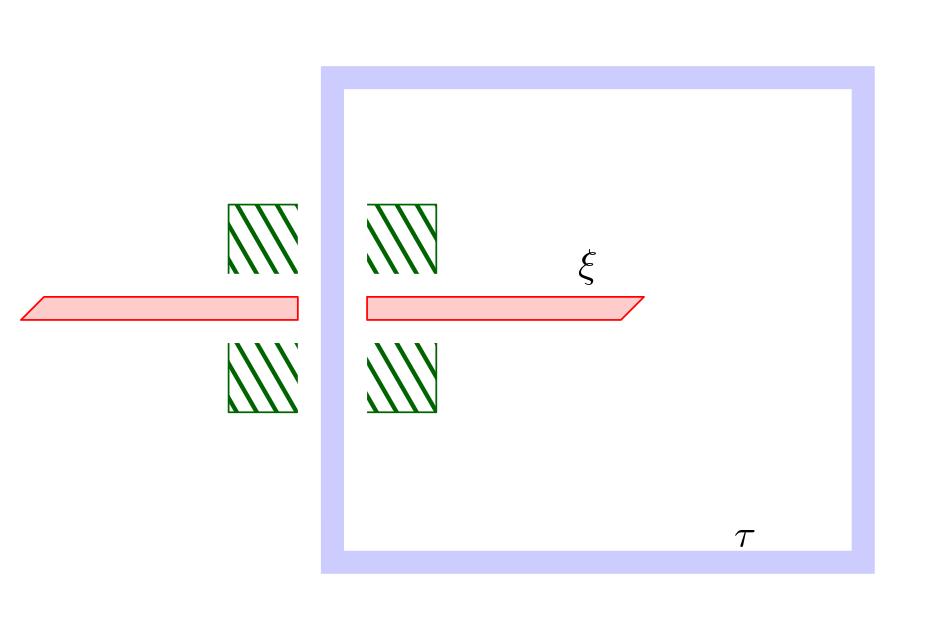}}\\
\subfigure[$K_\tau\Pi_{\cS(\tau,\tau')} O_\xi  \Psi $]{\includegraphics[width=8cm]{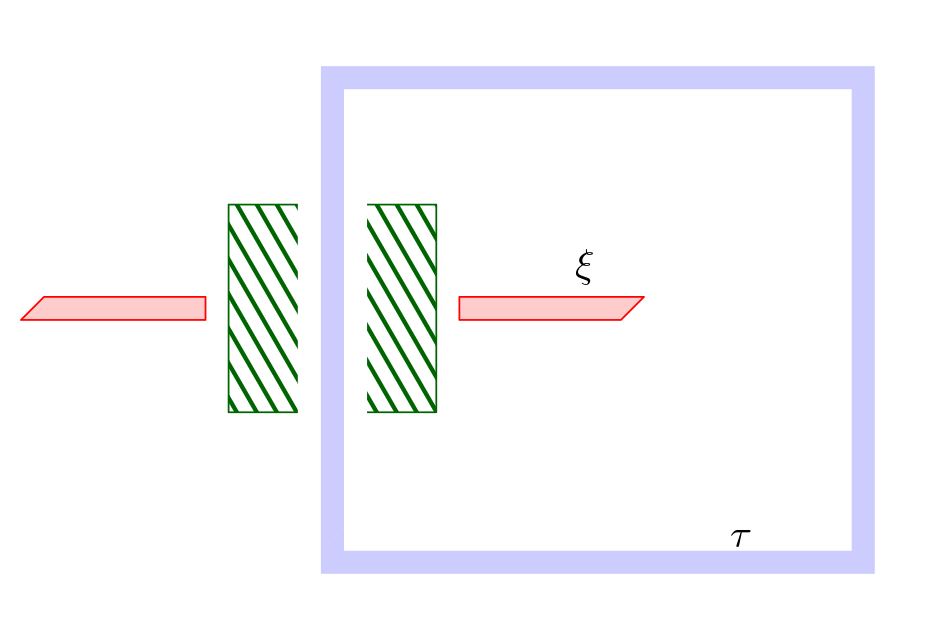}}\qquad\qquad
\subfigure[$K_{\tau'}  \Pi_{\cS(\tau,\tau')}O_\xi \Psi$]{\includegraphics[width=8cm]{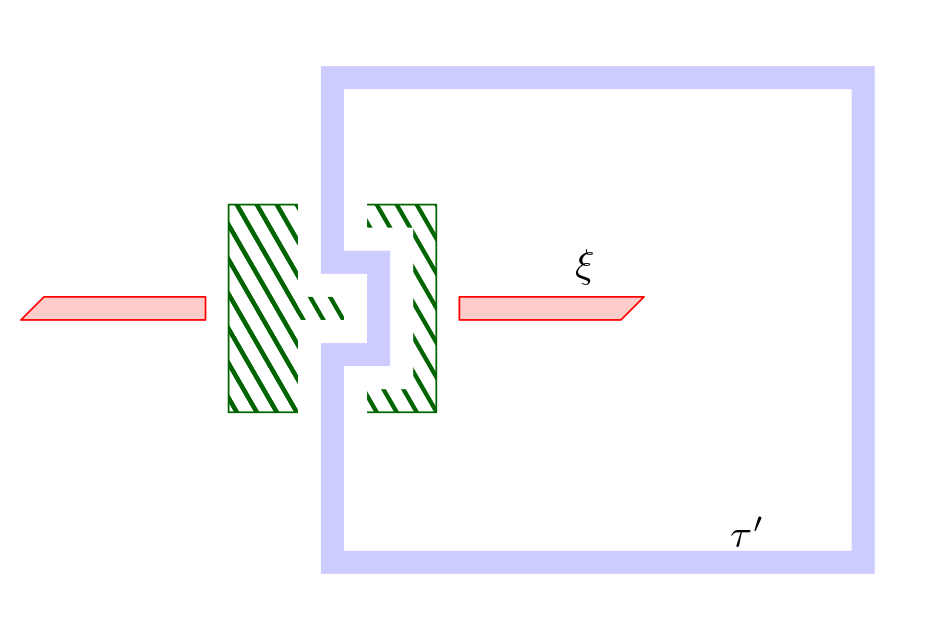}}\\
\subfigure[$K_{\tau'} O_\xi \Psi$]{\includegraphics[width=8cm]{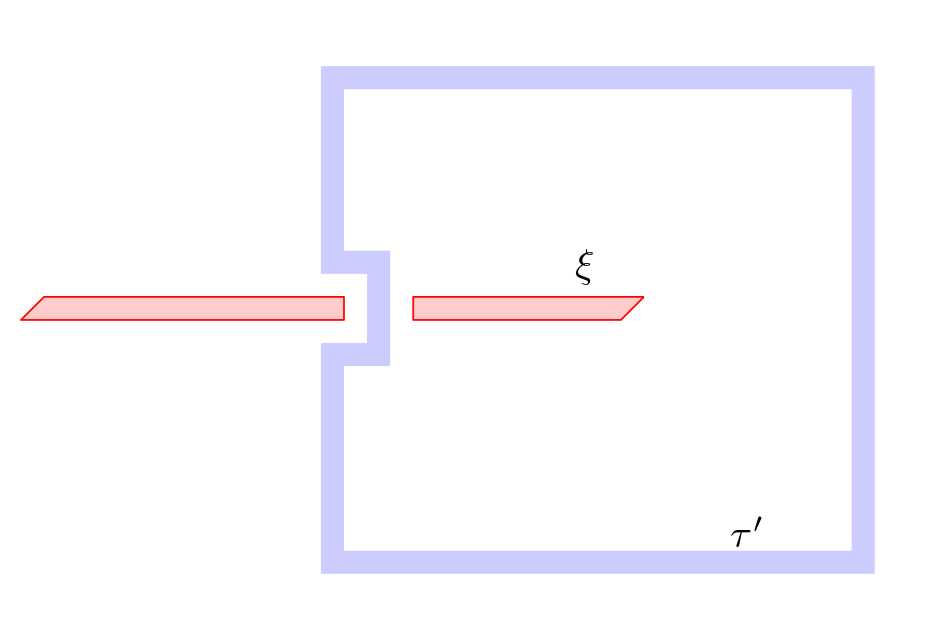}}
\caption{Illustration of the proof of identity~\eqref{eq:deformationidentity} 
\label{fig:diagrammaticproof}}
\end{figure}

Suppose now that $U$ is a constant-depth circuit implementing~$O_\xi$. Consider the states
\begin{align}
\Psi'_1&:=U^\dagger \Psi_1\\
\Psi'_2&:=U^\dagger \Psi_2\ .
\end{align}
Then~\eqref{eq:psionetwooverlap} implies that 
\begin{align}
|\langle \Psi'_1,\Psi'_2\rangle|&<1\label{eq:psipoverlap}
\end{align}
By definition, $\Psi'_1=\Psi$, and 
\begin{align}
\Psi'_2&=U^\dagger K_\tau   \Psi_1\\
&=(U^\dagger K_\tau U)\Psi\ .\label{eq:udaggerku}
\end{align}
In particular, $\Psi'_1$ and $\Psi'_2$ are two states  that differ by a local operator~$U^\dagger K_\tau U$.
By assumption,~$\Psi'_1=\Psi$ is a ground state. We claim that $\Psi'_2$ is also a ground state. With~\eqref{eq:psipoverlap}, this contradicts the uniqueness of the ground state.

To show that $\Psi'_2$ is a ground state, consider any stabilizer~$S$ with support far from~$\tau$. Then expression~\eqref{eq:udaggerku} and the locality of~$(U^\dagger K_\tau U)$ show that~$S\Psi'_2=\Psi'_2$.
More generally,  if a stabilizer~$S$ is such that its support intersects the support of~$(U^\dagger K_\tau U)$, we can simply replace~$\tau$ by a deformed closed ribbon~$\tau'$ such that the corresponding intersection is empty. 

The argument used here generalizes to the case where the ground space is degenerate as e.g., in the case of a lattice embedded in a higher-genus surface. This is because the ground space is an error correcting code. The only required assumption is that the ribbon~$\xi$ and the closed ribbon~$\tau$ are supported inside a contractible region~$R$ not containing any excitations or holes. Such a region is correctable: a logical operator supported inside~$R$ has trivial action on the ground space. It then follows as before that $\Psi_1'$ and $\Psi_2'$ are two ground states that differ only by a global phase, a contradiction to~\eqref{eq:psipoverlap}. 
\end{proof}

\section{Adaptive ground state preparation for a solvable group~$G$ \label{sec:statepreparation}}
In this section, we show that a ground state of the quantum double model can be prepared by an adaptive constant-depth circuit with local gates if the underlying group~$G$ is solvable.

In Section~\ref{sec:preparinggroundstatesthree}, we  exemplify our construction using the group~$G=S_3$ and a square lattice. We note that for $S_3$, a similar procedure was proposed in~\cite{verresen2022efficiently}. In Section~\ref{sec:qdoublestate} we define a quantum double state for any connected planar graph.  This state is the unique ground state of Kitaev's quantum double model on this graph. In Section~\ref{sec:stateprepcircuit} we then construct a constant-depth adaptive  circuit which prepares this state for any solvable group~$G$.

\subsection{Example: Preparing the ground state associated with $G=S_3$\label{sec:preparinggroundstatesthree}}
We begin by illustrating our construction using the symmetric group~$S_3$ and a planar square lattice with ``smooth" boundaries on all sides, see Fig.~\ref{fig:baselattice}.

\begin{figure}[H]
	\centering
	\includegraphics[width=0.3\textwidth]{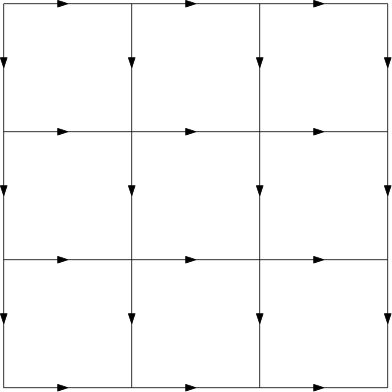}
	\caption{The base (oriented) lattice. Each  edge~$e$ carries a qudit~$\sysC_e\cong\mathbb{C}^{|G|}$.\label{fig:baselattice}}
\end{figure}

Let $r=(1\ 2\ 3)\in S_3$ be the $3$-cycle.  Then 
\begin{align}
	H&=\{\idG,r,r^2\}
\end{align}
is a normal subgroup of~$S_3$ isomorphic to~$\mathbb{Z}_3$. The quotient group~$S_3/H$ is cyclic and isomorphic to~$\mathbb{Z}_2$. More explicitly, let $a=(1\ 2)$ be the $2$-cycle. Then $S_3$ is the disjoint union of cosets
\begin{align}
	S_3=H\cup aH\ .
\end{align}
Every element $g\in S_3$ can thus be written  as 
\begin{align}
	g&=a^{\theta} h \label{eq:uniquelywrittenSthree}
\end{align}
where $\theta\in \mathbb{Z}_2$ and $h\in H$ are uniquely determined by~$g$. 
%We use the correspondence between elements $g\in S_3$ and $(\theta,h)\in\mathbb{Z}_2\times H$ defined by~\eqref{eq:uniquelywrittenSthree} to define a Hilbert space isomorphism $\mathbb{C}^{|G|}\cong\mathbb{C}^{2}\otimes\mathbb{C}^3$.

%The Hilbert space associated with an edge then has an orthonormal basis~$\{\ket{\theta}\otimes\ket{g}\}_{(\theta,g)\in\mathbb{Z}_2\times S_3}$.

The constant-depth adaptive circuit for preparing the ground state of~$S_3$ requires systems as follows:
To every edge~$e$ of the lattice, we associate Hilbert spaces
\begin{align}
	\sysB_e &\cong \mathbb{C}^2=\mathbb{C}^{|S_3/H|}\\
	\sysC_e &\cong \mathbb{C}^6=\mathbb{C}^{|S_3|}\ .
\end{align} We also associate a qubit Hilbert~$\sysA_v\cong\mathbb{C}^2$ to every vertex~$v$ of the lattice with orthonormal basis~$\{\ket{\theta}\}_{\theta\in\mathbb{Z}_2}$. The relevant systems and their geometric arrangement 
are illustrated in Fig.~\ref{fig:preparationgroundstatesolvablegroup}
for a general solvable group~$G$. The final state of the protocol on the system~$\sysC$ is the  
ground state of the quantum double~$D(S_3)$ which we denote~$\Psi(S_3)$.

We begin by noting that the ground state $\Psi(\mathbb{Z}_3)$  associated with the quantum double $D(\mathbb{Z}_3)$ 
can be prepared in constant adaptive depth with local (quantum) operations (and non-local classical processing) since~$\mathbb{Z}_3$ is an abelian group. In particular, the state~$\Psi(\mathbb{Z}_3)$ is a stabilizer state and can be prepared by syndrome measurement and subsequent (generalized) Pauli corrections. An example of a corresponding procedure  relying on a graph state can be found, e.g.,  in~\cite{joooetal}.

Our protocol for preparing the state $\Psi(S_3)$ proceeds as follows: 
\begin{enumerate}[(i)]
	\item
	Each qubit~$\sysA_v$   associated with vertex~$v$ is initialized in the state~$\ket{+}:=\frac{1}{\sqrt{2}}(\ket{0}+\ket{1})$.
	\item
	Each qubit~$\sysB_e$ associated with an edge~$e$ is initialized in
	the state~$\ket{0}$.
	\item
	The subsystem $\bigcup_{e\in \calL_1} \sysC_e$ is initialized in the quantum double state $\Psi(\mathbb{Z}_3)$.
	Here we identify  a basis vector $|h\ra \in \CC^3$ associated with a group element $h\in \ZZ_3=\{\idG,r,r^2\}\subset S_3$
	and the respective basis vector $|h\ra \in \CC^{|S_3|}$.
\end{enumerate}

Our preparation procedure uses 
a controlled-left-multiplication gate~$CX^{L}:\mathbb{C}^2\otimes \mathbb{C}^{|S_3|}\rightarrow\mathbb{C}^2\otimes\mathbb{C}^{|S_3|}$ where the group element~$g\in S_3$ contained in the target register is left-multiplied by the $2$-cycle~$a$ if the control bit is in the state~$\ket{1}$, that is,
\begin{align}
	CX^{L} \left(\ket{\theta}\otimes\ket{g}\right)&=\ket{\theta}\otimes \ket{a^{\theta}g}\qquad\textrm{ for }\qquad (\theta,g)\in \{0,1\}\times S_3\ .
\end{align}
It also uses a similarly defined controlled-right-multiplication gate~$CX^{R}$ defined by 
\begin{align}
	CX^{R} \left(\ket{\theta}\otimes\ket{g}\right)&=\ket{\theta}\otimes \ket{ga^{-\theta}}\qquad\textrm{ for }\qquad (\theta,g)\in \{0,1\}\times S_3\ .
\end{align}
Finally, we need a unitary~$E:\mathbb{C}^{|S_3|}\times\mathbb{C}^2\rightarrow\mathbb{C}^{|S_3|}\times\mathbb{C}^2$ that coherently computes the bit~$\theta\in \{0,1\}$ associated with~$g\in S_3$ by~\eqref{eq:uniquelywrittenSthree}, that is,
\begin{align}
	E(\ket{g}\otimes\ket{0})&=\ket{g}\otimes\ket{\theta}\ .
\end{align}
Our preparation procedure for the $S_3$-toric code  uses a sequence of steps consisting of layers of gates and measurements.  We illustrate  the following steps in Fig.~\ref{fig:preparationgroundstatesolvablegroup} in the setting of a general solvable group.
\begin{enumerate}
	\item
	For every edge~$e$, the gate $CX^L$ is applied to~$\sysA_{e^+}\sysC_e$ and the gate
	$CX^R$ is applied to  $\sysA_{e^-}\sysC_e$.
	Since left- and right-multiplication commute, these gates can be applied in any order (and simultaneously for all edges). 
	\item
	Measure the qubit~$\sysA_v$ associated with each vertex~$v$ in the computational basis, obtaining measurement outcomes~$m=(m_v)_{v}$.
	\item
	Apply the unitary~$E$ to each pair of systems~$\sysC_e\sysB_e$ for every edge~$e$.
	\item
	Apply a single-qubit ``correction'' operator to each system~$\sysB_e$ for $e\in E$. For each qubit, this is either
	the identity or a Pauli-$Z$ operation. Whether or not~$Z$ is applied to a qubit~$\sysB_e$ is determined by the measurement results~$m$; we describe how to compute the relevant bit below.
	\item
	Apply the inverse of the unitary~$E$ to every pair~$\sysC_e\sysB_e$ for every edge~$e$.
\end{enumerate}
We will show that the resulting state on~$\sysC$ is the ground state of the quantum double model associated with~$S_3$ -- the remaining degrees of freedom can be traced out.

We note that the case of the group~$S_3$ considered here is special because $a^2=e$ and $|S_3/H|=|\mathbb{Z}_2|=2$, that is, the order of the element~$a\in S_3$ matches that of the quotient group~$S_3/H$. For a general solvable group~$G$, this is not necessarily the case and we may need to use qudits of different dimensions at the vertices (i.e., the systems~$\sysA_v$) and edges (i.e., the systems~$\sysB_e$). The corresponding generalization is described in Section~\ref{sec:stateprepcircuit} where we also argue that this procedure indeed creates the quantum double state.

As we comment below, qubits~$\sysB_e$ can be removed from the circuit simply by redefining the gates acting on the 
remaining degrees of freedom~$\sysA_v$ and~$\sysC_e$. However, these qubits make the structure of the circuit more 
clear and simplify the presentation.

\subsection{Quantum double state\label{sec:qdoublestate}}

Let $\calL$ be a connected planar  graph, e.g. the 2D square lattice.
We write $\calL_0,\calL_1$ and $\calL_2$ for the sets of 
vertices, edges, and faces of $\calL$ respectively.
Every edge $e\in \calL_1$ 
is oriented such that one endpoint of $e$ is designated as the head and the other endpoint
is designated as the tail. We write $e^+$ and $e^-$ for the head and the tail of $e$.

Suppose $G$ is a finite group. 
Define a discrete version of a  $p$-dimensional differential form on $\calL$ or simply a $p$-form
as a map $\theta \, : \, \calL_p\to G$.
We shall only consider $p=0,1$.
Let $\calF_p(G)$ be the set of all $G$-valued $p$-forms.  
Define a discrete version of the exterior  derivative operator $\der\, : \, \calF_0(G)\to \calF_1(G)$.
Namely, for any $0$-form $\theta$ define a $1$-form $\der\theta$ such that
\begin{align}
%\label{derivative}
\der\theta(e) = \theta(e^+)\theta(e^-)^{-1} \quad \mbox{for all $e\in \calL_1$}.
\end{align}
A $1$-form $\omega$ is called exact if  $\omega=\der\theta$ for some 
$0$-form $\theta$.

Consider a Hilbert space $\calH=(\CC^{|G|})^{\otimes |\calL_1|}$ with a $|G|$-dimensional qudit assigned to every edge $e\in \calL_1$.
Basis states of a qudit located on an edge $e\in \calL_1$ 
are denoted $|g\ra_e$, where $g\in G$.
Basis states of the full Hilbert space $\calH$ correspond to $1$-forms $|\omega\ra$ with $\omega\in \calF_1(G)$.  In Appendix~\ref{app:groundspacekitaev} we show that Kitaev's quantum double Hamiltonian
defined on the graph $\calL$ has a unique ground state  which is the uniform superposition of all exact $1$-forms, that is, 
\begin{align}
	\label{ground_state}
	|\Psi(G)\ra \sim \sum_{\theta \in \calF_0(G)}\; |\der\theta\ra. 
\end{align}

\subsection{State preparation circuit\label{sec:stateprepcircuit}}

In this section we assume that $G$ is a solvable group
and describe a local constant-depth adaptive circuit preparing the quantum double state
$|\Psi(G)\ra$ starting from a product state.

Since $G$ is solvable, it has 
a normal subgroup $H$ such that the factor group $G/H$ is cyclic,
that is, $G/H\cong \ZZ_d$ for some integer $d$. (In the following, elements of $\ZZ_d$ are labeled by integers $\{0,1,\ldots,d-1\}$.) 
Choose a group element $a\in G$ such that $G$ is a disjoint union of cosets
\begin{align}
	G = H \cup aH \cup a^2 H \cup \ldots \cup a^{d-1} H\ .\label{eq:disjointunioncosetsa}
\end{align}
This assumption implies that~$a^d\in H$, but in general, $a^d\ne \idG$. Let $s$ be the order of $a$ in the group $G$, that is, $s$ is the smallest nonzero integer such that $a^s = \idG$. To give an example where $s\neq d$, consider
the group $G = Q_8 = \{\idG, - \idG, i, -i, j, -j, k, -k\}$ of  quaternions with $H = \{\idG, i, - \idG, - i\}$ and $a = k$. In this case, $s = 4$ while $G/H \cong \mathbb{Z}_2$ and thus $d = 2$.

%Recall that if a finite group $G$ is solvable (as is the case with $S_3$), it has a normal subgroup $H$ such that $G/H \cong \mathbb{Z}_d$ for some integer $d$, and there is a group element $a \in G$ such that $G$ can be written as a disjoint union of cosets~\eqref{eq:disjointunioncosetsa}.
% The simpler procedure outlined below can be used whenever $\operatorname{ord}(a) =: s = d$. 
%Note that  in general~$s \geq d$ and we have 
%$s = 0 \mod d$ (since $(aH)^d$ is the trivial coset~$\idG H$ by assumption). The  inequality can be strict as can be seen by choosing $G = Q_8 = \{\idG, - \idG, i, -i, j, -j, k, -k\}$ to be the group of quaternions, $H = \{\idG, i, - \idG, - i\}$ and $a = k$. In this case, $s = 4$ while $G/H \cong \mathbb{Z}_2$ and thus $d = 2$. However, note that in general 

We describe an algorithm to prepare $|\Psi(G)\ra$ which immediately translates into a constant-depth adaptive circuit. It uses three quantum registers $\sysA,\sysB,\sysC$ such that 
\begin{enumerate}[(i)]
	\item Register $\sysA$ contains $|\calL_0|$ qudits of dimension $s$: to each vertex~$v\in\calL_0$, we associate a qudit~$\sysA_v\cong\mathbb{C}^s$.
	\item Register $\sysB$ contains $|\calL_1|$ qudits of dimension $d$: we associate a system~$\sysB_e\cong\mathbb{C}^{d}$ with
	each edge~$e\in\cL_1$.
	\item Register $\sysC$ contains $|\calL_1|$ qudits of  dimension $|G|$, a qudit~$\sysC_e\cong\mathbb{C}^{|G|}$ for every edge~$e\in\calL_1$. 
\end{enumerate}
The algorithm is given in Fig.~\ref{fig:preparationcircuit}. 
Figure~\ref{fig:preparationgroundstatesolvablegroup} illustrates the procedure using a planar square lattice with ``smooth'' boundaries on all sides, see Figure~\ref{fig:baselattice}.  
In step~\eqref{it:wgateapplication}, it uses the two-qudit gates  $U, W\, : \, \CC^s\otimes  \CC^{|G|} \to \CC^s\otimes \CC^{|G|}$ defined by 
\begin{align}\label{eq:defUWgates}
	U\left(|j\ra \otimes |g\ra\right) = |j\ra \otimes |a^j g\ra \qquad \text{ and } \qquad W\left(|j\ra \otimes |g\ra\right) = |j\ra \otimes |ga^{-j}\ra
\end{align}
for all $j\in \ZZ_s$ and $g \in G$, and in steps~\eqref{it:gateEstep} and~\eqref{it:gateEinversestep}, it uses a two-qudit gate $E \, : \, \CC^{|G|}\otimes  \CC^{d} \to \CC^{|G|}\otimes  \CC^{d}$ that satisfies
\begin{align} \label{eq:defEgate}
	E \ket{g} \otimes \ket{0} = \ket{g} \otimes \ket{j(g)} \quad 
\end{align}
for every $g\in G$, where $j(g) \in \mathbb{Z}_d$ is uniquely determined by the relation $g = a^{j(g)} h(g)$ with $h(g) \in H$ (see~\eqref{eq:disjointunioncosetsa}).

Our algorithm also uses constants~$\{\sigma_{v,e}\}_{(v,e)\in\cL_0\times\cL_1}\subset \{-1,0,+1\}$. These can be defined as follows, starting from a fixed orientation of the edges of~$\cL$. Arbitrarily choose a vertex~$r\in\cL_0$  designated as root, and for every vertex~$v\in \cL_0$, fix a  (directed) path~$\Gamma_v\subseteq \cL_1$ connecting~$r$ and~$v$. Then set
\begin{align}
	\sigma_{v,e}:=\begin{cases}
		1\qquad &\textrm{ if } e \in \Gamma_v \textrm{ and the orientation of~$e$ agrees with that of~$\Gamma_v$}\\
		-1\qquad &\textrm{ if } e \in \Gamma_v \textrm{ and the orientation of~$e$ disagrees with that of~$\Gamma_v$}\\
		0&\textrm{ if }e\not\in \Gamma_v\ ,
	\end{cases}
\end{align}
for every pair $(v,e)\in\cL_0\times\cL_1$.

We note that with these definitions, the value~$\theta(v)$ of a map~$\theta:\cL_0\rightarrow \mathbb{Z}_d$ (i.e., a $\mathbb{Z}_d$-valued $0$-form) at a vertex~$v\in \cL_0$ can be expressed as a linear combination of
differences of function values~(i.e., using an exact $1$-form)
along the path~$\Gamma_v$. More precisely, define
\begin{align}
	z(e)=\theta(e^+)-\theta(e^-)\in \ZZ_d\ 
\end{align}
for every edge~$e\in\cL_1$. Then we can express $\theta$ in terms of~$z$ as 
\begin{align} \label{eq:parallel_transport}
	\theta(v)=\theta(r) + \sum_{e\in \Gamma_v} \sigma_{v,e} z(e) {\pmod d}\ .
\end{align}
This ``parallel transport'' expression is an essential component of our construction.

\begin{figure}
	\begin{mdframed}[
		linecolor=black,
		linewidth=2pt,
		roundcorner=4pt,
		backgroundcolor=gray!15,
		userdefinedwidth=\textwidth,
		]
		\algnewcommand\algorithmicforeach{\textbf{for each}}
		\algdef{S}[FOR]{ForEach}[1]{\algorithmicforeach\ #1\ \algorithmicdo}
		\begin{algorithmic}[1]
			%    \Function{\fkl}{$3$-regular graph $G=(V, E)$, cut $C$}
			\State Initialize each qudit~$\sysA_v$, $v\in \cL_0$ in the state  $|+\ra\sim \sum_{i\in \ZZ_s} |i\ra$.   Initialize each qudit $\sysB_e$ for $e\in \cL_1$  in the state $|0\ra\in\mathbb{C}^d$. Initialize the register $\sysC$ in the state
			\[
			|\Psi(H)\ra \sim  \sum_{\phi \in \calF_0(H)} \; \bigotimes_{e\in E} |\der\phi(e)\ra_{\sysC_e},
			\] where elements of $H$ are identified with respective  elements of $G$. \label{it:initializationstep}
			\State For every edge $e\in \calL_1$, apply the gate $U$ to the qudits $\sysA_{e^+}$ and $\sysC_e$. For every edge $e\in \calL_1$, apply the gate $W$ to the qudits $\sysA_{e^-}$ and $\sysC_e$.  \label{it:controlledgatesstep} \label{it:wgateapplication}
			\State Measure each qudit $\sysA_v$, $v\in\cL_0$ in the eigenbasis of the generalized $s$-dimensional Pauli-$X$ operator. Let $m\in (\ZZ_s)^{\times |\calL_0|}$ be the measured outcome. \label{it:stepmeasurementstep}
			\State For every edge $e \in \calL_1$, apply the gate $E$ to each pair of qudits $\sysC_e$ and $\sysB_e$. \label{it:gateEstep}
			\State  For each edge $e \in \cL_1$, apply the gate
			\[
			Z^{\sum_{v\in \calL_0} k_v \sigma_{v,e}},
			\]
			to qudit~$\sysB_e$, where $Z=\sum_{j\in \ZZ_d} e^{2\pi i j/ d} |j\ra\la j|$ is the generalized $d$-dimensional
			Pauli-$Z$ operator  and where  $k_v \in \mathbb{Z}_d$ is such that 
			\begin{align}
				m_v = k_v \left( s / d \right) {\pmod s}\ .\end{align}
			\label{it:correctionstep}
			\State For every edge $e \in \calL_1$, apply the gate $E^{-1}$ to each pair of qudits $\sysC_e$ and $\sysB_e$. \label{it:gateEinversestep}
			
		\end{algorithmic}
	\end{mdframed}
	\caption{The algorithm (circuit) preparing~$\Psi(G)$ on system~$\sysC$. The definitions of the gates $U, V$ and $E$ are given in Eqs.~\eqref{eq:defUWgates} and~\eqref{eq:defEgate}. \label{fig:preparationcircuit}}
\end{figure}

\begin{figure}
	\subfigure[Step 1: Extended lattice and initial state used. The qudits on~$\sysC$ are in the ground state~$\Psi(H)$ 
	associated with $H$, 
	where we embed~$\mathbb{C}^{|H|}\cong\mathsf{span} \{\ket{h}\}_{h\in H}\subset\mathbb{C}^{|G|}$. \label{fig:initialstatesolvablepic} ]{\includegraphics[width=9cm]{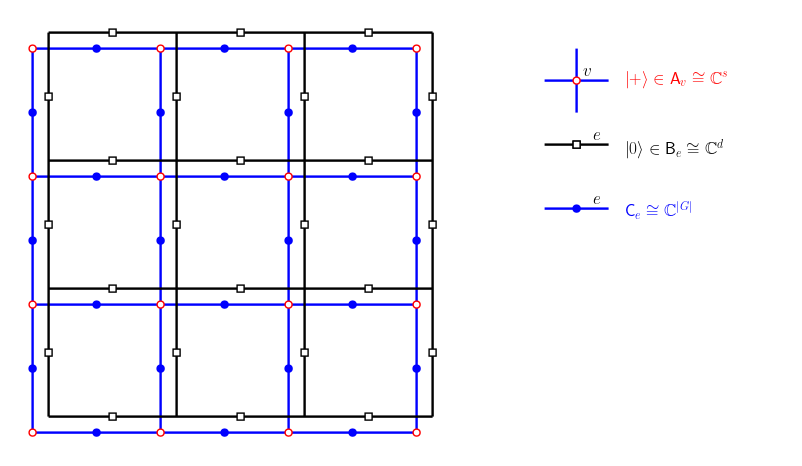}}\qquad
	\subfigure[Step 2: Application of the two-qudit gates $U$ and $W$.\label{fig:applicationUWLayersolvablepic}]{\includegraphics[width=9cm]{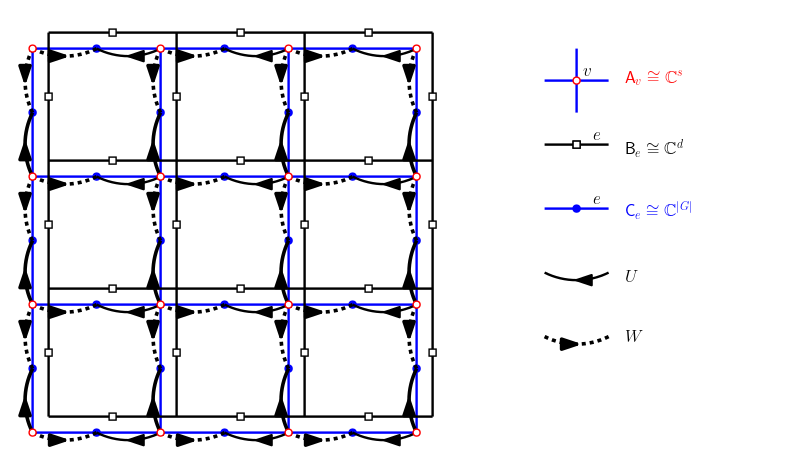}}\\
	\subfigure[Step 3: Each qudit in $\sysA$ is measured in the computational basis. \label{fig:measurementsolvablepic} ]{\includegraphics[width=9cm]{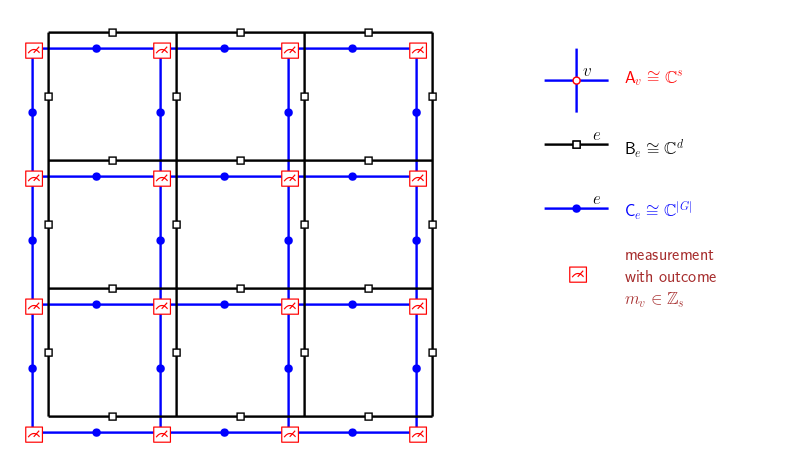}}\qquad
	\subfigure[Step 4: Application of the two-qudit gates $E$.  \label{fig:applicationELayersolvable}]{\includegraphics[width=9cm]{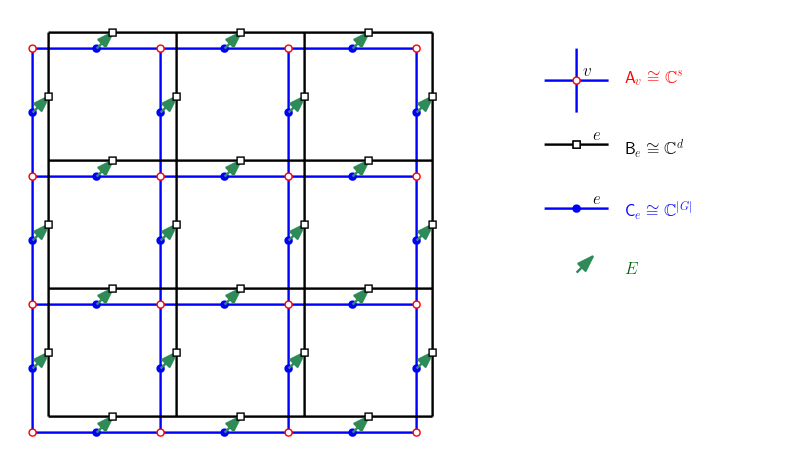}}\\
	\subfigure[Step 5: A correction operation is applied to each qudit in $\sysB$.  \label{fig:correctionsolvablepic}]{\includegraphics[width=9cm]{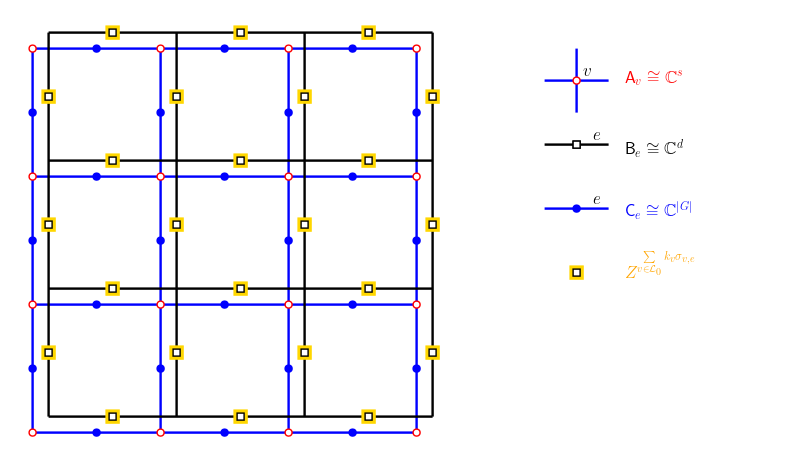}}\qquad
	\subfigure[Step 6: Application of the two-qudits gates $E^{-1}$. The qudits in $\sysA$ and $\sysB$ can be traced out. \label{fig:applicationEInverseLayersolvablepic}]{\includegraphics[width=9cm]{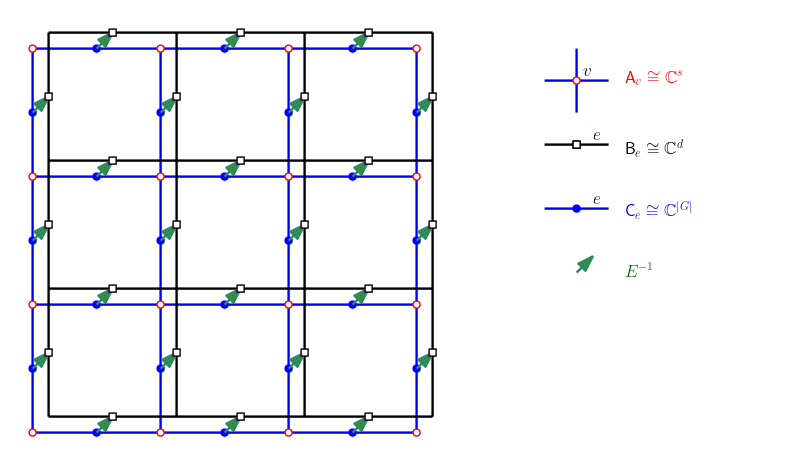}}
	\caption{Preparation of a ground state of the quantum double model associated with a solvable group using the procedure given in Figure~\ref{fig:preparationcircuit}. \label{fig:preparationgroundstatesolvablegroup}}
\end{figure}

\begin{theorem}\label{thm:mainpreparation}
	Let $G,H$ and $a\in G$ be as described above. 
	The algorithm given in Fig.~\ref{fig:preparationcircuit} prepares the state~$\Psi(G)$ (cf.~\eqref{ground_state})
	on register~$\sysC$.
\end{theorem}

Before proving Theorem~\ref{thm:mainpreparation}, we shall first rewrite the quantum double state $|\Psi(G)\ra$ in terms of $H$-valued and $\mathbb{Z}_s$-valued forms. 

\begin{lemma}
	\label{lem:exactforms}  
	The quantum double state for the group $G$ can be written as
	\begin{align}
		\label{eq:Gstate_via_H}
		|\Psi(G)\ra \sim \sum_{\theta \in \calF_0(\ZZ_s)}\; \sum_{\phi \in \calF_0(H)}\;
		\bigotimes_{e\in \calL_1} |a^{\theta(e^+)}\phi(e^+)\phi(e^{-})^{-1} a^{-\theta(e^-)} \ra.
	\end{align}
\end{lemma}
\begin{proof}
	By definition, we have 
	\[
	|\Psi(G)\ra \sim \sum_{\eta \in \calF_0(G)}\; \bigotimes_{e\in \calL_1} |\eta(e^+) \eta(e^{-})^{-1}\ra.
	\]
	Consider some vertex $v\in \calL_0$. Then $\eta(v)$ can be uniquely written as
	$\eta(v) = a^{\theta(v)} \phi(v)$ for some $\theta(v)\in \ZZ_d$ and $\phi(v)\in H$.
	Since the summation over $\eta(v) \in G$ is equivalent to independent summations
	over $\theta(v)\in \ZZ_d$ and $\phi(v)\in H$, we get
	\[
	|\Psi(G)\ra \sim \sum_{\theta\in \calF_0(\ZZ_d)}\; \sum_{\phi \in \calF_0(H)} \; 
	\bigotimes_{e\in \calL_1} |a^{\theta(e^+)} \phi(e^+) \phi(e^{-})^{-1} a^{-\theta(e^-)}\ra.
	\]
	Pick any  form $\omega \in \calF_0(\ZZ_{s/d})$
	and perform a change of variables $\phi(v) \gets a^{d\cdot \omega(v)} \phi(v)$
	in the summation over $\phi \in \calF_0(H)$. Note that the extra factor $a^{d\cdot \omega(v)}$
	belongs to $H$ for any $\omega$ since $a^d\in H$. Thus we get
	\[
	|\Psi(G)\ra \sim \sum_{\theta\in \calF_0(\ZZ_d)}\; \sum_{\phi \in \calF_0(H)} \; 
	\bigotimes_{e\in \calL_1} |a^{\theta(e^+)+d\cdot \omega(e^+)} \phi(e^+) \phi(e^{-})^{-1} a^{-\theta(e^-)-d\cdot \omega(e^-)}\ra
	\]
	for any form $\omega \in \calF_0(\ZZ_{s/d})$.  Symmetrizing this expression over all $\omega$'s one gets
	\[
	|\Psi(G)\ra \sim \sum_{\omega\in \calF_0(\ZZ_{s/d})} \; \sum_{\theta\in \calF_0(\ZZ_d)}\; \sum_{\phi \in \calF_0(H)} \; 
	\bigotimes_{e\in \calL_1} |a^{\theta(e^+)+d\cdot \omega(e^+)} \phi(e^+) \phi(e^{-})^{-1} a^{-\theta(e^-)-d\cdot \omega(e^-)}\ra.
	\]
	For any forms $\omega\in \calF_0(\ZZ_{s/d})$ and $\theta\in \calF_0(\ZZ_d)$ 
	define a form $\theta'\in \calF_0(\ZZ_s)$ such that 
	\[
	\theta'(v)= \theta(v)+d\cdot \omega(v) {\pmod s}
	\]
	for all $v\in \calL_0$.
	Clearly, the summation over all pairs $\omega\in \calF_0(\ZZ_{s/d})$ and  $\theta \in \calF_0(\ZZ_d)$ 
	is equivalent to the summation over all $\theta'\in \calF_0(\ZZ_s)$. Thus we get
	\[
	|\Psi(G)\ra \sim \sum_{\theta'\in \calF_0(\ZZ_s)} \;  \sum_{\phi \in \calF_0(H)} \; 
	\bigotimes_{e\in \calL_1} |a^{\theta'(e^+)} \phi(e^+) \phi(e^{-})^{-1} a^{-\theta'(e^-)}\ra.
	\]
	This is equivalent to  \eqref{eq:Gstate_via_H}.
\end{proof}

With expression~\eqref{eq:Gstate_via_H} for $\Psi(G)$, we can analyze the algorithm given in Fig.~\ref{fig:preparationcircuit} as follows. 
\begin{proof}[Proof of Theorem~\ref{thm:mainpreparation}:] 
	After the initialization step~\ref{it:initializationstep} the combined state of $\sysA\sysB\sysC$ has the form
	\[
	|\Psi_1\ra \sim \sum_{\theta \in \calF_0(\ZZ_s)} \; \sum_{\phi \in \calF_0(H)} \;
	\bigotimes_{v\in \calL_0} |\theta(v)\ra_{\sysA_v}
	\bigotimes_{e\in \calL_1}  |0\ra_{\sysB_e}\otimes 
	|\der\phi(e)\ra_{\sysC_e}.
	\]
	Here the subscripts $\sysA,\sysB,\sysC$ of a qudit state indicate which register holds the qudit. 
	
	After the application of the $U$ and $W$ gates in step~\ref{it:controlledgatesstep}, the combined state of $\sysA\sysB\sysC$ has the form
	\[
	|\Psi_2\ra \sim 
	\sum_{\theta \in \calF_0(\ZZ_s)} \; \sum_{\phi \in \calF_0(H)} \;
	\bigotimes_{v\in \calL_0} |\theta(v)\ra_{\sysA_v}
	\bigotimes_{e\in \calL_1}  |0 \ra_{\sysB_e}\otimes 
	|a^{\theta(e^+)} \phi(e^+) \phi(e^-)^{-1} a^{- \theta(e^-)}\ra_{\sysC_e}.
	\]
	Consider a  generalized Pauli operator
	$X\, : \, \CC^s \to \CC^s$ defined as 
	\begin{align}
		\label{eq:pauliX}
		X=\sum_{j\in \ZZ_s} |j+1 {\pmod s} \ra \la j|.
	\end{align}
	We claim that the state $\ket{\Psi_2}$ is an eigenvector of the Pauli operator $(X_{\sysA_v})^d$
	with the eigenvalue $+1$ for any vertex $v\in \calL_0$.
	Indeed, we have 
	\[
	(X_{\sysA_v})^d|\Psi_2\ra = 
	\sum_{\theta \in \calF_0(\ZZ_s)}\; 
	\sum_{\phi \in \calF_0(H)} \bigotimes_{w\in \calL_0} |\theta(w) + d\delta_{v,w} {\pmod s}\ra_{\sysA_w}  \bigotimes_{e\in \calL_1} \ket{0}_{\sysB_e}
	| a^{\theta(e^+)} \phi(e^+) \phi(e^{-})^{-1} a^{-\theta(e^-)}\ra_{\sysC_e}.
	\]
	Here $\delta_{v,w}=1$ if $v=w$ and $\delta_{v,w}=0$ otherwise.
	The sum over $\theta\in \calF_0(\ZZ_s)$ is unchanged if we perform a change of variables
	$\theta(v)\gets \theta(v)+d  {\pmod s}$ at the vertex $v\in \calL_0$. Thus
	\[
	(X_{\sysA_v})^d|\Psi_2\ra = 
	\sum_{\theta \in \calF_0(\ZZ_s)}\; 
	\sum_{\phi \in \calF_0(H)} \bigotimes_{w\in \calL_0} |\theta(w) \ra_{\sysA_w}  \bigotimes_{e\in \calL_1} \ket{0}_{\sysB_e}
	| a^{\theta(e^+)-d\delta_{v,e^+}} \phi(e^+) \phi(e^{-})^{-1} a^{-\theta(e^-)+d \delta_{v,e^-}}\ra_{\sysC_e}.
	\]
	The sum over $\phi\in \calF_0(H)$ is unchanged if we 
	perform a change of variables $\phi(v)\gets a^{d } \phi(v)$ at the vertex $v\in \calL_0$.
	Here we noted that $a^d\in H$, so this change of variables is well defined.
	Thus
	\[
	(X_{\sysA_v})^d|\Psi_2\ra = 
	\sum_{\theta \in \calF_0(\ZZ_s)}\; 
	\sum_{\phi \in \calF_0(H)} \bigotimes_{w\in \calL_0} |\theta(w) \ra_{A_w}  \bigotimes_{e\in \calL_1} \ket{0}_{\sysB_e}
	| a^{\theta(e^+)} \phi(e^+) \phi(e^{-})^{-1} a^{-\theta(e^-)}\ra_{\sysC_e}=|\Psi_2\ra.
	\]
	This confirms that $|\Psi_2\ra$ is an eigenvector of $(X_{\sysA_v})^d$ with the eigenvalue $+1$.  
	
	Following the measurement in step~\ref{it:stepmeasurementstep} with measurement outcome $m\in (\ZZ_s)^{\times |\calL_0|}$, the joint state of $\sysB\sysC$ has the form
	\[
	|\Psi_3(m)\ra \sim 
	\sum_{\theta \in \calF_0(\ZZ_s)}\; 
	\sum_{\phi \in \calF_0(H)} \; e^{-(2\pi i/s) \sum_{v\in \calL_0} m_v \theta(v)}  \bigotimes_{e\in \calL_1} \ket{0}_{\sysB_e}
	| a^{\theta(e^+)} \phi(e^+) \phi(e^{-})^{-1} a^{-\theta(e^-)}\ra_{\sysC_e}.
	\]
	
	Since $|\Psi_2\ra$ is  an eigenvector of the Pauli operator $(X_{A_v})^d$ with the eigenvalue $+1$, one must have
	$e^{(2\pi i/s)dm_v} = 1$ for any $v\in \calL_0$.
	This is only possible if $m_v$ is an integer multiple of $s/d$, that is,
	\begin{align}
		\label{eq:m_v}
		m_v = k_v(s/d) {\pmod s} \quad \mbox{for some $k_v\in \ZZ_d$}.
	\end{align}
	Thus 
	\[
	|\Psi_3(m)\ra \sim 
	\sum_{\theta \in \calF_0(\ZZ_s)}\; 
	\sum_{\phi \in \calF_0(H)} \; e^{-(2\pi i/d) \sum_{v\in \calL_0} k_v \theta(v)}  \bigotimes_{e\in \calL_1} \ket{0}_{\sysB_e}
	| a^{\theta(e^+)} \phi(e^+) \phi(e^{-})^{-1} a^{-\theta(e^-)}\ra_{\sysC_e}.
	\]
	By comparing this and \eqref{eq:Gstate_via_H}, one concludes
	that $|\Psi_3(m)\ra$ coincides with the desired quantum double state $|\Psi(G)\ra$
	modulo an unwanted phase factor
	$e^{-(2\pi i/d) \sum_{v\in \calL_0} k_v \theta(v)}$. 
	
	Following the application of the $E$-gate in step~\ref{it:gateEstep}, the joint state of $\sysB\sysC$ has the form 
	\[
	|\Psi_4(m)\ra \sim 
	\sum_{\theta \in \calF_0(\ZZ_s)}\; 
	\sum_{\phi \in \calF_0(H)} \; e^{-(2\pi i/d) \sum_{v\in \calL_0} k_v \theta(v)}  \bigotimes_{e\in \calL_1} |\theta(e^+) - \theta(e^-) {\pmod d}\ra_{\sysB_e}
	| a^{\theta(e^+)} \phi(e^+) \phi(e^{-})^{-1} a^{-\theta(e^-)}\ra_{\sysC_e} , 
	\]
	where we used the fact that \[
	j(a^{\theta(e^+)} \phi(e^+) \phi(e^{-})^{-1} a^{-\theta(e^-)})=\theta(e^+) - \theta(e^-) {\pmod d}
	\]
	for any forms  $\theta \in \calF_0(\ZZ_s)$ and $\phi\in \calF_0(H)$.

	Following the correction step~\ref{it:correctionstep}, the joint state of $\sysB\sysC$ has the form 
	\[
	|\Psi_5(m)\ra \sim 
	\sum_{\theta \in \calF_0(\ZZ_s)}\; 
	\sum_{\phi \in \calF_0(H)} \; e^{-(2\pi i/d) \sum_{v\in \calL_0} k_v \theta(r)}  \bigotimes_{e\in \calL_1} |\theta(e^+) - \theta(e^-) {\pmod d}\ra_{\sysB_e}
	| a^{\theta(e^+)} \phi(e^+) \phi(e^{-})^{-1} a^{-\theta(e^-)}\ra_{\sysC_e} 
	\]
	due to the fact that each term in $|\Psi_4(m)\ra$ with a fixed $\theta$ picks up an extra phase factor
	\[
	e^{ (2\pi i/d) \sum_{e\in \calL_1} \sum_{v\in \calL_0} k_v \sigma_{v,e} z(e)}
	=e^{ (2\pi i/d)  \sum_{v\in \calL_0} k_v (\theta(v) - \theta(r))}
	\]
	see~\eqref{eq:parallel_transport}. Next 
	we claim that 
	\begin{align}
		\label{eq:mvexpr}
		e^{(2\pi i/s) \sum_{v\in \calL_0} m_v} = 1 
	\end{align}
	for any measurement outcome $m$. 
	Indeed, it suffices to check that the state $|\Psi_2\ra$ is stabilized by the operator
	\[
	S= \prod_{v\in \calL_0} X_{\sysA_v}.
	\]
	Here $X\, : \, \CC^s \to \CC^s$ is defined in~\eqref{eq:pauliX}.
	Indeed, applying this operator
	to $|\Psi_2\ra$ is equivalent to a change of variable $\theta(v)\gets \theta(v)+1 {\pmod s}$
	for every vertex $v\in \calL_0$ in the summation over $\theta \in \calF_0(\ZZ_s)$.
	Thus we get
	\[
	S|\Psi_2\ra = 
	\sum_{\theta \in \calF_0(\ZZ_s)}\; 
	\sum_{\phi \in \calF_0(H)} \bigotimes_{v\in \calL_0} |\theta(v)\ra_{\sysA_v} \bigotimes_{e\in \calL_1} \ket{0}_{\sysB_e}
	| a^{\theta(e^+)-1} \phi(e^+) \phi(e^{-})^{-1} a^{-\theta(e^-)+1}\ra_{\sysC_e}.
	\]
	Perform a change of variables $\phi(v)\gets a \phi(v) a^{-1}$ for each $v\in \calL_0$.
	This change of variables is well-defined since $\phi$ is an $H$-valued form and $H$ is a normal subgroup,
	that is, $a H a^{-1}=H$.
	We get 
	\[
	S|\Psi_2\ra = 
	\sum_{\theta \in \calF_0(\ZZ_s)}\; 
	\sum_{\phi \in \calF_0(H)} \bigotimes_{v\in \calL_0} |\theta(v)\ra_{\sysA_v}  \bigotimes_{e\in \calL_1} \ket{0}_{\sysB_e}
	| a^{\theta(e^+)} \phi(e^+) \phi(e^{-})^{-1} a^{-\theta(e^-)}\ra_{\sysC_e}=|\Psi_2\ra,
	\]
	as claimed. Thus the measured eigenvalues $e^{(2\pi i/s) m_v}$ of the Pauli operators
	$X_{\sysA_v}$ must obey \eqref{eq:mvexpr}. Combining \eqref{eq:mvexpr} and \eqref{eq:m_v} one infers that 
	\[
	e^{-(2\pi i/d) \sum_{v\in \calL_0} k_v \theta(r)}=1
	\]
	and thus
	\[
	|\Psi_5(m)\ra \sim 
	\sum_{\theta \in \calF_0(\ZZ_s)}\; 
	\sum_{\phi \in \calF_0(H)} \;   \bigotimes_{e\in \calL_1} |\theta(e^+) - \theta(e^-) {\pmod d}\ra_{\sysB_e}
	| a^{\theta(e^+)} \phi(e^+) \phi(e^{-})^{-1} a^{-\theta(e^-)}\ra_{\sysC_e} .
	\]
	
	Finally, following the application of $E^{-1}$ in step~\ref{it:gateEinversestep}, the joint state of $\sysB\sysC$ has the form 
	\[
	|\Psi_6(m)\ra \sim 
	\sum_{\theta \in \calF_0(\ZZ_s)}\; 
	\sum_{\phi \in \calF_0(H)} \;   \bigotimes_{e\in \calL_1} |0\ra_{\sysB_e}
	| a^{\theta(e^+)} \phi(e^+) \phi(e^{-})^{-1} a^{-\theta(e^-)}\ra_{\sysC_e} .
	\]
	Discarding the qudits $\sysB_e$, one gets the desired state
	$|\Psi(G)\ra$ in the register $\sysC$, expressed as in \eqref{eq:Gstate_via_H}.

\end{proof}

We note that Steps~4,5,6 in our state preparation circuit can be combined
and all qudits $\sysB_e$ can be removed from the construction. Indeed, define a diagonal unitary gate 
$V\, : \, \CC^{|G|} \to \CC^{|G|}$ such that 
\[
V|a^j h\ra = e^{2\pi i j/d} |a^j h\ra \quad \mbox{for $j\in \ZZ_d$ and $h\in H$}.
\]
One can easily check that Steps~4,5,6 in Fig.~\ref{fig:preparationcircuit} are equivalent to applying the gate
\[
V^{\sum_{v\in \calL_0} k_v \sigma_{v,e}}
\]
to each qudit $\sysC_e$. Now the register $\sysB$ plays no role and can be removed.
This simplification also reduces the number of entangling gates used by the algorithm.
We chose to keep the register $\sysB$ since it simplifies the analysis of the algorithm.

Theorem~\ref{thm:mainpreparation} shows that the algorithm shown in Fig.~\ref{fig:preparationcircuit} reduces the problem of preparing~$\Psi(G)$ to that of preparing~$\Psi(H)$ (up to a constant-depth adaptive local circuit). It remains to show how to prepare the state $|\Psi(H)\ra$.  We show that this can be done recursively if the group~$G$ is solvable. This gives the main result of this section:

\begin{corollary}
	Let $G$ be a solvable group. Then there is a constant-depth adaptive local circuit which prepares~$\Psi(G)$.
\end{corollary}
\begin{proof}
	Since we assumed that $G$ is a solvable group,  it has a normal subgroup such that~$G/H$ is cyclic. The group~$H$ is either a cyclic group or  has a normal subgroup $H'$ such that the factor group
	$H/H'$ is cyclic. 
	In the former case one can view $|\Psi(H)\ra$ as a qudit version of the standard
	surface code and use the syndrome measurement circuit 
	followed by a generalized Pauli correction operator to prepare $|\Psi(H)\ra$
	by a local constant-depth adaptive circuit.
	In the latter case one can apply the state preparation circuit described  above
	to the state $|\Psi(H)\ra$  in a recursive manner.
\end{proof}

\section{Implementing group multiplication in solvable groups}
In this section, we consider various unitaries realizing group multiplication of several group elements.
Our constructions show that if the group~$G$ is solvable, then these operations can be implemented in constant depth adaptively irrespective of the number of group elements multiplied. Furthermore, all  involved quantum operations (i.e., auxiliary state preparations, unitaries and measurements) are $1D$-local, i.e., between nearest neighbor qudits when the qudits are arranged on a line. We say that the implementation is $1D$-local.

\subsection{Group multiplication and adaptive constant-depth local implementations}
Let $L\in \mathbb{N}$ be fixed. Let $\sysC_1,\ldots,\sysC_L,\sysC_{L+1}$ be $n+1$~systems each isomorphic to~$\mathbb{C}^{|G|}$.  We write $\sysC^L=\sysC_1\cdots \sysC_L\cong(\mathbb{C}^{|G|})^{\otimes L}$ and define the (left)
group multiplication unitary~$\multgate^{\Rightarrow}_{\sysC^L\rightarrow\sysC_{L+1}}$ by its action on basis states as
\begin{align}
\multgate^{\Rightarrow}_{\sysC^L\rightarrow\sysC_{L+1}} (\ket{g_1,\ldots,g_n}_{\sysC^L}\ket{g_{n+1}}_{\sysC_{L+1}})&=\ket{g_1,\ldots,g_n}_{\sysC^L}\ket{\left(\prod_{j=1}^n g_j\right)g_{n+1}}_{\sysC_{L+1}}\ 
\end{align}
where $\prod_{j=1}^n g_j=g_1\cdots g_n$.
We also define the left- and right multiplication gates
\begin{align}
  \multgate_{\sysC_{L+1}\rightarrow\sysC^L}^{\Rightarrow} (\ket{g_1,\ldots,g_n}_{\sysC^L}\ket{g_{n+1}}_{\sysC_{L+1}})&=\ket{g_{n+1}g_1,\ldots,g_{n+1}g_n}_{\sysC^L}\ket{g_{n+1}}_{\sysC_{L+1}}\\
\multgate_{\sysC_{L+1}\rightarrow\sysC^L}^{\Leftarrow} (\ket{g_1,\ldots,g_n}_{\sysC^L}\ket{g_{n+1}}_{\sysC_{L+1}})&=\ket{g_1g_{n+1},\ldots,g_ng_{n+1}}_{\sysC^L}\ket{g_{n+1}}_{\sysC_{L+1}}\ .
\end{align}
Let us also introduce a unitary $V_{\sysC^L}$ that acts as
\begin{align}
  V_{\sysC^L} \ket{y_1, \hdots, y_L} &= \ket{\hat{y}_1, \hdots, \hat{y}_L}\qquad\textrm{ where }\qquad  \hat{y}_j=\prod_{k=1}^j y_j\ .
\end{align}

In the following, we consider the implementation of these unitaries and their inverses. Let $S$ denote the single-qudit inversion gate
defined by
\begin{align}
S\ket{h}&=\ket{-h}\qquad\textrm{ for }\qquad h\in\mathbb{Z}_d\ ,
\end{align}
and let $\multgate^{\Rightarrow}_{A\rightarrow B}$ be the two-qudit multiplication gate defined by 
\begin{align}
\multgate^{\Rightarrow}_{A\rightarrow B}(\ket{g}_A\otimes\ket{h}_B)&=\ket{g}_{A}\otimes\ket{gh}_B\qquad\textrm{ for }\qquad g,h\in G\ .
\end{align}
 Then we have the following relations:
\begin{align}
  \begin{matrix}
  \multgate^{\Rightarrow}_{\sysC^L\rightarrow\sysC_{L+1}}&=&
V_{\sysC^L}^\dagger \multgate^{\Rightarrow}_{\sysC_L\rightarrow\sysC_{L+1}}
  V_{\sysC^L}\\
  \multgate_{\sysC_{L+1}\rightarrow\sysC^L}^{\Rightarrow} &=& S^{\otimes n+1}\multgate_{\sysC_{L+1}\rightarrow\sysC^L}^{\Leftarrow}  S^{\otimes n+1}\\
  (\multgate_{\sysC_{L+1}\rightarrow\sysC^L}^{\Leftarrow})^\dagger &=&S_{\sysC_{L+1}}\multgate_{\sysC_{L+1}\rightarrow\sysC^L}^{\Leftarrow} S_{\sysC_{L+1}} \ .
  \end{matrix}
  \label{eq:implementationcxsys}
  \end{align}
We will show below that the unitaries~$V_{\sysC^L}$, $V_{\sysC^L}^\dagger$ and $\multgate_{\sysC_{L+1}\rightarrow\sysC^L}^{\Leftarrow} $ can be implemented in constant adaptive depth with $1D$-local operations if the group~$G$ is solvable.
Together with~\eqref{eq:implementationcxsys}
and because $S$ is a single-qudit unitary and $\multgate^{\Rightarrow}_{\sysC_L\rightarrow\sysC_{L+1}}$ is a local $2$-qudit gate, this implies the following.
\begin{lemma}\label{lem:vcxgatestoconstruct}
  Let $G$ be a solvable group. Then the unitaries~$V_{\sysC^L},\multgate_{\sysC_{L+1}\rightarrow\sysC^L}^{\Leftarrow},  \multgate_{\sysC_{L+1}\rightarrow\sysC^L}^{\Rightarrow} $ and their inverses  can be implemented in adaptive constant depth
  with local operations when the systems are arranged as $(\sysC_1,\ldots,\sysC_{L+1})$ on a line.
\end{lemma}

A more involved gate that will be of interest is the following. Let $\sysA,\sysC_1,\ldots,\sysC_L,\sysB_1,\ldots,\sysB_L$ be systems with Hilbert space~$\mathbb{C}^{|G|}$. For $h\in G$, define the controlled-unitary~$U^h_{\sysC^L\rightarrow \sysB^L}$ by its action
\begin{align}
CU^h_{\sysC^L\rightarrow \sysB^L} (\ket{y_1,\ldots,y_L}\ket{x_1,\ldots,x_L})&=\ket{y_1,\ldots,y_L}\ket{hx_1,\hat{y}_1^{-1}h\hat{y}_1 x_2,\hat{y}_2^{-1}h\hat{y}_2 x_3,\ldots,
\hat{y}_{L-1}^{-1}h\hat{y}_{L-1} x_L}\ \label{eq:Cuhdefinition}
\end{align}
on computational basis states, where  $\hat{y}_j=\prod_{k=1}^j y_j$. We denote its controlled version by 
\begin{align}
  CCU_{\sysA\sysC^L\rightarrow \sysB^L} &=\sum_{h\in G}\proj{h}_\sysA\otimes CU^h_{\sysC^L\rightarrow\sysB^L}\ . \label{eq:controlledcuxi}
\end{align} 
Here we show that when the systems are organized as $\left((\sysB_1\sysC_1),\ldots,(\sysB_L,\sysC_L),\sysA\right)$ on a line,
then the unitary $CCU_{A\sysC^L\rightarrow \sysB^L}$ can be implemented using nearest neighbor operations. We call such operations $1D$-local. More precisely, we show the following:
\begin{lemma}\label{lem:constandepthimplementationccu}
  Let $G$ be a solvable group. Then  the unitary~$CCU_{\sysA\sysC^L\rightarrow\sysB^L}$ and its inverse~$CCU_{\sysA\sysC^L\rightarrow\sysB^L}^\dagger$
  can be implemented in adaptive constant depth with  $1D$-local gates and measurements.
  \end{lemma}
The unitary can be factored as 
\begin{align}
CCU_{\sysA\sysC^L\rightarrow \sysB^L}&=V_{\sysC^L}^\dagger CW_{\sysC^L\sysB^L}^\dagger \multgate^{\Rightarrow}_{A\rightarrow \sysB^L} CW_{\sysC^L\rightarrow \sysB^L} V_{\sysC^L} \label{eq:ccufactorization}
\end{align}
where the unitary $CW_{\sysC^L\rightarrow\sysB^L}$ that act as
\begin{align}
CW_{\sysC^L\rightarrow\sysB^L}(\ket{y_1,\ldots,y_L}\ket{x_1,\ldots,x_L})&=\ket{y_1,\ldots,y_L}\ket{x_1,y_1x_2,y_2x_3,\ldots,y_{L-1}x_L}\ .
 \end{align}
Identity~\eqref{eq:ccufactorization} is conveniently expressed by the circuit identity
%\[
%\begin{quantikz}[transparent]
%  \lstick{$\sysA$} &\qw & \ctrl{2} & \qw\\
%  \lstick{$\sysC^L$} &\qw & \control{1} & \qw\\
%  \lstick{$\sysB^L$} &\qw & \gate{U} & \qw
%\end{quantikz}=
%\begin{quantikz}[transparent]
%  \lstick{$\sysA$} &\qw	& \qw      & \qw              &\qw & \ctrl{2}  &\qw     %         & \qw&\qw &\qw \\ 
%  \lstick{$\sysC^L$}&\qw	& \gate{V} & \ctrl{1} &\qw & \qw       &\ctrl{1}% &\qw &\gate{V^\dagger}&\qw\\
%\lstick{$\sysB^L$}&\qw	& \qw      & \gate{W}        &\qw &  \gate{X}& \gate{W^\dagger}       &\qw &\qw &\qw  
%\end{quantikz}
%\]
\[
\begin{quantikz}[transparent]
  \lstick{$\sysC^L$} &\qw & \ctrl{1} & \qw\\
  \lstick{$\sysB^L$} &\qw & \gate{U} & \qw\\
  \lstick{$\sysA$} &\qw & \ctrl{-1} & \qw  
\end{quantikz}=
\begin{quantikz}[transparent]
  \lstick{$\sysC^L$}&\qw	& \gate{V} & \ctrl{1} &\qw & \qw       &\ctrl{1} &\qw &\gate{V^\dagger}&\qw\\
  \lstick{$\sysB^L$}&\qw	& \qw      & \gate{W}        &\qw &  \gate{X}& \gate{W^\dagger}       &\qw &\qw &\qw  \\
    \lstick{$\sysA$} &\qw	& \qw      & \qw              &\qw & \ctrl{-1}  &\qw              & \qw&\qw &\qw 
\end{quantikz}
\]

Eqs.~\eqref{eq:ccufactorization} implies that 
to show Lemma~\ref{lem:constandepthimplementationccu}, it is sufficient to exhibit $1D$-local adaptive constant-depth implementations for
the unitaries  $V_{\sysC^L}$, $V_{\sysC^L}^\dagger$,$CW_{\sysC^L\rightarrow\sysB^L}$, $CW_{\sysC^L\rightarrow\sysB^L}^\dagger$ as well as $\multgate^{\Rightarrow}_{\sysC_{L+1}\rightarrow\sysC^L}$  (which is identical to $\multgate^{\Rightarrow}_{\sysA\rightarrow\sysB^L}$ after renaming the systems) and
its adjoint. Assuming Lemma~\ref{lem:vcxgatestoconstruct} (whose proof we give below), all that is left to prove Lemma~\ref{lem:constandepthimplementationccu} is to see how to realize  $CW_{\sysC^L\rightarrow\sysB^L}$ and its inverse.

\subsubsection*{Implementability of  $CW_{\sysC^L\rightarrow\sysB^L}$ and its inverse (Proof of Lemma~\ref{lem:constandepthimplementationccu})}
The claim is immediate for $CW_{\sysC^L\rightarrow\sysB^L}$: This unitary can be factorized as
\begin{align}
  CW_{\sysC^L\rightarrow\sysB^L}&=\prod_{j=1}^{L-1} \multgate^\Leftarrow_{\sysC_j\rightarrow\sysB_{j+1}}\ ,\label{eq:CWfactorization}
\end{align}
where $\multgate^\Leftarrow$ is the controlled-group multiplication gate acting as 
\begin{align}
\multgate^\Leftarrow\left(\ket{y}\ket{x}\right)&=\ket{y}\ket{yx}\ ,
\end{align}
see Fig.~\ref{fig:Wfactorizationexample}. Since the controlled-group multiplication gates in Eq.~\eqref{eq:CWfactorization} act on disjoint subsets of qudits, they can be applied in parallel. In particular, this gives a $1D$-local (non-adpative) depth-$1$ circuit realizing~$CW_{\sysC^L\rightarrow\sysB^L}$. Identical reasoning applies to the inverse $CW_{\sysC^L\rightarrow\sysB^L}^\dagger$.
\begin{figure}
  \centering
  \begin{center}
    \[
\begin{quantikz}[transparent]
\lstick{\ket{y_1}}&\qw	& \ctrl{5} & \qw&\qw&\qw&\rstick{\ket{y_1}}\\ 
\lstick{\ket{y_2}}&\qw	& \qw      &\ctrl{5}&\qw&\qw&\rstick{\ket{y_2}}\\ 
\lstick{\ket{y_3}}  &\qw	& \qw      &\qw&\ctrl{5}&\qw&\rstick{\ket{y_3}}\\     
\lstick{\ket{y_4}}&\qw	& \qw      & \qw&\qw&\qw&\rstick{\ket{y_4}}\\
\lstick{\ket{x_1}}&\qw	& \qw      & \qw&\qw&\qw&\rstick{\ket{x_1}}\\
\lstick{\ket{x_2}}&\qw	& \gate{X}  &\qw &\qw&\qw&\rstick{\ket{y_1x_2}}\\
\lstick{\ket{x_{3}}}  &\qw	& \qw       &\gate{X}&\qw&\qw&\rstick{\ket{y_{2}x_3}}\\
\lstick{\ket{x_4}}&\qw	& \qw               &\qw  & \gate{X}&\qw&\rstick{\ket{y_3x_4}}\\
\end{quantikz}
\qquad\qquad
\begin{quantikz}[transparent]
\lstick{\ket{g}}&\qw	& \ctrl{1}& \qw &\rstick{\ket{g}}\\
\lstick{\ket{h}}&\qw    &\gate{X} & \qw & \rstick{\ket{gh}}
\end{quantikz}
\] 
\end{center}
  \caption{$1D$-local depth-$1$ circuit realizing $CW_{\sysC^L\rightarrow\sysB^L}$ illustrated for $L=4$, and the controlled-multiplication gate~$\multgate^\Leftarrow$ applied to $\ket{g}\ket{h}$. \label{fig:Wfactorizationexample}}
  \end{figure}
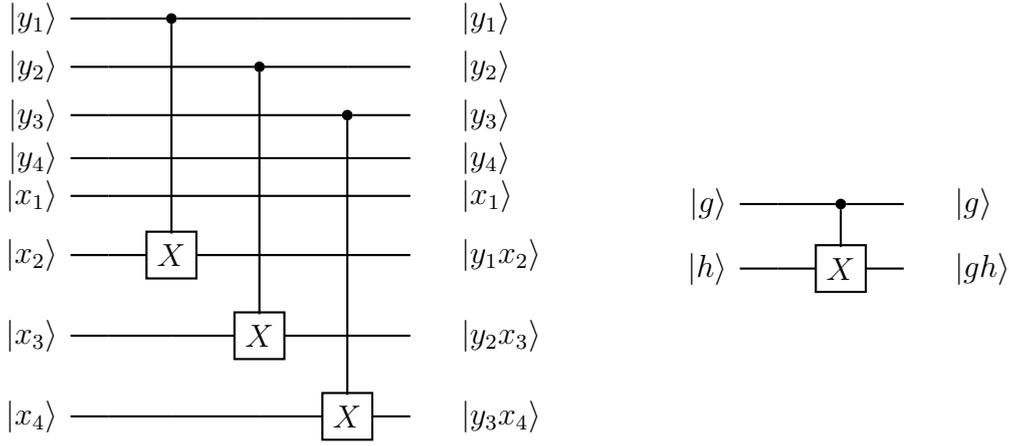

\subsubsection*{Implementability of $V_{\sysC^L}$, $V_{\sysC^L}^\dagger$ and $\multgate^{\Leftarrow}_{\sysC_{L+1}\rightarrow \sysC^L}$ (Proof of Lemma~\ref{lem:vcxgatestoconstruct})}
We can treat these  unitaries in a unified manner. Each of these acts as 
\begin{align}
U\ket{g_1,\ldots,g_n}&= \ket{\prod_{k=1}^n g_{k}^{A_{1,k}},\ldots, \prod_{k=1}^n g_{k}^{A_{n,k}}}\qquad\textrm{ for all }\qquad g_1,\ldots,g_n\in G\ .
\end{align}
Here $n\in \{L,L+1\}$ and $A\in \mathsf{Mat}_{n\times n}(\{-1,0,1\})$ determine the unitary under consideration, i.e., $U=U_A$. 
In fact, any matrix $A$ with the property that the map
\begin{align}
\Gamma^G_A:\left(g_1,\ldots,g_n\right)\mapsto (\hat{g}_1,\ldots,\hat{g}_n):=\left(\prod_{k=1}^n g_{k}^{A_{1,k}},\ldots, \prod_{k=1}^n g_{k}^{A_{n,k}}\right)\label{eq:gammgamap}
\end{align}
is a bijection (for any group~$G$) defines a unitary $U_A^G$ on $(\mathbb{C}^{|G|})^{\otimes n}$. It will be sufficient for our purposes to restrict to a certain subsets of such matrices. A vector $\sigma\in\{-1,0,1\}^n$ is called monotone (increasing) if the subsequence $(\sigma_{j_1},\ldots,\sigma_{j_{|\sigma|}})$ of non-zero entries of $\sigma$ switches signs at most once from $-1$~to $1$.
 We say that a matrix $A\in\mathsf{Mat}_{n\times n}(\{-1,0,1\})$ is  monotone if each of its rows is monotone. 
It is easy to check that the unitary $V_{\sysC^L}$
and its inverse are both of the form $U_A^G$  for a monotone matrix~$A$, see Fig.~\ref{fig:exampleWWdag} for an example.
\begin{figure}
  % elementary building blocks illustrated for $L=4$
% unitary V
  \[
  \begin{quantikz}[transparent]
\lstick{\ket{g_1}}&\qw	& \gate[4]{V} & \qw & \rstick{\ket{g_1}}\\ 
\lstick{\ket{g_2}}&\qw	&           & \qw &\rstick{\ket{g_1g_2}}\\ 
\lstick{\ket{g_3}}  &\qw	&  & \qw &\rstick{\ket{g_1g_2g_3}}\\     
\lstick{\ket{g_4}}&\qw	& & \qw & \rstick{\ket{g_1g_2g_3g_4}} 
\end{quantikz}\qquad\qquad\qquad\qquad
\begin{quantikz}[transparent]
\lstick{\ket{g_1}}&\qw	& \gate[4]{V^\dagger} & \qw & \rstick{\ket{g_1}}\\ 
\lstick{\ket{g_2}}&\qw	&                    & \qw &\rstick{\ket{g_1^{-1}g_2}}\\ 
\lstick{\ket{g_3}}  &\qw&                     & \qw &\rstick{\ket{g_2^{-2}g_3}}\\     
\lstick{\ket{g_4}}&\qw	&                   & \qw & \rstick{\ket{g_{3}^{-1}g_4}} 
\end{quantikz}
  \]
  \caption{For $L=4$, the monotonicity property for the unitaries~$V$ and $V^\dagger$ is evident from these expressions: We have $V=U^{G}_{A}$
and $V^\dagger=U^{G}_{\tilde{A}}$ for suitable monotone matrices~$A$.
    \label{fig:exampleWWdag}}
\end{figure}
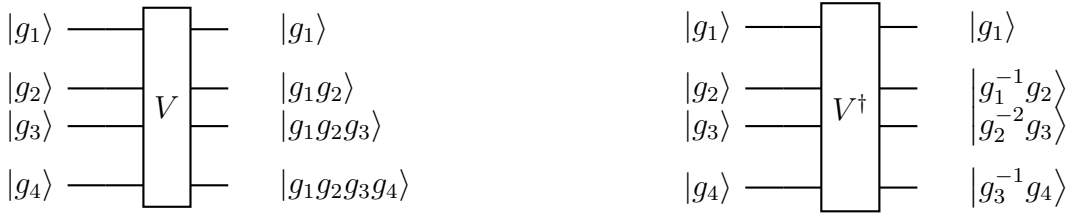
The same applies to $\multgate^{\Leftarrow}_{\sysC_{L+1}\rightarrow \sysC^L}$, see Fig.~\ref{fig:cxgatemultifanoutgate}.
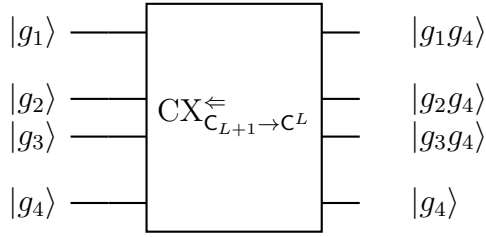
\begin{figure}
  \[
  \begin{quantikz}[transparent]
\lstick{\ket{g_1}}&\qw	& \gate[4]{\multgate^{\Leftarrow}_{\sysC_{L+1}\rightarrow\sysC^L}} & \qw & \rstick{\ket{g_1g_4}}\\ 
\lstick{\ket{g_2}}&\qw	&           & \qw &\rstick{\ket{g_2g_4}}\\ 
\lstick{\ket{g_3}}  &\qw	&  & \qw &\rstick{\ket{g_3g_4}}\\     
\lstick{\ket{g_4}}&\qw	& & \qw & \rstick{\ket{g_4}} 
\end{quantikz}
  \]
  \caption{For $L=3$, this expression shows the monotonicity property
    of the unitary~$\multgate^{\Leftarrow}_{\sysC_{L+1}\rightarrow \sysC^L}$.
        \label{fig:cxgatemultifanoutgate}}
\end{figure}

Now assume that $G$ is a group with a normal subgroup~$H$ such that $G/H$ is cyclic, i.e., $G/H\cong\mathbb{Z}_d$. As before,
we choose~$a\in G$ such that 
$G$ is a disjoint union of cosets generated by $aH$ as in Eq.~\eqref{eq:disjointunioncosetsa}. Each group element~$g\in G$ can then be written uniquely as
\begin{align}
	g&=a^j h\qquad\textrm{ for a pair }\qquad (j,h)\in\mathbb{Z}_d\times H\ .\label{eq:gajhdef}
\end{align}
In particular, we may define a (unitary) map
$E:\mathbb{C}^{|G|}\rightarrow\mathbb{C}^{|\mathbb{Z}_d|}\otimes\mathbb{C}^{|H|}$
which realizes the bijection $g\mapsto (j,h)$ given by~\eqref{eq:gajhdef}, i.e., it acts
\begin{align}
	E(\ket{g})&=\ket{j}\ket{h}\ ,
\end{align}
see Fig.~\ref{fig:Eisometrydef} for an illustration. Our key tool is the following reduction:

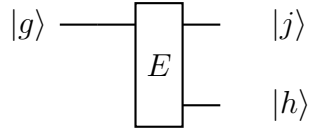
\begin{figure}
	\centering
	\begin{quantikz}[transparent]
		\lstick{\ket{g}}&\qw&\gate[2,nwires={2}]{E} &\qw&\rstick{\ket{j}}\\
		&   &                      &\qw& \rstick{\ket{h}}
	\end{quantikz}
	\caption{Definition of the unitary~$E:\mathbb{C}^{|G|}\rightarrow\mathbb{C}^{|\mathbb{Z}_d|}\otimes\mathbb{C}^{|H|}$\label{fig:Eisometrydef}}
\end{figure}

\begin{theorem}\label{thm:mainbarrington}
	Suppose $n\in\mathbb{N}$ and $A\in\mathsf{Mat}_{n\times n}(\{-1,0,1\})$ is monotone.
	Then
	\begin{align}
		U_A^G &=(E^\dagger)^{\otimes n}U_A^HU_A^{\mathbb{Z}_d} E^{\otimes n}\ ,
	\end{align}
	where $U_A^{\mathbb{Z}_d}$ and $U_A^H$ act on the spaces $(\mathbb{C}^{|\mathbb{Z}_d|})^{\otimes n}$ and $(\mathbb{C}^{|H|})^{\otimes n}$ with respect to the standard ordering, see Fig.~\ref{fig:implementationUAG}.
	
	In particular if $U_A^H$ and $U_A^{\mathbb{Z}_d}$ are realizable  in adaptive constant depth with~$1D$-local operations, then this is also the case for $U_A^G$.
\end{theorem}

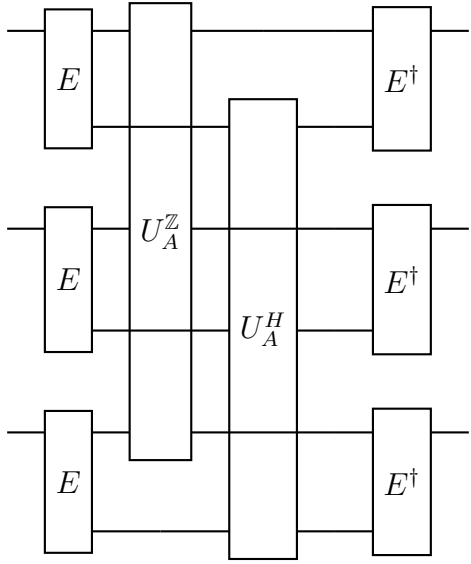
\begin{figure}
	\centering
	\begin{quantikz}[transparent]
		\qw&\gate[2,nwires={2}]{E}    & \gate[5]{U_A^{\mathbb{Z}}} & \qw            & \qw & \gate[2,nwires={2}]{E^\dagger} &\qw\\
		\qw&                          &  \linethrough          & \gate[5]{U_A^{H}}& \qw     &\qw                        &  \\
		\qw&\gate[2,nwires={2}]{E}    &                        & \linethrough    & \qw     & \gate[2,nwires={2}]{E^\dagger} &\qw\\
		\qw&                          &  \linethrough          &                 & \qw     & \qw                          & \\    
		\qw&\gate[2,nwires={2}]{E}    &                        & \linethrough    & \qw     & \gate[2,nwires={2}]{E^\dagger} &\qw\\
		\qw&                          &   \qw                  &                 & \qw     & \qw                          &  
	\end{quantikz}
	\caption{Implementation of the gate~$U_A^G$ based on $U_A^{\mathbb{Z}_d}$ and $U_A^{H}$ for $n=3$.\label{fig:implementationUAG}}
\end{figure}

The proof of Theorem~\ref{thm:mainbarrington} mirrors~\cite[Theorem 6]{barrington}, where it is shown that the word problem for any fixed solvable group~$G$ is $AC^0$~reducible to computing the function $k \mapsto (k \mod d)$. Here $d=|G|$.

Consider first a monotone vector~$\sigma=(\sigma_1,\ldots,\sigma_n)\in \{-1,0,1\}^n$. The monotonicity condition means that
the product $\prod_{j=1}^n g_j^{\sigma_j}$ has the special form
\begin{align}
	\prod_{m=1}^n g_m^{\sigma_m}&=(g_{p_1}^{-1}\cdots g_{p_r}^{-1})(g_{q_1}\cdots g_{q_t})\qquad\textrm{ with }\qquad p_1<p_2<\cdots <p_r<q_1<\cdots <q_t\ .\label{eq:hatgdefinition} 
\end{align}
(Either $r$ or $t$ could  vanish.)

For simplicity, we first prove Theorem~\ref{thm:mainbarrington}
	in the special case where order of the group element~$a\in G$ and the subgroup~$H\subset G$ (cf.~\eqref{eq:gajhdef}) satisfy the identity
	\begin{align}
		\mathop{ord}(a)&=|G/H|\ ,\label{eq:ordergroupelementassumption}
	\end{align}
	see the discussion in Section~\ref{sec:statepreparation}. 
	We subsequently comment on how assumption~\eqref{eq:ordergroupelementassumption}
	can be lifted. 
\begin{lemma}\label{lem:reductionHgroup}
	Let $\sigma=(\sigma_1,\ldots,\sigma_n)\in \{-1,0,1\}^n$ be monotone and $g_1,\ldots,g_n\in G$  arbitrary. Write each~$g_m\in G$ as $g_m=a^{j_m}h_m$ as in~\eqref{eq:gajhdef}. Define
	\begin{align}
		\hat{j}_m&=\sum_{\ell=1}^m \sigma_m j_m \in\mathbb{Z}_d\qquad\textrm{ for }\qquad m\in \{1,\ldots,n\}\label{eq:hjsum}
	\end{align}
	and
	\begin{align}
		\hat{h}_m &=a^{\hat{j}_{m-1}}h_m a^{-\hat{j}_{m-1}}\in H\qquad\textrm{ for }\qquad m\in \{1,\ldots,n\}\ .\label{eq:hjuncationmapajm}
	\end{align}
	with the convention that $\hat{j}_0=0$. Then
	\begin{align}
		g_1^{\sigma_1}\cdots g_n^{\sigma_n}&=(\hat{h}_1^{\sigma_1}\cdots \hat{h}_n^{\sigma_n})a^{\hat{j}_n}\ .
	\end{align}
\end{lemma}
Since the expressions~$\sum_{\ell=1}^m \sigma_m j_m $ and $\hat{h}_1^{\sigma_1}\cdots \hat{h}_n^{\sigma_n}$
simply correspond to taking the analogous product of $(j_1,\ldots,j_m)$ and $(\hat{h}_1,\ldots,\hat{h}_n)$ in the groups~$\mathbb{Z}_d$ and $H$, respectively, Lemma~\ref{lem:reductionHgroup} implies that for a monotone matrix~$A$, the map~$\Gamma^G_A$ (Eq.~\eqref{eq:gammgamap})
can be rewritten in terms of the maps~$\Gamma^{\mathbb{Z}_d}_A$ and~$\Gamma^H_A$, respectively. This involves using the conjugation map~\eqref{eq:hjuncationmapajm} and can immediately be translated into a  circuit identity as shown in Fig.~\ref{fig:implementationUAG}. This implies  Theorem~\ref{thm:mainbarrington} in the case where~\eqref{eq:ordergroupelementassumption} is satisfied.

We note that the  assumption~\eqref{eq:ordergroupelementassumption} 
	can be dropped with a small modification of the statement of Lemma~\ref{lem:reductionHgroup}. 
	If $s=ord(a)$ is the order of~$G$, then the sum~\eqref{eq:hjsum} should be taken modulo~$s$ instead of~$d$. (Correspondingly, one should use the gate~$U_A^{\mathbb{Z}_s}$.)  Writing~$j=k\cdot d+(j\pmod d)$ for some $k\in\mathbb{Z}_{s/d}$, 
	one can  then use that $a^j=(a^d)^k a^{j\pmod d}$ is a product of an element~$\tilde{h}:=(a^d)^k\in H$ and a power $a^{j'}$ with $j'\in\mathbb{Z}_d$. The factor~$\tilde{h}$ can be absorbed using an additional multiplication gate in the last register  corresponding to~$h_n$.  This leads to a (slightly more involved) constant-depth circuit for $U_A^G$ using the unitaries~$U_A^{\mathbb{Z}_s}$ and~$U_A^H$.

\begin{proof}[Proof of Lemma~\ref{lem:reductionHgroup}: ]
	Consider an expression of the form~\eqref{eq:hatgdefinition}. Slightly abusing notation to avoid clutter, let us write $g_{p_\ell}=a^{j_\ell}h_\ell$ and $g_{q_\ell}=a^{j'_\ell}h'_\ell$. Then
	\begin{align}
		g_{p_1}^{-1}\cdots g_{p_r}^{-1}&=h_1^{-1}a^{-j_1}h_2^{-1}a^{-j_2}h_3^{-1}\cdots h_r^{-1}a^{-j_r}\\
		&=h_1^{-1}(a^{-j_1}h_2^{-1}a^{j_1})(a^{-j_1-j_2}h^{-1}_3 a^{j_1+j_2})\cdots (a^{-\sum_{\ell=1}^{r-1}j_\ell}h^{-1}_r a^{\sum_{\ell=1}^{r-1}j_\ell}) a^{-\sum_{\ell=1}^{r}j_\ell}
	\end{align}
	Defining
	\begin{align}
		\hat{j}_m&=-\sum_{\ell=1}^m j_\ell\qquad\textrm{ for }\qquad m\in \{1,\ldots,r\}\ 
	\end{align}
	and
	\begin{align}
		\hat{h}_m&=\begin{cases}
			h_m\qquad&\textrm{ if }m=1\\
			a^{\hat{j}_{m-1}}h_m a^{-\hat{j}_{m-1}}&\textrm{ for }m\in\{2,\ldots,r\}\ ,
		\end{cases}
	\end{align}
	this can be written in the form
	\begin{align}
		g_{p_1}^{-1}\cdots g_{p_r}^{-1}&=\left(\hat{h}_1^{-1} \hat{h}_2^{-1}\cdots
		\hat{h}_{r-1}^{-1}\right)\cdot a^{\hat{j}_r}\ .\label{eq:intermediateproductexpression}
	\end{align}
	By analogous reasoning we have
	\begin{align}
		g_{q_1}\cdots g_{q_t}&=\left(\prod_{m=1}^t \tilde{h}'_m \right)\cdot \left(a^{\sum_{m=1}^t j_m}\right)\qquad\textrm{ with }\qquad \tilde{h}'_m=a^{\sum_{\ell=1}^m j_\ell}\tilde{h}_m a^{-\sum_{\ell=1}^m j_\ell}\ .
	\end{align}
	We have
	\begin{align}
		a^{\hat{j}_r}g_{q_1}\cdots g_{q_t}&=\left(\prod_{m=1}^t  a^{\hat{j}_r}\tilde{h}'_m a^{-\hat{j}_r} \right)\cdot \left(a^{\hat{j}_r+\sum_{m=1}^t j_m}\right) \\
		&=\left(\prod_{m=1}^t \hat{h}'_m \right)\cdot a^{\hat{j}'_t}\label{eq:ajrqonetoqs}
	\end{align}
	where
	\begin{align}
		\hat{h}'_m&=a^{\hat{j}'_m}\tilde{h}'_m a^{-\hat{j}'_m}\\
		\hat{j}'_m&=-\sum_{\ell=1}^r j_\ell+\sum_{\ell=1}^m j'_\ell\ 
	\end{align}
	for $m\in \{1,\ldots,t\}$. Combining~\eqref{eq:intermediateproductexpression} with~\eqref{eq:ajrqonetoqs} we conclude that
	\begin{align}
		g_{p_1}^{-1}\cdots g_{p_r}^{-1} g_{q_1}\cdots g_{q_t}&=
		\left(\prod_{m=1}^r \hat{h}_{m}^{-1}\right)
		\left(\prod_{m=1}^t \hat{h}_{m}\right)a^{\hat{j}'_t}\ .    
	\end{align}
	The claim follows from this.
\end{proof}

Theorem~\ref{thm:mainbarrington}
implies that to implement $U^G_A$, it suffices to consider how to implement $U^{\mathbb{Z}_d}_A$ and $U^{H}_A$. If $H$ is not cyclic, we can apply Theorem~\ref{thm:mainbarrington} to $G=H$. Proceeding recursively, obtain the following.

\begin{corollary}\label{cor:cdepthimplementations}
	Let $A\in\mathsf{Mat}_{n\times n}(\{-1,0,1\})$ be monotone. Let $G$ be a solvable group. 
	Suppose that  for any $d\in\mathbb{N}$, we have a adaptive constant-depth $1D$-local implementation of $U_A^{\mathbb{Z}_d}$. Then
	$U_A^G$ can be implemented in adaptive constant depth with $1D$-local operations.
\end{corollary}

Recall that our goal is to show that
\begin{enumerate}[(i)]
\item the group multiplication gate~$\multgate^{\Leftarrow}_{\sysC_{L+1}\rightarrow \sysC^L}$,
\item the unitary $V_{\sysC^L}$ and  its inverse
\end{enumerate}
have constant-depth adaptive $1D$-local implementations.  These unitaries correspond to monotone matrices.
According to Corollary~\ref{cor:cdepthimplementations}, it suffices to show that the following unitaries (acting on products of $\mathbb{C}^{|\mathbb{Z}_d|}$) have constant-depth implementations, for any $d\in\mathbb{N}$ (and $n=L$). Here additions/substractions are computed modulo~$d$.
\begin{enumerate}[(i)]
%\item
%  The multiple addition gate~$\multgate_{A^n\rightarrow A_{n+1}}$ defined by
%  \begin{align}
%\multgate_{A^n\rightarrow A_{n+1}}\ket{j_1,\ldots,j_{n},j_{n+1}}&=\ket{j_1,\ldots,j_n,\sum_{k=1}^{n+1}j_k}\qquad\textrm{ for all}\qquad j_1,\ldots,j_{n+1}\in\mathbb{Z}_d\ .
%    \end{align}
  \item
 The {\em partial summation gate} $U$ defined by 
\begin{align}
U \ket{j_1, \hdots, j_n} = \ket{j_1, j_1 + j_2, \ldots , \sum\limits_{k = 1}^n j_k}\qquad\textrm{ for }\qquad j_1,\ldots,j_n\in\mathbb{Z}_d\ .
\end{align}
This corresponds to~$V_{\sysC^L}$.
\item
  The {\em successive difference gate} $\Delta$ defined by
  \begin{align}
\Delta \ket{j_1, \hdots, j_n} = \ket{j_1, j_2 - j_1,j_3-j_2, \ldots , j_n-j_{n-1}}\qquad\textrm{ for }\qquad j_1,\ldots,j_n\in\mathbb{Z}_d\ ,
  \end{align}
  which corresponds to~$V_{\sysC^L}^\dagger$.
\item The multitarget $\multgate$ gate on $n+1$ qubits  defined by
  \begin{align}
\multgate_{A_{n+1}\rightarrow A^n} (\ket{x_1,\ldots,x_n}_{A^n}\ket{h}_{A_{n+1}})&=\ket{x_1+h,\ldots,x_n+h}_{A^n}\ket{h}_{A_{n+1}}\qquad\textrm{ for }\qquad h,j_1,\ldots,j_n\in\mathbb{Z}_d
  \end{align}
  corresponding to $\multgate^\Leftarrow_{\sysC_{L+1}\rightarrow \sysC^L}$. 
\end{enumerate}

\subsection{Adaptive constant-depth $1D$-local circuits for $U$, $\Delta$ and $\multgate_{A_{n+1}\rightarrow A^n}$}
Here we show that the circuits
$U$, $\Delta$ and $\multgate_{A_{n+1}\rightarrow A^n}$ (all defined in terms of the group $\mathbb{Z}_d$) 
  can be implemented with adaptive constant-depth $1D$-local circuits. We exploit that the unitaries of interest are Cliffords and can
therefore be realized using gate teleportation~\cite{gottesmanchuang99}. This is briefly reviewed in Section~\ref{sec:cliffcircuitsgate}.

\subsubsection{Clifford circuits for $G = \mathbb{Z}_d$ and gate teleportation\label{sec:cliffcircuitsgate}}
\label{sec:Zd}
 Consider the cyclic abelian group $G = (\mathbb{Z}_d, +)$ for $d \geq 2$. Recall the definition of the generalized Pauli operators 
\begin{align}
	X = \sum\limits_{j \in \mathbb{Z}_d} \ketbra{j + 1}{j} \qquad \text{ and } \qquad Z = \sum\limits_{j \in \mathbb{Z}_d} \omega_d^j \ketbra{j}{j} \quad ,
\end{align}
where $\omega_d = e^{2 \pi i / d}$.  It is straightforward to check that 
\begin{align}
	XZ = \omega_d^{-1} ZX \qquad \text{ and thus }\qquad \left( X^r Z^s \right) \left( X^{r'} Z^{s'} \right) = \omega_d^{r' \cdot s - r \cdot s'} \left( X^{r'} Z^{s'} \right) \left( X^r Z^s \right) \label{eq:XZcommutationrelations}
\end{align}
for all $r, s, r', s' \in \mathbb{Z}_d$. Furthermore, the operators $\{X^r Z^s\}_{(r, s) \in \mathbb{Z}_d^2}$ are orthonormal with respect to the normalized Hilbert-Schmidt inner product $\langle A, B \rangle_{HS} = \frac{1}{d} \tr \left( A^\dagger B \right)$.  For an $n$-qudit system, we set
\begin{align}
X^rZ^s:=X^{r_1}Z^{s_1}\otimes \cdots\otimes X^{r_n}Z^{s_n}\qquad\textrm{ for all }\qquad r=(r_1,\ldots,r_n), s=(s_1,\ldots,s_n)\in\mathbb{Z}_d^n\ .
\end{align}
We again have orthonormality with respect to the (normalized) Hilbert-Schmidt inner product, i.e., 
\begin{align}
\tr((X^rZ^s)^\dagger X^{r'}Z^{s'})=d^{n}\delta_{r,r'}\delta_{s,s'}\qquad\textrm{ for all }\qquad r,s,r',s'\in\mathbb{Z}_d^n\ .\label{eq:orthogonalityxrzs}
\end{align}
Let $\ket{\Phi} = \frac{1}{\sqrt{d}} \sum\limits_{j \in \mathbb{Z}_d} \ket{j} \ket{j}\in\mathbb{C}^d\otimes\mathbb{C}^d$  be the maximally entangled state.
 Eq.~\eqref{eq:orthogonalityxrzs} implies that the vectors
\begin{align}
\ket{\Phi^{(r,s)}}_{A^nB^n}:=(X^rZ^s\otimes I)\ket{\Phi}^{\otimes n}_{AB}\qquad\textrm{ for }\qquad r,s\in\mathbb{Z}_d^n
\end{align}
form an orthonormal basis~$\{\Phi^{(r,s)}\}_{(r,s)\in(\mathbb{Z}_d^n)^2}$ of $((\mathbb{C}^d)^{\otimes n})^{\otimes 2}$.

The $n$-qudit Pauli group is generated by the generalized Pauli operators~$\{X^rZ^s\}_{r,s\in\mathbb{Z}_d^n}$. An $n$-qudit Clifford is a unitary which preserves this group under conjugation. For example, the two-qudit gate~$\multgate = \sum\limits_{j \in \mathbb{Z}_d} \ketbra{j}{j} \otimes X^j$  is Clifford since
\begin{align} 
\begin{matrix}
	\multgate \left( X \otimes I \right) \multgate^\dagger &=& X \otimes X \quad  \\
	\multgate \left( I \otimes X \right) \multgate^\dagger &= &I \otimes X \quad  \\
	\multgate \left( Z \otimes I \right) \multgate^\dagger &= &Z \otimes I \quad  \\
	\multgate \left( I \otimes Z \right) \multgate^\dagger &= &Z^{-1} \otimes Z\ .
	\end{matrix}\label{eq:conjcnot}
\end{align}
To describe the gate teleportation technique, let $U$ be an $n$-qudit Clifford unitary. Let
\begin{align}
\ket{\Phi_U}_{A^{n}B^n}:=(I_{A^n}\otimes U_{B^n})\ket{\Phi}^{\otimes n}\ .
\end{align}
Consider the following protocol, starting from an initial state of the form~$\ket{\Psi}_{C^n}\otimes \ket{\Phi_U}_{A^nB^n}$.
\begin{enumerate}[(i)]
	\item
	  Perform a von Neumann measurement in the orthonormal basis
          $\{\Phi^{(r,s)}\}_{(r,s)\in\mathbb{Z}_d^2}$  on the pair of systems~$C^nA^n$, getting an outcome $(r,s)\in\mathbb{Z}_D^2$.
	\item
	  Apply the operator $U^\dagger (X^rZ^s) U$ to system~$B^n$.\label{eq:udaggerxrzs}
\end{enumerate}
This procedure maps $\ket{\Psi}_{C^n}$ to $(U\ket{\Psi})_{B^n}$ using the resource state~$\Phi_U$. Because $U$ is Clifford, the  operation applied in step~\eqref{eq:udaggerxrzs}  is a Pauli operator, i.e., a depth-one circuit of single-qudit gates. In particular, if the resource state~$\Phi_U$ can be created in adaptive constant depth, then this yields an adaptive constant-depth implementation of~$U$.

We note here that determining the Pauli operator applied in step~\eqref{eq:udaggerxrzs} of this procedure  involves an (efficient) classical computation which may or may not be local in general. Here we focus on the locality of the quantum operations (i.e., measurements and unitaries) and consider non-local (efficient) classical processing to be free.

\subsubsection{Partial Summation Gate} \label{subs:partSum}
We will first consider the partial summation unitary~$U$. This unitary can be decomposed as 
\begin{align}
	U_{B^n} = \multgate_{B_{n-1} \to B_n} \cdots \multgate_{B_2 \to B_3} \multgate_{B_1 \to B_2}
\end{align}
as illustrated in Figure~\ref{fig:partialsum}. 
\begin{figure}
  \[
    \begin{quantikz}[transparent]
    		\lstick{$\ket{j_1}$} &\qw &\ctrl{1} &\qw & \qw & \qw & \qw & \qw & \qw &  \\
		\lstick{$\ket{j_2}$} &\qw & \gate{X} &\qw & \ctrl{1} & \qw & \qw & \qw & \qw & \\
		\lstick{$\ket{j_3}$} &\qw &\qw & \qw & \gate{X} & \qw & \ctrl{1} & \qw & \qw & \\
		\lstick{$\ket{j_4}$} &\qw &\qw & \qw & \qw & \qw & \gate{X} & \qw & \qw & \\
\end{quantikz}\]
	\caption{The partial summation gate $U$ for $n=4$.}
	\label{fig:partialsum}
\end{figure}
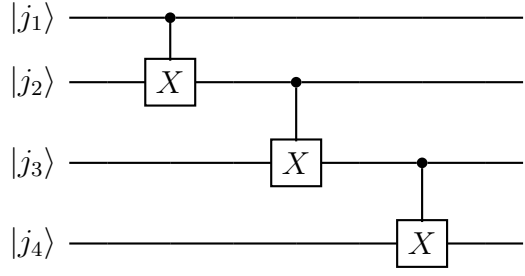
The state $\ket{\Phi_U}_{A^nB^n}$ is a stabilizer state with generators
\begin{align} 
  \begin{matrix}
    T_k^X&=&
    X_{A_k}\prod_{\ell=k}^n X_{B_\ell}\qquad\textrm{ for }k=1,\ldots,n\\
  T_k^Z&=&\begin{cases}
    Z_{A_1}^{-1} Z_{B_1}\qquad&\textrm{ for }k=1\\
    Z_{A_k}^{-1}Z_{B_{k-1}}^{-1}Z_{B_k} &\textrm{ for }k=2,\ldots,n\ .
  \end{cases}
  \end{matrix}\label{eq:stabpartsum}
 \end{align}
 Indeed, this follows because
 the maximally entangled state~$\ket{\Phi}$ is (up to a phase) the unique simultaneous $+1$ eigenstate of $X \otimes X$ and $Z^{-1} \otimes Z$, and hence
$\ket{\Phi}_{AB}^{\otimes n}$ is a stabilizer state with stabilizer generators
\begin{align}
	S_k^X := X_{A_k} X_{B_k} \qquad \text{ and } \qquad \hat{S}_k^Z := Z_{A_k}^{-1} Z_{B_k} \qquad \text{ for } k=1,\ldots,n\ .
\end{align}
Expression~\eqref{eq:stabpartsum} thus follows by application of~\eqref{eq:conjcnot} with the definition
\begin{align} 
\begin{matrix}
T_k^X&:= &(I_{A^n} \otimes U_{B^n} ) S_k^X (I_{A^n} \otimes U_{B^n} )^\dagger\\
T_k^Z&:=& (I_{A^n} \otimes U_{B^n} ) S_k^Z (I_{A^n} \otimes U_{B^n} )^\dagger\ .
\end{matrix} \label{eq:tkxzpartialsum}
\end{align}
Eq.~\eqref{eq:stabpartsum} suggests the following procedure for preparing~$\ket{\Psi_U}_{A^nB^n}$. 
\begin{enumerate}
\item Prepare the $2n$-qudit state $\ket{+}^{\otimes 2n}$
  where $\ket{+}:=\frac{1}{\sqrt{d}}\sum_{j\in\mathbb{Z}_d}\ket{j}$. 
	This state is stabilized by all operators $T^X_k$, $k=1,\ldots,n$. 
	
	\item Measure the $n$ stabilizer generators $T_k^Z$ from~\eqref{eq:stabpartsum}. Denote the measurement result by $(\omega_d^{r_1}, \hdots, \omega_d^{r_n})$ with $r_k \in \mathbb{Z}_d$ for $k = 1, \hdots, n$.\label{it:steptwoprepproc}
	\item Apply the unitary $\prod_{k=1}^{n}X_{A_k}^{r_k}$ to the post-measurement state.\label{it:laststepcorrection}
\end{enumerate}
By definition and since each of the  operators $\{T_k^Z\}_{k=1}^n$ is geometrically $2$- respectively $3$-local in $1D$, this is a constant-depth adaptive $1D$-local circuit. Note that even the classical processing is local here: the correction applied to system~$A_k$ in step~\eqref{it:laststepcorrection} only depends on the measurement result~$r_k$.
\begin{lemma} \label{lem:prodPhi}
	The procedure above produces the state $\ket{\Phi_U}$.
\end{lemma}

\begin{proof}
  The post-measurement state $\ket{\Psi}$ after the measurements in step~\eqref{it:steptwoprepproc} satisfies
	\begin{align}
	T_k^X \ket{\Psi} = \ket{\Psi} \qquad \text{ and } \qquad T_k^Z \ket{\Psi} = \omega_d^{r_k} \ket{\Psi}\qquad  \text{ for } k=1,\ldots,n\ .
	\end{align}
        The claim thus follows from the fact that (by Eq.~\eqref{eq:XZcommutationrelations}) $X_{A_k}$ commutes with  all operators~\eqref{eq:stabpartsum} except for $T_k^Z$ and
        \begin{align}
T_k^Z X_{A_k}^{-1}&=\omega X_{A_k}^{-1}T_k^Z\ ,
        \end{align}
        that is, $T_k^Z X_{A_k}^{-(d-1)}=\omega^{-1} X_{A_k}^{-1}T_k^Z$ or  
        $T_k^Z X_{A_k}=\omega^{-1} X_{A_k}T_k^Z$
\end{proof}
Given~$\ket{\Phi_U}$, one can then use gate teleportation to realize~$U$.

\subsubsection{Implementation of the successive difference gate}
Now consider the successive difference gate~$\Delta$. It has the decomposition
\begin{align}
\Delta_{B^n} &=\multgate^{-1}_{B_1\rightarrow B_2}\cdots \multgate^{-1}_{B_{n-2}\rightarrow B_{n-1}}\multgate^{-1}_{B_{n-1}\rightarrow B_{n}}
\end{align}
as illustrated in Fig.~\ref{fig:successivedifference}. 
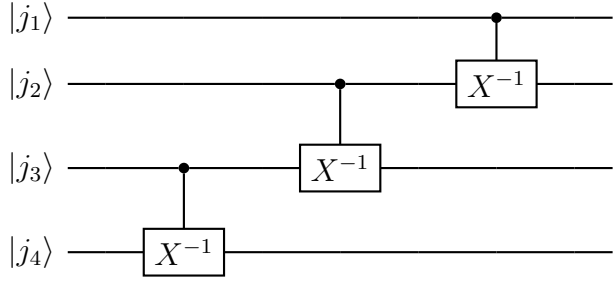
\begin{figure} 
  \[
  \begin{quantikz}[transparent]
    		\lstick{$\ket{j_1}$} &\qw &\qw      &\qw & \qw & \qw & \ctrl{1} & \qw & \qw &  \\
		\lstick{$\ket{j_2}$} &\qw & \qw     &\qw & \ctrl{1} & \qw & \gate{X^{-1}} & \qw & \qw & \\
		\lstick{$\ket{j_3}$} &\qw &\ctrl{1} & \qw & \gate{X^{-1}} & \qw & \qw & \qw & \qw & \\
		\lstick{$\ket{j_4}$} &\qw &\gate{X^{-1}} & \qw & \qw & \qw & \qw & \qw & \qw & \\
\end{quantikz}\]
  
%	\Qcircuit @C=0.5em @R=1em {

	        %}                  \]
	\caption{The successive difference gate~$\Delta$ for $n=4$.}
	\label{fig:successivedifference}
\end{figure}
With
\begin{align} 
\begin{matrix}
	\multgate^\dagger \left( X \otimes I \right) \multgate &=& X \otimes X^{-1} \quad  \\
	\multgate^\dagger \left( I \otimes X \right) \multgate &= &I \otimes X \quad  \\
	\multgate^\dagger \left( Z \otimes I \right) \multgate &= &Z \otimes I \quad  \\
	\multgate^\dagger \left( I \otimes Z \right) \multgate &= &Z \otimes Z\ .
	\end{matrix}\label{eq:conjcnotdagger}
\end{align}
it is straightforward to check that the state $\ket{\Phi_\Delta}_{A^nB^n}$ in this case is a stabilizer state with stabilizer generators
\begin{align} 
  \begin{matrix}
  T_k^X&=&\begin{cases}
X_{A_k}X_{B_k}X_{B_{k+1}}^{-1}\qquad&\textrm{ for }k=1,\ldots,n-1\\
X_{A_n}X_{B_n}   \qquad&\textrm{ for }k=n\\
  \end{cases}\\
  T_k^Z&=& Z_{A_k}^{-1}\prod_{j=1}^{k-1}Z_{B_j}\qquad\textrm{ for }k=1,\ldots,n\ .
  \end{matrix}\label{eq:stabpartsumb}
 \end{align}
Eq.~\eqref{eq:stabpartsumb} implies that  the following procedure for preparing~$\ket{\Psi_\Delta}_{A^nB^n}$. 
\begin{enumerate}
\item Prepare the $2n$-qudit state $\ket{0}^{\otimes 2n}$. 
	\item Measure the $n$ stabilizer generators $T_k^X$ from~\eqref{eq:stabpartsum}. Denote the measurement result by $(\omega_d^{r_1}, \hdots, \omega_d^{r_n})$ with $r_k \in \mathbb{Z}_d$ for $k = 1, \hdots, n$.\label{it:steptwoprepproc}
	\item Apply the unitary $\prod_{k=1}^{n}Z_{A_k}^{r_k}$ to the post-measurement state.
\end{enumerate}
Again, the corresponding circuit is adaptive and local in~$1D$. Once the state~$\ket{\Phi_\Delta}$ has been generated, gate teleportation can be used to realize~$\Delta$.

\subsubsection{Implementation of multitarget~$\multgate$ gate and its inverse}
By~\eqref{eq:implementationcxsys}, the group multiplication
$\multgate_{A^n\rightarrow A_{n+1}}$ can be implemented
        in constant adaptive depth with local operations using
        corresponding implementations of $V_{A^n}$ and $V^\dagger_{A^L}$, as well as the two-qudit gate $\rightarrow$.  This is because
        \begin{align}
\multgate_{A^n\rightarrow A_{n+1}}&= V_{A^n}^\dagger \multgate_{A_{n}\rightarrow A_{n+1}} V_{A^n}\ .\label{eq:cxananplusone}
        \end{align} Eq.~\eqref{eq:cxananplusone} also implies that 
        (because $\multgate_{A_{n}\rightarrow A_{n+1}}$ is a two-qudit gate)      its inverse~
        $\multgate_{A^n\rightarrow A_{n+1}}^\dagger$ can be implemented in constant adaptive depth. (We note that  slightly more efficient implementations could be obtained by applying gate teleportation directly, but we do not optimize constants here.)

Consider the multitarget gate~$\multgate_{A_{n+1}\rightarrow A^n}$. This can be written as a product of $n$~two qudit unitaries as follows:
\begin{align}
\multgate_{A_{n+1}\rightarrow A^n}&=\prod_{j=1}^n \multgate_{A_{n+1}\rightarrow A_j}
\end{align}
see Fig.~\ref{fig:multitargetcxgate}. 
In other words, this is similar to the multiplication gate
\begin{align}
\multgate_{A^n\rightarrow A_{n+1}}&=\prod_{j=1}^n \multgate_{A_j\rightarrow A_{n+1}}
\end{align}
(see Fig.~\ref{fig:groupmultiplicationgate}) but with source and target of each controlled gate~$\multgate$ switched. To exploit this similarly, let 
 $F = \frac{1}{\sqrt{d}} \sum\limits_{k, l \in \mathbb{Z}_d} \omega_d^{k \cdot l} \ketbra{k}{l}$
denote the Fourier transform over~$\mathbb{Z}_d$. A straightforward computation gives 
\begin{align}
	\left( F \otimes F^\dagger \right) \multgate_{A_1\rightarrow A_2} \left( F \otimes F^\dagger \right)^\dagger = \multgate_{A_2\rightarrow A_1}\
\end{align}
which immediately implies that
\begin{align}
\multgate_{A_{n+1}\rightarrow A^n}&
=(F^{\otimes n}_{A^n}\otimes F^\dagger_{A_{n+1}}) \multgate_{A^n\rightarrow A_{n+1}}(F^{\otimes n}_{A^n}\otimes F^\dagger_{A_{n+1}})^\dagger\ .\label{eq:multitargetreexpression}
\end{align}
We conclude from Eq.~\eqref{eq:multitargetreexpression} that the multitarget gate~$\multgate_{A_{n+1}\rightarrow A^n}$  can also be implemented in constant adaptive depth with local operations.
\begin{figure}
  \[
    \begin{quantikz}[transparent]
    		\lstick{$\ket{j_1}$} &\qw &\gate{X} &\qw & \qw & \qw & \qw & \qw & \qw &  \\
		\lstick{$\ket{j_2}$} &\qw & \qw &\qw & \gate{X} & \qw & \qw & \qw & \qw & \\
		\lstick{$\ket{j_3}$} &\qw &\qw & \qw & \qw & \qw & \gate{X} & \qw & \qw & \\
		\lstick{$\ket{j_4}$} &\qw &\ctrl{-3} & \qw & \ctrl{-2} & \qw & \ctrl{-1} & \qw & \qw & \\
\end{quantikz}\]
	\caption{The multitarget $\multgate_{A_{n+1}\rightarrow A^n}$ gate for $n=3$.}	\label{fig:multitargetcxgate}
\end{figure}
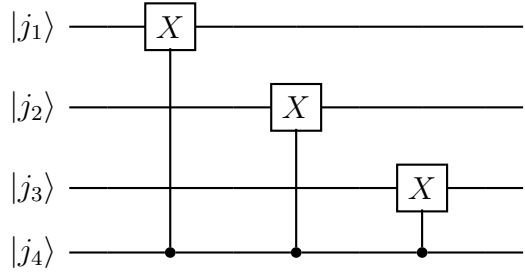

\begin{figure}
  \[
    \begin{quantikz}[transparent]
    		\lstick{$\ket{j_1}$} &\qw &\ctrl{3} &\qw & \qw & \qw & \qw & \qw & \qw &  \\
		\lstick{$\ket{j_2}$} &\qw & \qw &\qw & \ctrl{2} & \qw & \qw & \qw & \qw & \\
		\lstick{$\ket{j_3}$} &\qw &\qw & \qw & \qw & \qw & \ctrl{1} & \qw & \qw & \\
		\lstick{$\ket{j_4}$} &\qw &\gate{X} & \qw & \gate{X} & \qw & \gate{X} & \qw & \qw & \\
\end{quantikz}\]
	\caption{The multiplication gate $\multgate_{A^n\rightarrow A_{n+1}}$ gate for $n=3$.}	\label{fig:groupmultiplicationgate}
\end{figure}
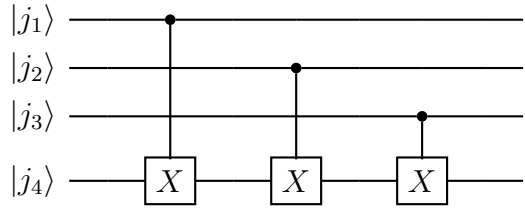

\section{Adaptive implementation of anyonic ribbon operators\label{sec:anyonicribbonimplementation}}
Recall from~\eqref{eq:fxicpi} that the operator~$F_\xi^{(C,\pi)}$ associated with an open ribbon~$\xi$ is 
 an element of the family 
\begin{align}
\{F_\xi^{(C,\pi);({\bf u},{\bf v})}\}_{({\bf u},{\bf v})}\ \label{eq:Cpiuvrandom}
\end{align}
of operators, namely that with ${\bf u}_0=(1,1)$ and ${\bf v}_0=(1,1)$.  Here we show the following:
\begin{theorem}\label{thm:implementability}
  Let $G$ be a solvable group and let $(C,\pi)$ be an arbitrary anyon label.
The operator $F_\xi^{(C,\pi)}$  can be implemented deterministically by a  local adaptive constant-depth circuit for any open ribbon~$\xi$ with the circuit having support only along the ribbon.
  \end{theorem}
  Theorem~\ref{thm:implementability} shows that
initial states for (topological) quantum computation can be 
generated with constant-depth circuits (for a solvable group); similarly, anyons can be moved an extensive distance in constant depth.

We show Theorem~\ref{thm:implementability} in two steps.
The main technical step is the construction of a  constant-depth adaptive circuit
that implements one of the operators $F_\xi^{(C,\pi);({\bf u},{\bf v})}$ of
the family~\eqref{eq:Cpiuvrandom} according to a certain distribution over pairs $({\bf u},{\bf v})$ (and simultaneously provides the pair $({\bf u}, {\bf v})$), see Lemma~\ref{lem:probabilisticimplementability} below.

Let us argue here that this suffices to implement $F_\xi^{(C,\pi)}$ deterministically. This is a consequence of the fact that the parameters $({\bf u},{\bf v})$ only determine local degrees of freedom. Consider the states~$\ket{(C,\pi);({\bf u},{\bf v})}$ defined by \eqref{eq:cpiuvstate}. We show the following:
\begin{lemma}\label{lem:localmappingunitary}
For any ${\bf u}, {\bf v}, {\bf u'},{\bf v'}$, there are
local unitaries $U_{s_0}$ and $U_{s_1}$ localized at the two endpoints~$s_0$ and~$s_1$ of~$\xi$ such that 
\begin{align}
U_{s_0}U_{s_1}\ket{(C,\pi);({\bf u},{\bf v})}&=
\ket{(C,\pi);({\bf u'},{\bf v'})}\ .\label{eq:usosone}
\end{align}
\end{lemma}
Lemma~\ref{lem:localmappingunitary}
is a refinement of~\eqref{eq:conversionlocaldegrees}, which guarantees the existence of local operators~$d^{{\bf u'}}_{{\bf u}}$, $d^{{\bf v'}}_{{\bf v}}$ whose product~$d^{{\bf u'}}_{{\bf u}}d^{{\bf v'}}_{{\bf v}}$
has the same action  on a state of the form~$\ket{(C,\pi);({\bf u},{\bf v})}$ (see \cite[Eq.~(23)]{BombinDelgado}). However, the operators constructed in~\cite{BombinDelgado} are generally not unitary. Lemma~\ref{lem:localmappingunitary} guarantees the existence of  map which is a product of unitaries.

 Lemma~\ref{lem:localmappingunitary} implies that if we can obtain the state $\ket{(C,\pi);({\bf u},{\bf v})}$ by some means, e.g., using the circuit of discussed in Lemma~\ref{lem:probabilisticimplementability}, then we can use local unitaries to create the state~$F^{(C,\pi)}_\xi \Psi$. Conversely, with the implementation of~$F_\xi^{(C,\pi)}$ as stated in the Theorem~\ref{thm:implementability}, one can realize
any state from the family~$\{F_\xi^{(C,\pi);({\bf u},{\bf v})}\Psi\}_{({\bf u},{\bf v})}$ using Lemma~\ref{lem:localmappingunitary}.

\begin{proof}
Since the claim
is that there are unitary maps~$U_{s_0}$, $U_{s_1}$ with property~\eqref{eq:usosone}, it suffices without loss of generality to consider the case where~${\bf u}$, ${\bf v}$ are arbitrary and ${\bf u'}={\bf v'}=(1,1)$. 

Let us again write~${\bf u}=(i,j)$ and ${\bf v}=(i',j')$.
We proceed in two steps:
\begin{enumerate}[(i)]
\item\label{it:unitaryconstructionfirst}
We show that there are local unitaries $V_{s_0}$ and 
$V_{s_1}$ such that
\begin{align}
V_{s_0} V_{s_1}\ket{(C,\pi);\left((i,j),(i',j')\right)}
&=\ket{(C,\pi);\left((1,j),(1,j')\right)}\ .
\end{align}
This follows from~\cite[Eq.~(D14)]{BombinDelgado} which states
 (because of our convention that $p_1=1$) that
\begin{align}
A_{s_0}^{p_i}A_{s_1}^{p_{i'}}\ket{(C,\pi);\left((i,j),(i',j')\right)}
&=\ket{(C,\pi);\left((1,j),(1,j')\right)}\ ,
\end{align}
together with the observation that the operators $A_{s_0}^{p_i}$ and $A_{s_1}^{p_{i'}}$ are unitary.
\item\label{it:unitaryconstructionsecond}
Next we show that there  are unitaries $W_{s_0}, W_{s_1}$ such that 
\begin{align}
W_{s_0}\ket{(C,\pi);\left((1,j),(1,j')\right)}
&=\ket{(C,\pi);\left((1,1),(1,j')\right)}\qquad\textrm{ and }\qquad \\
W_{s_1}\ket{(C,\pi);\left((1,1),(1,j')\right)}
&=\ket{(C,\pi);\left((1,1),(1,1)\right)}\ .
\end{align}
We give the argument for $W_{s_0}$ (the argument for $W_{s_1}$ is analogous), i.e., our goal is to show that there is   a local unitary $W_{s_0}$ that maps $\Psi_1$ to $\Psi_2$ where \begin{align}
\ket{\Psi_1}&=\ket{(C,\pi);\left((1,j),(1,j')\right)}\\
\ket{\Psi_2}&=\ket{(C,\pi);\left((1,j),(1,1)\right)}\ .
\end{align}
Following~\cite[Eq.~(D8)]{BombinDelgado}, there is a
local operator $a_{s_0}:=a^{j',1}_{s_0}$ supported near~$s_0$ such that 
\begin{align}
\begin{matrix}
\Psi_2&=&a_{s_0}\Psi_1\\
a_{s_0}^\dagger a_{s_0} \Psi_1&=&\Psi_1\ ,
\end{matrix}\label{eq:psionetworeduced}
\end{align}
see Lemma~\ref{lem:labelchangeops} below for details. 
Let $S$ be the support of~$a_{s_0}$.  Eq.~\eqref{eq:psionetworeduced} implies that the reduced density operators
$\rho_j:= \tr_{S}\proj{\Psi_j}$ for $j=1,2$ obtained when tracing out the qudits belonging to~$S$ are identical. The existence of a unitary $W_{s_0}$ supported on~$S$ and satisfying $W_{s_0}\Psi_1=\Psi_2$ now follows from Uhlmann's theorem.
\end{enumerate}
Combining~\eqref{it:unitaryconstructionfirst} and~\eqref{it:unitaryconstructionsecond}, we conclude that the unitaries $U_{s_0}:=W_{s_0}V_{s_0}$ and $U_{s_1}:=W_{s_1}V_{s_1}$ are supported near~$s_0$ and~$s_1$, respectively, and 
satisfy
\begin{align}
U_{s_0}U_{s_1}\ket{(C,\pi);\left((i,j),(i',j')\right)}&=
\ket{(C,\pi);\left((1,1),(1,1)\right)}
\end{align}
as claimed.
\end{proof}
It remains to show the existence of an operator~$a_{s_0}$ satisfying~\eqref{eq:psionetworeduced}.

\begin{lemma}\label{lem:labelchangeops}
Consider the operators
\begin{align}
a^{x,y}_{s_0}&=\frac{d_\pi}{|E(C)|}\sum_{k\in E(C)}\Gamma^{-1}_\pi(k)_{x,y} A_{s_0}^{k}\ .
\end{align}
Then 
\begin{align}
a^{x,y}_{s_0}\ket{(C,\pi);\left((1,j),(1,j')\right)}&=
\delta_{x,j'}\ket{(C,\pi);\left((1,j),(1,y)\right)}\label{eq:firstclaimaxyszero}\\
(a^{x,y}_{s_0})^\dagger a^{x,y}_{s_0}\ket{(C,\pi);\left((1,j),(1,j')\right)}&=\delta_{x,j'}\ket{(C,\pi);\left((1,j),(1,j')\right)}\ .
\label{eq:secondlcaimaxyzero}
\end{align}
\end{lemma}

\begin{proof}
We have 
with the commutation relations~\eqref{eq:startpointcommutex}
\begin{align}
a^{x,y}_{s_0}F^{(C,\pi);\left((1,j),(1,j')\right)}&=
\frac{d_\pi^2}{|E(C)|^2}
\sum_{k,\ell\in E(C)}\Gamma^{-1}_\pi(k)_{x,y} 
\Gamma^{-1}_\pi(\ell)_{j,j'} 
A_{s_0}^{k}F_\xi^{r_C^{-1},\ell}\\
&=\frac{d_\pi^2}{|E(C)|^2}
\sum_{k,\ell\in E(C)}\Gamma^{-1}_\pi(k)_{x,y} 
\Gamma^{-1}_\pi(\ell)_{j,j'} 
F_\xi^{r_C^{-1},k\ell}A_{s_0}^{k}
\end{align}
where we used that $k r_C^{-1}k^{-1}=r_C^{-1}$ for $k\in E(C)$. Since a ground state~$\Psi$ is stabilized by each operator~$A^k_{s_0}$, we have (substituting $k$ by $k^{-1}$)
\begin{align}
a^{x,y}_{s_0}F^{(C,\pi);\left((1,j),(1,j')\right)}
\Psi &=\frac{d_\pi^2}{|E(C)|^2}
\sum_{k,\ell\in E(C)}\Gamma_\pi(k)_{x,y} 
\overline{\Gamma_\pi(\ell)_{j',j}} 
F_\xi^{r_C^{-1},k^{-1}\ell}\Psi\\
&=\frac{d_\pi^2}{|E(C)|^2}
\sum_{k,m\in E(C)}\Gamma_\pi(k)_{x,y} 
\overline{\Gamma_\pi(km)_{j',j}} 
F_\xi^{r_C^{-1},m}\Psi\\
&=\frac{d_\pi^2}{|E(C)|^2}
\sum_{r=1}^{d_\pi}\sum_{k,m\in E(C)}\Gamma_\pi(k)_{x,y} 
\overline{\Gamma_\pi(k)_{j',r}}\overline{\Gamma_\pi(m)_{r,j}} 
F_\xi^{r_C^{-1},m}\Psi\\
&=\delta_{x,j'}\frac{d_\pi}{|E(C)|}
\sum_{m\in E(C)}\overline{\Gamma_\pi(m)_{y,j}} 
F_\xi^{r_C^{-1},m}\Psi\\
&=\delta_{x,j'}\frac{d_\pi}{|E(C)|}
\sum_{m\in E(C)}\Gamma^{-1}_\pi(m)_{j,y} 
F_\xi^{r_C^{-1},m}\Psi\ .
\end{align}
In summary, this shows that 
\begin{align}
a^{x,y}_{s_0}F^{(C,\pi);\left((1,j),(1,j')\right)}
&=\delta_{x,j'}
F^{(C,\pi);\left((1,j),(1,y)\right)}\Psi\ ,
\end{align}
which is the first claim~\eqref{eq:firstclaimaxyszero}.

It is easy to verify that
\begin{align}
(a^{x,y}_{s_0})^\dagger &=a^{y,x}_{s_0}\  .
\end{align}
With~\eqref{eq:firstclaimaxyszero}, this implies
the claim~\eqref{eq:secondlcaimaxyzero}.
\end{proof}

\subsection{Probabilistic implementation of a single operator}
For $h,g\in G$, a ``basic'' ribbon operator~$F_\xi^{h,g}$ associated with a ribbon~$\xi$ of length~$L$ acts on
qubits along the ribbon as
\begin{align}
F_\xi^{h,g}(\ket{y_1,\ldots,y_L}\ket{x_1,\ldots,x_L})&=\delta_{g,y_1\cdots y_L} 
\ket{y_1,\ldots,y_L}\ket{hx_1,\hat{y}_1^{-1}h\hat{y}_1 x_2,\hat{y}_2^{-1}h\hat{y}_2 x_2,\ldots,
\hat{y}_{L-1}^{-1}h\hat{y}_{L-1} x_L}
\end{align}
where $\hat{y}_k:=\prod_{j=1}^k y_j$, see Fig.~\ref{fig:ribbonoperatoraction}.
It will be convenient to factor this as
\begin{align}
  F^{h,g}_\xi&=P^g_{\sysC^L}CU^h_{\sysC^L\rightarrow\sysB^L}\label{eq:fhgxipg}
  \end{align}
where 
\begin{align}
P_{\sysC^L}^g &=\sum_{(y_1,\ldots,y_L): y_1\cdots y_L=g} \proj{y_1,\ldots,y_L}
\end{align}
is a projection
and
$CU^h_{\sysC^L\rightarrow \sysB^L}$ is the unitary introduced in~\eqref{eq:Cuhdefinition}.

\begin{lemma}\label{lem:probabilisticimplementability}
Let $G$ be a solvable group.
Let $(C,\pi)$ be fixed. Also fix ${\bf v}=(i',j')\in \{1,\ldots,|C|\}\times \{1,\ldots,d_\pi\}$. 
Then we can in constant adaptive depth probabilistically map an arbitrary state $\Psi$ to a state proportional to 
\begin{align}
F_\xi^{(C,\pi);({\bf u},{\bf v})}\Psi\ ,
\end{align} where the distribution over~${\bf u}\in \{1,\ldots,|C|\}\times \{1,\ldots,d_\pi\}$
is determined by~$\Psi$. 
\end{lemma}

\begin{proof}
  Consider the ooperaotr $F_\xi^{(C,\pi);({\bf u},{\bf v})}$ with ${\bf v}=(1,1)$, i.e., (cf.~\eqref{eq:generalribbonoperatoranyonic})
\begin{align}
	F_\xi^{(C,\pi);((i,j),(1,1))}&:=\frac{d_\pi}{|E(C)|}\sum_{k\in E(C)}(\Gamma_{\pi}^{-1}(k))_{j,1} F_\xi^{(c_i^{-1},p_i k)}\ .
\end{align}
With~\eqref{eq:fhgxipg} we can rewrite this as
\begin{align}
  F_\xi^{(C,\pi);((i,j),(1,1))}&=
  G_{\sysC^L}^{(i,j)}
  CU^{c_i^{-1}}_{\sysC^L\rightarrow\sysB^L}
  \qquad\textrm{ where }\qquad G_{\sysC^L}^{(i,j)}:=\frac{d_\pi}{|E(C)|}\sum_{k\in E(C)}(\Gamma_{\pi}^{-1}(k))_{j,1} P_{\sysC^L}^{p_ik}\ .
\end{align}
The unitary~$CU^{c_i^{-1}}_{\sysC^L\rightarrow\sysB^L}$ can be realized by using an auxiliary system~$\sysA$ initialized in the state~$\ket{c_i^{-1}}_\sysA$ and applying the controlled version $CCU_{\sysA\sysC^L\rightarrow\sysB^L}$. As shown in Lemma~\ref{lem:constandepthimplementationccu} the unitary $CCU_{\sysA\sysC^L\rightarrow \sysB^L}$ can be implemented in adaptive constant depth.

Now consider the operators~$G_{\sysC^L}^{(i,j)}$. Let us first consider a set of single-qudit operator of a similar form, namely
\begin{align}
  \hat{G}^{(i,j)}:=\frac{d_\pi}{|E(C)|}\sum_{k\in E(C)} (\Gamma_{\pi}^{-1}(k))_{j,1} \proj{p_ik}
  \end{align}
By definition, every $y\in G$ can be written as $y=p_i k$ for a unique pair $(i,k)\in \{1,\ldots,|C|\}\times E(C)$.  Let us write $(i(y),k(y))$ for this pair and define an isomorphism $V:\mathbb{C}^{|G|}\rightarrow \mathbb{C}^{|C|}\otimes\mathbb{C}^{|E(C)|}$ acting as $V\ket{y}=\ket{i(y)}\ket{k(y)}$.  This implies that
\begin{align}
V \hat{G}^{(i,j)}V^\dagger &=\proj{i}\otimes  \sum_{k\in E(C)}\Gamma^{-1}_\pi(k)_{j,1}\proj{k}\ .\label{eq:gijdecomposition}
\end{align}
Let $J$ be an auxiliary system with  Hilbert space~$\mathbb{C}^{d_\pi}$  initially prepared in the state~$\ket{1}$. The irreducible unitary representation $\pi$ acts on~$J$ by the matrices $\{\Gamma_\pi(k)\}_{k\in E(C)}$. Let
\begin{align}
\Pi_{KJ}:=\sum_{k\in E(C)} \proj{k}_K\otimes (\Gamma^{-1}_\pi(k))_J
\end{align}
be a corresponding controlled unitary.  Then
\begin{align}
 \Gamma^{-1}_\pi(k)_{j,1}\proj{k}&=(I_K\otimes \bra{j}_J) \Pi_{KJ}(I_K\otimes\ket{1}_J)\label{eq:gammapiijdecompotwo}
\end{align}
for any $j\in \{1,\ldots,d_\pi\}$.
Eq.~\eqref{eq:gijdecomposition} and~\eqref{eq:gammapiijdecompotwo}
show that one element of $\{V\hat{G}^{(i,j)}V^\dagger\}_{(i,j)}$ can be applied probabilistically initializing the system~$J$  in the state~$\ket{1}$, applying $\Pi_{KJ}$, and subsequently measuring the systems $I$ and $J$ in the computational basis obtaining outcome~$(i,j)$.  Since~$V$ is an isometry, this immediately implies that one element of $\{\hat{G}^{(i,j)}\}_{(i,j)}$ can be applied probabilistically, see Fig.~\ref{fig:probabilisticapplicationG} for the corresponding circuit.
\begin{figure}
  \centering
  \[
  \begin{quantikz}[transparent]
\lstick{\ket{y}}&\qw&\gate[2,nwires={2}]{V} &\qw             &\qw      & \qw          &\qw & \meter{} & \push{\ i\ }   &\qw &\gate[2]{V^\dagger}    &\qw \\
                &   &                       &\qw             &\qw      & \gate[2]{\Pi}  &\qw & \qw      & \qw                &\qw &                   &\\
                &   &                       &\lstick{$\ket{1}_J$}&\qw  &              &\qw & \meter{}      & \push{\ j\ }     &\qw & \qw               &
  \end{quantikz}
  \]
  \caption{Circuit implementing one element of the family $\{\hat{G}^{(i,j)}\}_{(i,j)}$ at random\label{fig:probabilisticapplicationG}}
\end{figure}
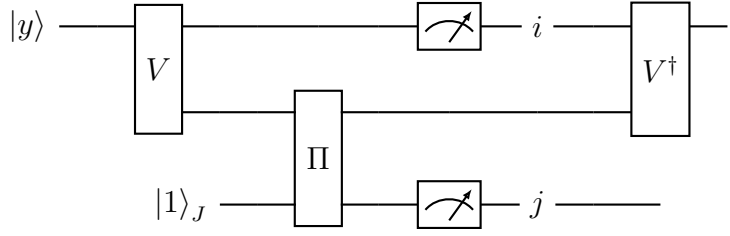
Finally, an element of the family~$\{G_{\sysC^L}^{(i,j)}\}_{(i,j)}$ can be applied by
coherently computing the product $y_1\cdots y_L$ into an auxiliary register~$\sysA_0$ using the group multiplication gate~$\multgate_{\sysA^L\rightarrow\sysA_0}$, probabilistically applying~$\{\hat{G}^{(i,j)}\}_{(i,j)}$, and uncomputing. This is because
\begin{align}
  (\bra{\idG}_{\sysC_{L+1}}\otimes I_{\sysC^L})
\multgate_{\sysC^L\rightarrow\sysC_{L+1}}^\dagger \hat{G}^{(i,j)}_{\sysC_{L+1}}\multgate_{\sysC^L\rightarrow\sysC_{L+1}}
  (\ket{\idG}_{\sysC_{L+1}}\otimes I_{\sysC^L})&=G^{(i,j)}_{\sysC^L}\ .
\end{align}
See Fig.~\ref{fig:figrefprobapplication} for an illustration of the corresponding circuit.
\begin{figure}
    \[
  \begin{quantikz}[transparent]
    \lstick{$\ket{g_1}$} & \qw &\gate[4]{\multgate_{\sysC^L\rightarrow\sysC_{L+1}}}   &\qw &\qw &\qw&\gate[4]{\multgate_{\sysC^L\rightarrow\sysC_{L+1}}^\dagger}&\qw\\
    \lstick{$\ket{g_2}$} & \qw &            &\qw& \qw&\qw&\qw&\qw\\
    \lstick{$\ket{g_3}$} & \qw &            &\qw &\qw&\qw&\qw&\qw\\
        \lstick{$\ket{\idG}$}& \qw &         &\qw &\gate{\{\hat{G}^{(i,j)}\}_{(i,j)}}&\qw&   &\qw &\rstick{$\ket{\idG}$}\\
  \end{quantikz}
  \]
  \caption{Circuit implementing one element of the family $\{G^{(i,j)}_{\sysC^L}\}_{(i,j)}$ at random, for $L=3$.\label{fig:figrefprobapplication}}
\end{figure}
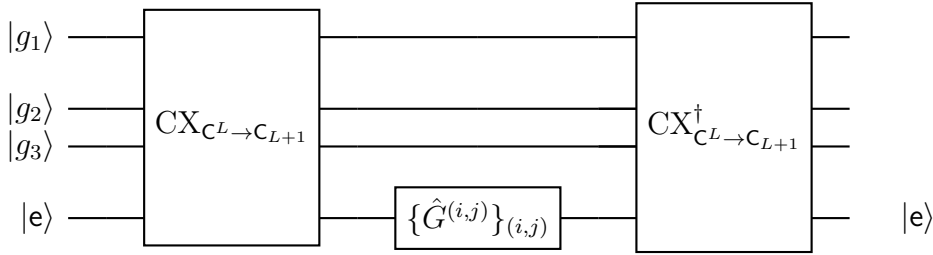
\end{proof}

\section{Adaptive implementation of charge measurements\label{sec:implementingtopologicalchargemeasurements}}
Let $\sigma$ be a closed ribbon. Recall (see Section~\ref{sec:topologicalchargeprojection}) that
the family~$\{K_\sigma^{(C,\pi)}\}_{(C,\pi)}$ of operators (cf. Eq.~\eqref{eq:ksigmarc}) defines a measurement consisting of orthogonal projections. Here we show the following:
\begin{theorem}
Suppose $G$ is a solvable group. Then the POVM~$\{K_\sigma^{(C,\pi)}\}_{(C,\pi)}$
can be realized by an adaptive constant-depth circuit with local gates and measurement. The circuit is supported on qudits along~$\sigma$.
\end{theorem}
\begin{proof}
Using the definitions~\eqref{eq:ksigmadc} and~\eqref{eq:ksigmarc}, the fact that
$\chi_\pi(D)=\chi_\pi(d)$ is the character of~$\pi$ for any representative~$d\in D$,
and that $E(C)$ is the disjoint union of its conjugacy classes, we have
  \begin{align}
    K_\sigma^{(C,\pi)}&=\frac{1}{|E(C)|}\sum_{d\in E(C)} \overline{\chi_\pi(d)} 
 \sum_{j=1}^{|C|} F_\sigma^{p_jd p_j^{-1},p_jr_Cp_j^{-1}}\ . \label{eq:rewrittenksigmarc}
  \end{align}
 Inserting  the factorization~\eqref{eq:fhgxipg} of a ribbon operator $F_\sigma^{h,g}$,
 this becomes 
 \begin{align}
    K_\sigma^{(C,\pi)}&=\frac{1}{|E(C)|}\sum_{d\in E(C)} \overline{\chi_\pi(d)} 
    \sum_{j=1}^{|C|}
    P^{p_jr_Cp_j^{-1}}_{\sysC^L}CU^{p_jd p_j^{-1}}_{\sysC^L\rightarrow\sysB^L}\ .
 \end{align}
 Using an auxiliary system~$\sysC_0$ in the state~$\ket{\idG}$, we can express this as
 \begin{align}
K_\sigma^{(C,\pi)}&=(\bra{\idG}\otimes I_{\sysB^L\sysC^L}) \multgate_{\sysC^L\rightarrow\sysC_0}^\dagger \hat{K}_\sigma^{(C,\pi)}\multgate_{\sysC^L\rightarrow\sysC_0}
   (\ket{\idG}\otimes I_{\sysB^L\sysC^L})
 \end{align}
 where
 \begin{align}
   \hat{K}_\sigma^{(C,\pi)}
   &:=\frac{1}{|E(C)|}\sum_{d\in E(C)} \overline{\chi_\pi(d)} 
    \sum_{j=1}^{|C|}
    \proj{p_jr_Cp_j^{-1}}_{\sysC_0}\otimes CU^{p_jd p_j^{-1}}_{\sysC^L\rightarrow\sysB^L}\ .
 \end{align}
 In particular, because the gate~$\multgate_{\sysC^L\rightarrow\sysC_0}$ and its inverse can be implemented in adaptive (local) constant depth,  it suffices
 to demonstrate that
 there is a constant-depth adaptive local circuit implementing the measurement~$\{\hat{K}_\sigma^{(C,\pi)}\}_{(C,\pi)}$.

 For every $y\in G$, there is a unique pair $(C(y),j(y))$ with $C(y)\in G^{cj}$ and
 $j(y)\in \{1,\ldots,|C(y)|\}$  such that $y\in C(y)$ and $y=p_{j(y)}r_{C(y)}p_{j(y)}^{-1}$.
 Let $c_{\max}=\max_{C\in G^{cj}}|C|$ be the maximal size of a conjugacy class~$C$ of $G$.
 Define an isometric map $W:\mathbb{C}^{|G|}\rightarrow\mathbb{C}^{|G^{cj}|}\times \mathbb{C}^{c_{\max}}$ by
 \begin{align}
W\ket{y}&= \ket{C(y)}\otimes\ket{j(y)}\qquad\textrm{ for every }\qquad y\in G\ .
 \end{align}
 Then
 \begin{align}
   W_{\sysC_0} \hat{K}^{(C,\pi)}_\sigma W_{\sysC_0}^\dagger &=
   \proj{C}\otimes G^C_\pi\label{eq:wckwc}
 \end{align}
 where
 \begin{align}
   G^C_\pi:=\frac{1}{|E(C)|}\sum_{j=1}^{|C|}\sum_{d\in E(C)} \overline{\chi_\pi(d)} 
\proj{j}_J\otimes CU^{p_jd p_j^{-1}}_{\sysC^L\rightarrow\sysB^L}\ .   
 \end{align}
 Eq.~\eqref{eq:wckwc} implies that the
 measurement~$\{W_{\sysC_0} \hat{K}^{(C,\pi)}_\sigma W_{\sysC_0}^\dagger\}_{(C,\pi)}$
 can be implemented by a two-step procedure: First a von Neumann measurement
 is performed on the first system to determine the conjugacy class~$C$.
 Subsequently, a measurement~$\{G^{C}_\pi\}_\pi$ is applied.
 It is straghtforward to check that this latter measurement consists of orthogonal projections. Indeed, this follows from Schur's orthogonality relation for
 characters and the fact that 
 $CU^{p_jd_1 p_j^{-1}}_{\sysC^L\rightarrow\sysB^L}CU^{p_jd_2 p_j^{-1}}_{\sysC^L\rightarrow\sysB^L}=CU^{p_j(d_1d_2)p_j^{-1}}_{\sysC^L\rightarrow\sysB^L}$
  (cf.~\eqref{eq:stringmultiplication}).

 Let now~$C$ be a fixed conjugacy class. It remains to argue  how to implement~$\{G^{C}_{\pi}\}_\pi$.
 % Let $D$ be an auxiliary system isomorphic to~$\mathbb{C}^{|E(C)|}$.
Let $\sysA\cong\mathbb{C}^{|G|}$ be an auxiliary system initialized in the state~$\ket{+}:=\frac{1}{|E(C)|^{1/2}}\sum_{d\in E(C)}\ket{d}$.
 For every irreducible representation~$\pi$ of $E(C)$ define the state
\begin{align}
\ket{\varphi^C_\pi}_\sysA:=\frac{1}{|E(C)|^{1/2}}\sum_{d\in E(C)}\chi_\pi(d)\ket{d}_\sysA\ .
\end{align}
By Schur's orthogonality relation for characters, the vectors~$\{\ket{\varphi^C_\pi}\}_{\pi}$ are orthonormal.

Let $E:\mathbb{C}^{|C|}\otimes \mathbb{C}^{|G|}\rightarrow \mathbb{C}^{|C|}\otimes \mathbb{C}^{|G|}$
be the unitary
that acts as
\begin{align}
E_{J\rightarrow\sysA}(\ket{j}_J\otimes \ket{d}_{\sysA})&=\ket{j}_J\otimes \ket{p_j d p_j^{-1}}_{\sysA}\ .
  \end{align}
Then
\begin{align}
  E_{J\rightarrow\sysA}(I_J\otimes \ket{+}_\sysA)
  &=\frac{1}{|E(C)|^{1/2}}\sum_{d\in E(C)} \sum_{j=1}^{|C|} \proj{j}_J \otimes
  \ket{p_j d p_j^{-1}}_{\sysA}
\end{align}
and thus
\begin{align}
  E_{J\rightarrow\sysA}^\dagger CCU_{\sysA\sysC^L\rightarrow\sysB^L} E_{J\rightarrow\sysA}(I_{J\sysB^L\sysC^L}\otimes \ket{+}_\sysA)
  &=\frac{1}{|E(C)|^{1/2}}\sum_{j=1}^{|C|}\sum_{d\in E(C)}\proj{j}_J\otimes 
CU^{p_jdp_j^{-1}}_{\sysC^L\rightarrow\sysB^L}\otimes\ket{d}_{\sysA}\ .
\end{align}
In particular, this implies that for any state~$\Psi_{J\sysB^L\sysC^L}$ we have
\begin{align}
  (I_{J\sysB^L\sysC^L}\otimes \bra{\varphi^C_\pi}_\sysA)  E_{J\rightarrow\sysA}^\dagger CCU_{\sysA\sysC^L\rightarrow\sysB^L} E_{J\rightarrow\sysA}(\ket{\Psi}_{J\sysB^L\sysC^L}\otimes \ket{+}_\sysA)
  &=G^{C}_\pi\ket{\Psi}\ .\label{eq:ijljm}
\end{align}
Using the fact that each $G^{C}_\pi$ is an orthogonal projection it follows from~\eqref{eq:ijljm} that
the statistics obtained by applying the projective measurement~$\{G^C_{\pi}\}_\pi$ to a state~$\Psi$
coincides to that obtained by applying the von Neumann measurement defined by~$\{\ket{\varphi^C_\pi}\}_\pi$
to the state~$E_{J\rightarrow\sysA}^\dagger CCU_{\sysA\sysC^L\rightarrow\sysB^L} E_{J\rightarrow\sysA}(\ket{\Psi}_{J\sysB^L\sysC^L}\otimes \ket{+}_\sysA)$.

It is straightforward to express this procedure as a constant-depth adaptive circuit.
\end{proof}

\textbf{Acknowledgments.} We thank  Nat Tantivasadakarn,
Ashvin Vishwanath, and Ruben Verresen 
for pointing out an error in an earlier version of this manuscript, and providing a helpful example elucidating the issue. In more detail,  our arguments in an earlier version of our paper relied on the additional  implicit assumption that the order of the considered group element~$a$ in Section~\ref{sec:statepreparation} matches that of the quotient group~$G/H$. This revised version lifts this assumption. RK thanks Amanda Young for discussions. IK thanks Tim Hsieh for discussions.  SB, AK and RK acknowledge support by the Army Research Office
under Grant Number W911NF-20-1-0014. The views and conclusions contained in
this document are those of the authors and should not be interpreted as representing
the official policies, either expressed or implied, of the Army Research Office or the
U.S. Government. The U.S. Government is authorized to reproduce and distribute
reprints for Government purposes notwithstanding any copyright notation herein. RK gratefully acknowledges support by the European Research Council under grant agreement no.~101001976 (project EQUIPTNT).

\textbf{Declarations.} Data sharing is not applicable to this article as no datasets were generated or analysed during the current study. The authors have no competing interests to declare that are relevant to the content of this article.

\bibliographystyle{unsrt}
\bibliography{q}

\appendix

\section{Orthogonality of the charge projections~\label{app:checkingorthogonality}}
Here we use expression~\eqref{eq:rewrittenksigmarc} to check that the operators
~$\{K_{\sigma}^{(C,\pi)}\}_{(C,\pi)}$ are indeed orthogonal projections. By definition, we have 
\begin{align}
K_\sigma^{(C,\pi)}K_\sigma^{(C',\pi')}&=\frac{1}{|E(C)|\cdot |E(C')|}\sum_{\substack{d\in E(C)\\
d'\in E(C')}} \overline{\chi_\pi(d)}\cdot \overline{\chi_{\pi'}(d')}
\sum_{\substack{j=1,\ldots,|C|\\
k=1,\ldots,|C'|}} F_\sigma^{p_j d p_j^{-1},p_j r_C p_j^{-1}} F_\sigma^{q_k d'q_k^{-1},q_k r_{C'}q_{k}^{-1}}\label{eq:Krcsigmakrcpsigmaprime}
\end{align}
where the conjugacy classes are given by $C=\{p_j r_C p_j\}_{j=1}^{|C|}$ and $C'=\{q_k r_{C'} q_k\}_{k=1}^{|C'|}$, respectively. With the multiplication rule~\eqref{eq:stringmultiplication} for basic ribbon operators we obtain
\begin{align}
F_\sigma^{p_j d p_j^{-1},p_j r_C p_j^{-1}}F_\sigma^{q_k d'q_k^{-1},q_k r_{C'}q_{k}^{-1}}&= \delta_{p_j r_C p_j^{-1},q_k r_{C'}q_k^{-1}} F_\sigma^{(p_j d p_j^{-1})(q_k d'q_k^{-1}), p_j r_C p_j^{-1}}\\
&=\delta_{C,C'}\delta_{p_j,q_k} F_\sigma^{p_j dd' p_j^{-1}, p_j r_C p_j^{-1}}\ ,
\end{align}
where we used that $p_j r_C p_j^{-1}=q_k r_{C'}q_k^{-1}$ if and only if these elements belong to the same conjugacy class $C=C'$, and $r_j=q_k$.  Inserting this into~\eqref{eq:Krcsigmakrcpsigmaprime}
yields
\begin{align}
K_\sigma^{(C,\pi)}K_\sigma^{(C',\pi')}&=\frac{\delta_{C,C'}}{|E(C)|^2}
\sum_{d,d'\in E(C)}\overline{\chi_\pi(d)}\cdot \overline{\chi_{\pi'}(d')}
\sum_{j=1}^{|C|}F^{p_j dd' p_j^{-1},p_jr_C p_j^{-1}}_\sigma\ .
\end{align}
Let us reparametrize this by setting $g:=dd'$, i.e., $d'=d^{-1}g$. Then we can write 
\begin{align}
K_\sigma^{(C,\pi)}K_\sigma^{(C',\pi')}&=\frac{\delta_{C,C'}}{|E(C)|^2}
\sum_{d,g\in E(C)}\overline{\chi_\pi(d)}\cdot \overline{\chi_{\pi'}(d^{-1}g)}
\sum_{j=1}^{|C|}F^{p_j g p_j^{-1},p_jr_C p_j^{-1}}_\sigma\ .\label{eq:ksigmarck}
\end{align}
We claim that
\begin{align}
\frac{1}{|E(C)|}\sum_{d\in E(C)} \overline{\chi_\pi(d)}\cdot \overline{\chi_{\pi'}(d^{-1}g)} &=\delta_{\pi,\pi'}\overline{\chi_\pi(g)}\qquad\textrm{ for every }\qquad g\in G\ .\label{eq:toproveec}
\end{align}
Inserting~\eqref{eq:toproveec} into~\eqref{eq:ksigmarck} gives 
\begin{align}
K_\sigma^{(C,\pi)}K_\sigma^{(C',\pi')}&=\frac{\delta_{C,C'}\delta_{\pi,\pi'}}{|E(C)|}
\sum_{g\in E(C)}\overline{\chi_\pi(g)}\sum_{j=1}^{|C|}F^{p_j g p_j^{-1},p_jr_C p_j^{-1}}_\sigma\\
&=\delta_{C,C'}\delta_{\pi,\pi'}K_\sigma^{(C,\pi)}\ ,
\end{align}
see~\eqref{eq:rewrittenksigmarc}, as claimed.

It remains to show~\eqref{eq:toproveec}. We have for every $g\in G$
\begin{align}
\frac{1}{|E(C)|}\sum_{d\in E(C)} \chi_\pi(d)\cdot \chi_{\pi'}(d^{-1}g) &=
\frac{1}{|E(C)|}\sum_{d\in E(C)}\sum_{k=1}^{\dim \pi}\sum_{\ell=1}^{\dim \pi'} \Gamma_\pi(d)_{k,k}
\Gamma_{\pi'}(d^{-1}g)_{\ell,\ell}\\
&=\frac{1}{|E(C)|}\sum_{d\in E(C)}\sum_{k=1}^{\dim \pi}\sum_{\ell,m=1}^{\dim \pi'} \Gamma_\pi(d)_{k,k}
\Gamma_{\pi'}(d^{-1})_{\ell,m} \Gamma_{\pi'}(g)_{m,\ell}\\
&=\frac{1}{|E(C)|}\sum_{d\in E(C)}\sum_{k=1}^{\dim \pi}\sum_{\ell,m=1}^{\dim \pi'} \Gamma_\pi(d)_{k,k}
\overline{\Gamma_{\pi'}(d)}_{m,\ell} \Gamma_{\pi'}(g)_{m,\ell}\\
\end{align}
where we used the unitarity of the representation~$\pi'$. The Schur orthogonality relations for matrix elements of irreducible representations state that
\begin{align}
\frac{1}{|E(C)|}\sum_{d\in E(C)} \Gamma_{\pi}(d)_{k,\ell}\overline{\Gamma_{\pi'}(d)_{s,t}}&=\delta_{\pi,\pi'}\delta_{k,s}\delta_{\ell,t}\ ,
\end{align}
which implies that
\begin{align}
\frac{1}{|E(C)|}\sum_{d\in E(C)} \chi_\pi(d)\cdot \chi_{\pi'}(d^{-1}g) &=
\delta_{\pi,\pi'}\cdot \sum_{k=1}^{\dim \pi}\Gamma_{\pi'}(g)_{k,k}
\\
&=\delta_{\pi,\pi'}\chi_\pi(g)\ .
\end{align}
Taking the complex conjugate implies the claim~\eqref{eq:toproveec}.

\section{On unitary linear combinations of anyonic string operators\label{app:unitarycombinations}}

Let $\xi$ be a ribbon. Fix $(C,\pi)$. In this appendix, we show that 
the operators~$\{F_\xi^{(C,\pi);((i,j),(i',j'))}\}_{(i,j),(i',j')}$ cannot  be linearly combined to form a unitary in general.
More precisely, we show the following:
\begin{lemma}\label{lem:linearcombinationunitary}
Let $\xi$ be a ribbon. 
\begin{enumerate}[(i)]
\item\label{it:firstitemunitarycombination}
If $G$ is abelian, then for any pair $(C,\pi)$ and any $i,i'\in \{1,\ldots,\dim(\pi)\}$, there are coefficients~$\{m_{j,j'}\}_{j,j'}$ such that  the linear combination
\begin{align}
F_\xi&:=\sum_{j,j'} m_{j,j'}F_\xi^{(C,\pi);((i,j),(i',j'))}
\end{align}
is unitary.
\item\label{it:seconditemunitarycombination}
For $G=S_3$, the conjugacy class~$C=\{\idG\}$,
and the $2$-dimensional irreducible representation~$\pi$ of~$S_3$, there is no choice of coefficients~$\{m_{i,i',j,j'}\}_{i,i',j,j'}$ such that 
\begin{align}
F_\xi &:=\sum_{i,i',j,j'} m_{i,i',j,j'}F_\xi^{(C,\pi);((i,j),(i',j'))}\label{eq:generalunitaryoplinearcomb}
\end{align}
is unitary.
\end{enumerate}
\end{lemma}

\begin{proof}
To prove Lemma~\ref{lem:linearcombinationunitary},
we will argue that a necessary and sufficient condition for~\eqref{eq:generalunitaryoplinearcomb} to be unitary
is that there is some $i_0\in \{1,\ldots,\dim(\pi)\}$  such that
\begin{align}
 m_{i,i',j,j'}&=\delta_{i,i_0}m_{j,j'}\textrm{ for some coefficients }\{m_{j,j'}\}_{j,j'}\textrm{ satisfying }\\
 \left|\sum\limits_{\substack{j, j'}}
 m_{j,j'}\left( \Gamma_\pi  \left( k \right) \right)_{j, j'} \right|&=\frac{|E(C)|}{\dim(\pi)}\qquad\textrm{ for all }\qquad k\in E(C)\ .\label{eq:modifiedconditionsimplified}
\end{align}
If $G$ and thus~$E(C)$ is abelian, then every irreducible representation is $1$-dimensional, and taking~$m_{j,j'}=\textfrac{|E(C)|}{\dim(\pi)}$, property~\eqref{eq:modifiedconditionsimplified} is satisfied, giving the first  claim~\eqref{it:firstitemunitarycombination}.

Let us now show~\eqref{it:seconditemunitarycombination}:  Consider $G=S_3$ and the conjugacy class~$C=\{\idG\}$. Here $E(C)=S_3$. We again consider the $2$-dimensional irreducible representation~$\pi$ of~$S_3$ obtained 
from the representation of~$S_3$ on~$\mathbb{C}^3$ which permutes standard basis vectors as
\begin{align}
\pi(k) e_j&=e_{k^{-1}(j)}\qquad\textrm{ for }\qquad j=1,\ldots,3\qquad\textrm{ for all }k\in S_3\ , \label{eq:basicrepresentationrestricted}
\end{align}
by restricting  to the subspace 
  \begin{align}
V:= \left\{
    \begin{pmatrix}
      a\\
      b\\
      c
      \end{pmatrix}\ \Big|\ a+b+c=0
    \right\}\ .
  \end{align} First observe that the question of  existence of  coefficients~$\{m_{j,j'}\}_{j,j'}$ satisfying~\eqref{eq:modifiedconditionsimplified} is independent of the basis of~$V$ chosen: It can be phrased as the existence of a $2\times 2$-matrix~$M$ such that
  \begin{align}
  \left|\tr\left(M\Gamma_\pi(k)\right)\right|=3\qquad\textrm{ for all }\qquad k\in E(C)\ .\label{eq:cexistencem}
  \end{align}
 By the cyclicity of the trace,
 for any invertible matrix~$W$ corresponding to a basis change, the existence of~$M$ satisfying~\eqref{eq:cexistencem} is equivalent to existence of a $2\times 2$-matrix~$\tilde{M}$ such that
 \begin{align}
  \left|\tr\left(\tilde{M}W\Gamma_\pi(k)W^{-1}\right)\right|=3\qquad\textrm{ for all }\qquad k\in E(C)\ .
 \end{align}
 This shows the basis independence of the question of interest.    We again use the orthonormal basis
  \begin{align}
    e_1=\begin{pmatrix}
    \sqrt{2/3}\\
    -1/\sqrt{6}\\
    -1/\sqrt{6}
    \end{pmatrix}\qquad\qquad\qquad
    e_2&=\begin{pmatrix}
    0\\
    1/\sqrt{2}\\
    -1/\sqrt{2}
    \end{pmatrix}\ .
    \end{align}    
    Then, expressed in this basis, we have
    \begin{align}
      \Gamma_\pi\left(
      \idG
      \right)&=\begin{pmatrix}
      1 & 0\\
      0 & 1
      \end{pmatrix}\\
      \Gamma_\pi\left(\begin{pmatrix}
      1 & 2 & 3\\
      1 & 3 & 2
      \end{pmatrix}
            \right)&=\begin{pmatrix}
      1 & 0\\
      0 & -1
      \end{pmatrix}\\
       \Gamma_\pi\left(\begin{pmatrix}
      1 & 2 & 3\\
      2 & 1 & 3
      \end{pmatrix}
            \right)&=\frac{1}{2}\begin{pmatrix}
      -1 & \sqrt{3}\\
      \sqrt{3} & 1
      \end{pmatrix}\\
        \Gamma_\pi\left(\begin{pmatrix}
      1 & 2 & 3\\
       3 & 1 & 2
      \end{pmatrix}
            \right)&=\frac{1}{2}\begin{pmatrix}
      -1 & -\sqrt{3}\\
      \sqrt{3} & -1
      \end{pmatrix}\\
        \Gamma_\pi\left(\begin{pmatrix}
          1 & 2 & 3\\
      2 & 3 & 1
      \end{pmatrix}
            \right)&=\frac{1}{2}\begin{pmatrix}
      -1 & \sqrt{3}\\
      -\sqrt{3} & -1
      \end{pmatrix}\\
       \Gamma_\pi\left(\begin{pmatrix}
      1 & 2 & 3\\
      3 & 2 & 1
      \end{pmatrix}
            \right)&=\frac{1}{2}\begin{pmatrix}
      -1 & -\sqrt{3}\\
      -\sqrt{3} & 1
      \end{pmatrix}\ .
    \end{align}
    It is easy to check from these expressions that a matrix~$M\in\mathsf{Mat}_{2\times 2}(\mathbb{C})$ satisfying~\eqref{eq:cexistencem} does not exist. This gives the claim~\eqref{it:seconditemunitarycombination}.
    
   It remains to verify that~\eqref{eq:modifiedconditionsimplified} is indeed a necessary and sufficient condition for a unitary linear combination~$F_\xi$ of the form~\eqref{eq:generalunitaryoplinearcomb} to exist. From the multiplication rule~\eqref{eq:stringmultiplication}
   of elementary ribbon operators,  the definition of~$F_\xi^{(C,\pi);((i,j),(i',j'))}$ and the action of the elementary ribbon operator~$F_\xi^{h,g}$ (by conjugation with $h$), it is clear that $F_\xi$ can only be unitary if $m_{i,i',j,j'}=\delta_{i,i_0}m_{i,i',j,j'}$
   for some fixed $i_0\in \{1,\ldots,\dim(\pi)\}$. In other words, it suffices to consider operators of the form
   \begin{align}
	F_\xi := \sum\limits_{i', j, j'} m_{i',j,j'} F_\xi^{(C,\pi);((i,j),(i',j'))}
\end{align}
for some complex coefficients~$\{m_{i',j,j'}\}_{i',j,j'}$. We have
\begin{align}
	F_\xi^\dagger F_\xi &= \frac{\dim \left( \pi \right)^2}{\left| E(C) \right|^2} \sum\limits_{i_1', j_1, j_1'} \sum\limits_{i_2', j_2, j_2'} \overline{m_{i_1',j_1,j_1'}} c_{i_2'j_2j_2'} \sum\limits_{k_1, k_2 \in E(C)} \overline{\left( \Gamma_\pi^{-1}  \left( k_1 \right) \right)_{j_1, j_1'}} \left( \Gamma_\pi^{-1}  \left( k_2 \right) \right)_{j_2, j_2'} F_\xi^{(c_i, p_i k_1 p_{i_1'}^{-1})} F_\xi^{(c_i^{-1}, p_i k_2 p_{i_2'}^{-1})} \\ 
	&= \frac{\dim \left( \pi \right)^2}{\left| E(C) \right|^2} \sum\limits_{i_1', j_1, j_1'} \sum\limits_{i_2', j_2, j_2'} \overline{m_{i_1',j_1,j_1'}} m_{i_2',j_2,j_2'} \sum\limits_{k_1, k_2 \in E(C)} \overline{\left( \Gamma_\pi^{-1}  \left( k_1 \right) \right)_{j_1, j_1'}} \left( \Gamma_\pi^{-1}  \left( k_2 \right) \right)_{j_2, j_2'} \delta_{p_i k_1 p_{i_1'}^{-1}, p_i k_2 p_{i_2'}^{-1}} F_\xi^{(\idG, p_i k_1 p_{i_1'}^{-1})} \\ 
	&= \frac{\dim \left( \pi \right)^2}{\left| E(C) \right|^2} \sum\limits_{i_1', j_1, j_1'} \sum\limits_{i_2', j_2, j_2'} \overline{m_{i_1',j_1,j_1'}} m_{i_2',j_2,j_2'} \sum\limits_{k \in E(C)} \overline{\left( \Gamma_\pi^{-1}  \left( k \right) \right)_{j_1, j_1'}} \left( \Gamma_\pi^{-1}  \left( k p_{i_1'}^{-1} p_{i_2'} \right) \right)_{j_2, j_2'} \delta_{kp_{i'_1}^{-1}p_{i'_2}\in E(C)}
	F_\xi^{(\idG, p_i k p_{i_1'}^{-1})} \ .
\end{align}
Here we used the multiplication rule~\eqref{eq:stringmultiplication} of the elementary ribbon operators and performed a variable substitution. We note  the following equivalences:
\begin{align}
	 k p_{i_1'}^{-1} p_{i_2'} \in E(C) \qquad   &\Leftrightarrow\qquad  k p_{i_1'}^{-1} \underbrace{p_{i_2'} r_C p_{i_2'}^{-1}}_{=c_{i_2'}}  p_{i_1'} k^{-1} = r_C \\
	 & \Leftrightarrow \qquad p_{i_1'}^{-1} c_{i_2'} p_{i_1'} = r_C \\ 
	 &\Leftrightarrow \qquad c_{i_2'} = p_{i_1'} r_C p_{i_1'}^{-1} = c_{i_1'} \\ 
	 &\Leftrightarrow \qquad i_1' = i_2' \ .
\end{align}
Therefore we obtain
\begin{align}
	F_\xi^\dagger F_\xi &= \frac{\dim \left( \pi \right)^2}{\left| E(C) \right|^2} \sum\limits_{\substack{i_1', j_1, j_1', \\ j_2, j_2'}} \overline{m_{i_1',j_1,j_1'}} m_{i_1',j_2,j_2'} \sum\limits_{k \in E(C)} \overline{\left( \Gamma_\pi^{-1}  \left( k \right) \right)_{j_1, j_1'}} \left( \Gamma_\pi^{-1}  \left( k \right) \right)_{j_2, j_2'} F_\xi^{(\idG, p_i k p_{i_1'}^{-1})} \\ 
	&= \frac{\dim \left( \pi \right)^2}{\left| E(C) \right|^2} \sum\limits_{k \in E(C)} \sum\limits_{i_1'}
	\left( \sum\limits_{\substack{j_1, j_1', \\ j_2, j_2'}} \overline{m_{i_1',j_1,j_1'}} m_{i_1',j_2,j_2'}  \overline{\left( \Gamma_\pi^{-1}  \left( k \right) \right)_{j_1, j_1'}} \left( \Gamma_\pi^{-1}  \left( k \right) \right)_{j_2, j_2'} \right) F_\xi^{(\idG, p_i k p_{i_1'}^{-1})} \ .
\end{align}
Using the fact that $E(C)$ is a group, we can replace~$k$ by $k^{-1}$, obtaining
\begin{align}
F_\xi^\dagger F_\xi &= \frac{\dim \left( \pi \right)^2}{\left| E(C) \right|^2} \sum\limits_{k \in E(C)} \sum\limits_{i_1'}
	\left( \sum\limits_{\substack{j_1, j_1', \\ j_2, j_2'}} \overline{m_{i_1',j_1,j_1'}} m_{i_1',j_2,j_2'}  \overline{\left( \Gamma_\pi  \left( k \right) \right)_{j_1, j_1'}} \left( \Gamma_\pi  \left( k \right) \right)_{j_2, j_2'} \right) F_\xi^{(\idG, p_i k^{-1} p_{i_1'}^{-1})} \ .
\end{align}
Recall that every group element $g\in G$ can be written as $g=p_j n$ for some $j\in \{1,\ldots,|C|\}$ and $n\in E(C)$. In particular, as $k$ ranges over~$E(C)$ and $i_1'$ ranges over $\{1,\ldots,|C|\}$, the product~$k^{-1} p_{i_1'}^{-1}$ ranges over all elements of~$G$. 
Observing that $\sum_{g\in G} \alpha_g F_\xi^{(\idG,g)}=1$ if and only if $\alpha_g=1$ for every $g\in G$ by definition of the basic ribbon operators~$\{F_\xi^{(h,g)}\}_{(h,g)}$, we conclude that
$F_\xi$ is unitary if and only if 
\begin{align} \label{eq:suffCondition}
	\sum\limits_{\substack{j_1, j_1', \\ j_2, j_2'}} \overline{m_{i_1',j_1,j_1'}} m_{i_1',j_2,j_2'}  \overline{\left( \Gamma_\pi^{-1}  \left( k \right) \right)_{j_1, j_1'}} \left( \Gamma_\pi^{-1}  \left( k \right) \right)_{j_2, j_2'} &=\frac{|E(C)|^2}{\dim(\pi)^2}	 \quad\textrm{ for every } k\in E(C)\textrm{ and } i_1'\in \{1,\ldots,\dim(\pi)\}\ .
\end{align}
In particular, this condition is satisfied if and only if 
there are coefficients $\{m_{j,j'}\}_{j,j'}$ as required in Eq.~\eqref{eq:modifiedconditionsimplified}.
\end{proof}

\section{The ground space of Kitaev's quantum double Hamiltonian\label{app:groundspacekitaev}}
In this appendix, we show that Kitaev's quantum double Hamiltonian defined on a planar connected graph~$\cL$ has a unique ground state which is a superposition of all exact $1$-forms, that is,
\begin{align}
\label{appeq:ground_state}
|\Psi(G)\ra \sim \sum_{\theta \in \calF_0}\; |\der\theta\ra.
\end{align}
as claimed (cf.~\eqref{ground_state}). 

We first recall some basic concepts and associated properties, see e.g.,~\cite[Section 3]{Cha17} or~\cite{Cui2020}. A $1$-form $\omega\in\cF_1(G)$ is called a flat $G$-connection or simply flat if for every site~$s=(v,p)$ of~$\cL$, we have 
\begin{align}
\omega(e_1)^{\sigma_1}\cdots \omega(e_m)^{\sigma_m}=\idG\ ,
\end{align}
where $(e_1,\ldots,e_m)$ are the edges of the boundary~$\partial p$ of the plaquette~$p$ traversed in counterclockwise order starting from~$v$ and where 
\begin{align}
\sigma_j&=
\begin{cases}
1\qquad &\textrm{ if the orientation of }e_j\textrm{ agrees  with the direction of traversal}\\
-1 &\textrm{ otherwise}\ .
\end{cases}\label{eq:traversalsign}
\end{align} It follows immediately from the definition that any exact $1$-form~$\omega=\der\theta$ is flat. 
We denote by $\cF_1^{\textrm{flat}}(G)$ the set of flat $G$-connections on~$\cL$.

For $\omega\in\cF_1(G)$ and
any simple closed path~$\Gamma=(e_1,\ldots,e_m)$ of edges
we can consider the product
\begin{align}
\Phi_\omega(\Gamma)=\omega(e_1)^{\sigma_1}\cdots \omega(e_m)^{\sigma_m}\ ,
\end{align} with~$\sigma_j$ defined as in~\eqref{eq:traversalsign}. 
We will need the following statement:
\begin{lemma}\label{lem:simpleclosedgraphtrivialflux}
  Let $\cL$ be a planar connected graph. Let $\omega\in\cF^{\textrm{flat}}_1(G)$ be a flat $G$-connection on~$\cL$  and $\Gamma$  a simple closed path. Then
    \begin{align}
\Phi_\omega(\Gamma)&=\idG\ .\label{eq:gammaomegaphi}
    \end{align}
\end{lemma}

\begin{proof}
  The proof follows the reasoning given, e.g., in the proof of~\cite[Lemma 3.1.3]{Cha17}: It relies on the fact that 
  for two faces~$f_1,f_2\in\cL_2$ sharing an edge,  we have
  \begin{align}
\Phi_\omega(\partial (f_1+f_2))&=\Phi_\omega(\partial f_1)\Phi_\omega(\partial f_2)\ .\label{eq:multiplicativityconditionfaces}
  \end{align}
  Here $\partial$ denotes the boundary map which respects the corresponding orientation. 
(In fact, the map $\Phi_\omega$ belongs to $\mathsf{Hom}(\pi_1(\cL),G)$, where $\pi_1(\cL)$ is the fundamental group of~$\cL$. We will not need this here.)
 This implies that the flatness condition extends from the two faces to their union.  

  In a planar connected graph~$\cL$, a cycle basis is given by the set of bounded faces of an embedding of~$\cL$, i.e., the set~$\cF$ of (bounded) faces. Decomposing a closed path~$\Gamma$ into faces, the claim~\eqref{eq:gammaomegaphi} follows from~\eqref{eq:multiplicativityconditionfaces}. 
  \end{proof}

Let $\theta\in \cF_0(G)$ be a $0$-form. For every $1$-form~$\omega$, we define the $1$-form~$\Delta_\theta\left(\omega\right)$ by
\begin{align}
\Delta_\theta\left(\omega\right)(e)=\theta(e^+) \omega(e)\theta(e^-)^{-1}\qquad\textrm{ for any }\qquad e\in E
\end{align}
where $(e^+,e^-)$ are the head and tail of~$e$. The map~$\Delta_\theta$ is  called a gauge transformation. It preserves the set of flat connections, i.e.,
\begin{align}
\Delta_\theta\left(\cF_1^{\textrm{flat}}(G)\right)\subset \cF_1^{\textrm{flat}}(G)
\end{align}
and has inverse 
\begin{align}
(\Delta_\theta)^{-1}=\Delta_{\theta^{-1}}\ .\label{eq:inversedeltathetadef}
\end{align}
 Here $\theta^{-1}$ is defined pointwise by
\begin{align}
\theta^{-1}(g)=\theta(g)^{-1}\qquad\textrm{ for all }\qquad g\in G\ .
\end{align}
We say that two flat connections~$\omega, \omega'\in\cF_1^{\textrm{flat}}(G)$ are equivalent (written~$\omega\sim\omega'$) if there is some~$\theta\in\cF_0(G)$ such that $\omega'=\Delta_{\theta}\left(\omega\right)$.   Let us denote by $[\omega]$ an equivalence class with representative~$\omega$, and by $\left[\cF_1^{\textrm{flat}}(G)\right]$ the set of equivalence classes.

These concepts are meaningful because the plaquette-constraints
\begin{align}
B_s\Psi&=\Psi\qquad\textrm{ for every  site }s
\end{align}
for a ground state~$\Psi$ amount to the statement that~$\Psi$ is a superposition of flat connections. On the other hand, the plaquette constraints
\begin{align}
A_v\Psi&=\Psi\qquad\textrm{ for every vertex }v\ 
\end{align}
amount to the statement that equivalent flat $G$-connections have identical weight in the superposition; this is because
\begin{align}
\left|\Delta_\theta\left(\omega\right)\rangle\right.&=\prod_{v}A_v^{\theta(v)}\ket{\omega}\qquad\textrm{ for any }\qquad \omega\in \cF_1(G)
\end{align}
and any $0$-form~$\theta$. In particular, this means that
for any equivalence class $[\omega]\in \left[\cF_1^{\textrm{flat}}(G)\right]$, the state
\begin{align}
\ket{\Psi_{[\omega]}}:=\sum_{\omega': \omega'\sim\omega}\ket{\omega'}\label{eq:equivalenceclassbndstate}
\end{align}
is a ground state, and the ground space is spanned by
\begin{align}
\{\ket{\Psi_{[\omega]}}\}_{[\omega]\in \left[\cF_1^{\textrm{flat}}(G)\right]}\ .
\end{align}
In particular, the dimension of the ground space is equal to the number of equivalence classes~$\left|\left[\cF_1^{\textrm{flat}}(G)\right]\right|$ of flat connections.

\begin{lemma}\label{lem:equivalenceclassesconnectedplanar}
Let $\cL$ be a planar connected graph. 
\begin{enumerate}[(i)]
\item We have 
\label{eq:dimboundplanarconnect}
\begin{align}
\left|\left[\cF_1^{\textrm{flat}}(G)\right]\right|&=1\ .
\end{align}
That is, the ground space is $1$-dimensional.
\item\label{it:exactnessomegadefinition}
Every flat $1$-form~$\omega$ is exact.
\end{enumerate}
\end{lemma}
\begin{proof}
Let~$\tilde{\idG}$ denote the constant~$1$-form defined by 
\begin{align}
\tilde{\idG}(e)=\idG\qquad\textrm{ for every }\qquad e\in E\ .
\end{align}
Clearly, we have $\tilde{\idG}\in \cF_1^{\textrm{flat}}(G)$.

Let now $\omega\in\cF_1^{\textrm{flat}}(G)$ be arbitrary. We show that $\omega\sim \tilde{\idG}$, i.e., we construct a $0$-form~$\theta$ such that 
\begin{align}
\omega&=\Delta_{\theta}\left(\tilde{\idG}\right)\ .\label{eq:omegadeltathetacond}
\end{align}
The claim then follows with~\eqref{eq:inversedeltathetadef}.

%As a finite connected graph,~$\cL$ has a spanning tree~$T$, i.e., a connected subgraph of~$\cL$ that has no cycles such that every vertex of~$\cL$ is incident to at least one edge of~$T$.

Let $r\in \cL_0$ be some fixed vertex which we call the root. For every $v\in V$, let $\Gamma_v\subseteq \cL_1$ be some fixed path on~$\cL$ connecting~$r$ and $v$. 
Then we define the $0$-form $\theta$ by
\begin{align}
\theta(v)&=\Phi_\omega(\Gamma_v)\qquad\textrm{ for all }\qquad v\in V\ .
\end{align}
Let us argue that~$\theta$ is well-defined, i.e., this definition does not depend on the choice of~$\Gamma_v$ for any~$v\in V$. Indeed, let $\Gamma_v$ and $\Gamma'_v$ be two distinct paths connecting~$r$ with~$v$. Let $\Gamma:=\Gamma_v-\Gamma'_v$ denote the path that extends from~$r$ to~$v$ along~$\Gamma_v$ and returns along~$\Gamma'_v$ traversed in the opposite direction.  Then~$\Gamma$ is closed and thus
\begin{align}
\Phi_\omega(\Gamma)&=\idG 
\end{align}
by Lemma~\ref{lem:simpleclosedgraphtrivialflux}. But
\begin{align}
\Phi_\omega(\Gamma)&=\left(\Phi_\omega(\Gamma_{v}')\right)^{-1}\Phi_\omega(\Gamma_v)\ 
\end{align}
by definition. We conclude that~$\Phi_\omega(\Gamma_v)=\Phi_\omega(\Gamma_v')$. This shows that~$\theta$ is well-defined.

It is easy to check  that $\theta$ satisfies~\eqref{eq:omegadeltathetacond}. This shows~\eqref{eq:dimboundplanarconnect}.

The claim~\eqref{it:exactnessomegadefinition} follows from~\eqref{eq:omegadeltathetacond} and the fact that~$\Delta_\theta(\tilde{e})=\der\theta$, as follows immediately from the definitions.
\end{proof}

Lemma~\ref{lem:equivalenceclassesconnectedplanar}~\eqref{eq:dimboundplanarconnect}
and expression~\eqref{eq:equivalenceclassbndstate} imply that the ground space (on a planar connected graph~$\cL$) is spanned by the vector
\begin{align}
\ket{\Psi(G)}&=\sum_{\omega\in \cF_1^{\textrm{flat}}(G)}\ket{\omega}\ ,
\end{align}
i.e., the uniform superposition of flat $1$-forms. Because every exact $1$-form is flat and Lemma~\ref{lem:equivalenceclassesconnectedplanar}~\eqref{it:exactnessomegadefinition}, the claim~\eqref{appeq:ground_state} follows.

\end{document}